\newif\ifabstract
\abstractfalse
\newif\iffull
\ifabstract \fullfalse \else \fulltrue \fi

\ifabstract
\documentclass[a4paper,cleveref,authorcolumns]{socg-lipics-v2019}
\fi
\iffull
\documentclass{article}
\fi
\usepackage[noend]{algpseudocode}
\usepackage[ruled,vlined,linesnumbered]{algorithm2e}
\usepackage{float}
\usepackage{svg}
\usepackage{colortbl}
\usepackage{diagbox}

\graphicspath{{otherImages/}{pspaceimg/}{./otherImages/}{./pspaceimg/}}

\iffull
\usepackage[utf8]{inputenc}
\usepackage[margin=1in]{geometry}
\usepackage{amsthm}
\usepackage{hyperref}
\usepackage{amsmath}

\newtheorem{theorem}{Theorem}
\newtheorem{lemma}{Lemma}
\newtheorem{corollary}{Corollary}
\newtheorem{proposition}[theorem]{Proposition}

\theoremstyle{definition}

\usepackage{amsfonts}
\usepackage{bm}
\usepackage{graphicx}
\usepackage[labelformat=simple]{subcaption}
\usepackage{float}
\usepackage{color}

\renewcommand\thesubfigure{~(\alph{subfigure})}
\fi

\newtheorem{observation}{Observation}

\renewcommand{\emph}{\textbf}

\graphicspath{{otherImages/}{pspaceimg/}{gadgetsImages/}}

\newcounter{section-preserve}
\newcounter{theorem-preserve}
\newcommand{\blank}[1]{}
\newtoks\magicAppendix
\magicAppendix={}
\newtoks\magictoks
\newif\iflater
\laterfalse
\long\def\later#1{\magictoks={#1}%
	\edef\magictodo{\noexpand\magicAppendix={\the\magicAppendix 
			\the\magictoks%
	}}
	\magictodo}
\long\def\both#1{\magictoks={#1}%
	\edef\magictodo{\noexpand\magicAppendix={\the\magicAppendix 
			\noexpand\setcounter{theorem-preserve}{\noexpand\arabic{theorem}}%
			\noexpand\setcounter{theorem}{\arabic{theorem}}%
			\noexpand\setcounter{section-preserve}{\noexpand\arabic{section}}%
			\noexpand\setcounter{section}{\arabic{section}}%
			\noexpand\let\noexpand\oldsection=\noexpand\thesection
			\noexpand\def\noexpand\thesection{\thesection}
			\noexpand\let\noexpand\oldlabel=\noexpand\label
			\noexpand\let\noexpand\label=\noexpand\blank
			\the\magictoks%
			\noexpand\setcounter{theorem}{\noexpand\arabic{theorem-preserve}}%
			\noexpand\setcounter{section}{\noexpand\arabic{section-preserve}}%
			\noexpand\let\noexpand\thesection=\noexpand\oldsection
			\noexpand\let\noexpand\label=\noexpand\oldlabel
	}}
	\magictodo
	\the\magictoks}
\def\magicappendix{\latertrue \the\magicAppendix}

\iffull
\long\def\both#1{#1}
\let\later=\both
\def\magicappendix{}
\fi

\ifabstract
\author{Hugo A.	Akitaya}{Carleton University, Canada}{hugoakitaya@gmail.com}{0000-0002-6827-2200}
{Supported by NSERC.}

\author{Erik D. Demaine}{Massachusetts Institute of Technology, USA}{edemaine@mit.edu}{0000-0003-3803-5703}{}

\author{Andrei Gonczi}{Tufts University, USA}{andrei.gonczi@tufts.edu}{0000-0002-5939-2366`}{}

\author{Dylan H. Hendrickson}{Massachusetts Institute of Technology, USA}{dylanhen@mit.edu}{}{}

\author{Adam Hesterberg}{Harvard University, USA}{achesterberg@gmail.com}{}{}

\author{Matias Korman}{Tufts University, USA}{matias.korman@tufts.edu}{}{}

\author{Oliver Korten}{Columbia University, USA}{oliver.korten@columbia.edu}{}{}

\author{Jayson Lynch}{University of Waterloo, Canada}{jayson.lynch@uwaterloo.ca}{0000-0003-0801-1671}{Supported by NSERC.}

\author{Irene Parada}{TU Eindhoven, The Netherlands}{i.m.de.parada.munoz@tue.nl}{https://orcid.org/0000-0003-3147-0083}{}

\author{Vera Sacrist\'an}{Universitat Polit\`ecnica de Catalunya, Spain}{vera.sacristan@upc.edu}{0000-0003-0203-256X}
{Partially supported by MTM2015-63791-R (MINECO/FEDER) and Gen. Cat. DGR 2017SGR1640.}

\authorrunning{H. A. Akitaya et al.}
\fi
\iffull
\author{Hugo A.	Akitaya\thanks{Carleton University, Canada}
\and
Erik D. Demaine\thanks{Massachusetts Institute of Technology, USA}
\and
Andrei Gonczi\thanks{Tufts University, USA}
\and
Dylan H. Hendrickson\footnotemark[3]
\and
Adam Hesterberg\thanks{Harvard University, USA}
\and
Matias Korman\footnotemark[4]
\and
Oliver Korten\thanks{Columbia University, USA}
\and
Jayson Lynch\thanks{University of Waterloo, Canada}
\and
Irene Parada\thanks{TU Eindhoven, The Netherlands}
\and
Vera Sacrist\'an\thanks{Universitat Polit\`ecnica de Catalunya, Spain}
}
\date{}
\fi

\title{Characterizing Universal Reconfigurability of Modular Pivoting  Robots\iffull\thanks{Research supported in part by NSERC, MTM2015-63791-R (MINECO/FEDER) and Gen. Cat. DGR 2017SGR1640.}\fi}

\newcommand{\ackn}{
This research started at the 34th Bellairs Winter Workshop on Computational Geometry in 2019. We want to thank all participants for the fruitful discussions and a stimulating environment.
}

\ifabstract
\acknowledgements{\ackn}



\Copyright{Hugo A.	Akitaya, Erik D. Demaine, Andrei Gonczi, Dylan H. Hendrickson, Adam Hesterberg, Matias Korman, Oliver Korten, Jayson Lynch, Irene Parada, Vera Sacrist\'an}
\fi

\begin{document}

\maketitle

\begin{abstract}
We give both efficient algorithms and hardness results for reconfiguring between two connected configurations 
of modules in the hexagonal grid. 
The reconfiguration moves that we consider are ``pivots'', where a hexagonal module rotates
around a vertex shared with another module. 
Following prior work on modular robots, we define two natural sets of hexagon pivoting moves of increasing power:
restricted and monkey moves. 
When we allow both moves, we present the first universal reconfiguration algorithm,
which transforms between any two connected configurations using $O(n^3)$ monkey moves.
This result strongly contrasts the analogous problem for squares, where there are rigid
examples that do not have a single pivoting move preserving connectivity.
On the other hand, if we only allow restricted moves,
we prove that the reconfiguration problem becomes PSPACE-complete.
Moreover, we show that, in contrast to hexagons, 
the reconfiguration problem for pivoting squares is PSPACE-complete regardless of the set of pivoting moves allowed. In the process, we strengthen the reduction framework of Demaine {\em et al.} [FUN'18] that we consider of independent interest.
\end{abstract} 


\section{Introduction}

Reconfiguration problems encompass a large family of problems in which we need to provide a sequence of steps to transform one object into another. In this paper we consider the problem of reconfiguring a collection of modular robots in a lattice using some prespecified moves. 
Many variants of this problem have been studied both in the robotics and in the computational geometry communities. 
In this paper we study the reconfiguration problem for (edge-)connected  configurations of hexagons and of squares.  
As necessary for many applications, we require that (edge-)connectivity is maintained at all times (fulfilling the so-called \emph{single backbone condition}~\cite{Dumitrescu-Suzuki-Yamashita-2004}). 
The moves allowed are pivots: 
a module can rotate around vertices
shared with other modules 
and at the end of a move the pivoting module must lie in a lattice cell. 
The interior of two modules can never intersect. 


Hexagons have only two types of pivoting moves, illustrated in Figure~\ref{fig:hex_moves}.
In a \emph{restricted} move a module $a$ adjacent to a module $s$ pivots around a vertex $v$ shared by $a$ and $s$ and ends the pivoting move in the other cell that has $v$ on the boundary. 
The \emph{restricted model} of pivoting only allows this move. 
In a \emph{monkey} move a module $a$ adjacent to a module $s$ starts pivoting around a vertex $v$ shared by $a$ and $s$ as in the restricted move, 
but halfway through the rotation another vertex $w$ of $a$ coincides with the vertex of a module $s'$. 
Then $a$ continues the move pivoting around $w$ in the same direction (clockwise or counterclockwise) as before until reaching a cell adjacent to~$s'$. 
The \emph{monkey model} of pivoting includes both the restricted move and the monkey move. 
Informally, the monkey move allows a module to keep pivoting in the same direction 
when a restricted move is not possible. 
Further insights into the differences between both hexagonal models are presented in Section~\ref{sec:models}.




\begin{figure}[ht]
\centering
\begin{subfigure}{.35\textwidth}
\centering
\includegraphics[page=2]{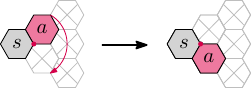}
\caption{Restricted move}
\end{subfigure}
\hfill
\begin{subfigure}{.55\textwidth}
\centering
\includegraphics[page=4]{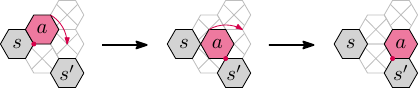}
\caption{Monkey move}
\end{subfigure}
\caption{Pivoting moves for hexagonal modules and their free-space requirements.}
\label{fig:hex_moves}
\end{figure}

{In the square grid two modules that share a vertex might not share an edge. 
Thus, for square modules there is a greater variety of pivoting moves. 
The different three sets of moves are 
illustrated in Figure~\ref{fig:square_moves}.  
The \emph{restricted model} includes only restricted moves, 
the \emph{leapfrog model} includes both restricted and leapfrog moves, and 
the \emph{monkey model} includes all moves. 
} 

\begin{figure}[ht]
\centering
\begin{minipage}{.67\textwidth}
\centering
\includegraphics[page=3]{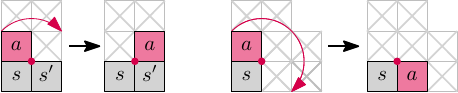}
\subcaption{Restricted moves}\label{fig:sq_res}
\end{minipage}
\hfill
\begin{minipage}{.3\textwidth}
\centering
\includegraphics[page=4]{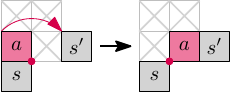}
\subcaption{Leapfrog move}\label{fig:sq_leap}
\end{minipage}
\medskip

\begin{minipage}{\textwidth}
\centering
\includegraphics[page=5]{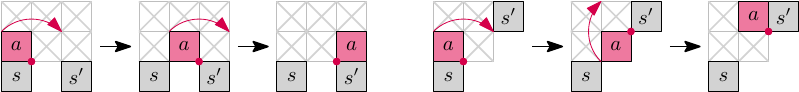}
\subcaption{Monkey moves}\label{fig:sq_mon}
\end{minipage}
\caption{Pivoting moves for square modules and their free-space requirements.}
\label{fig:square_moves}
\end{figure}

\later{
\iffull\paragraph*{Prototypes.}\label{sec_proto}
\fi
Most lattice-based modular robots are based on a hexagonal or on a square/cubic lattice. 
HexBots~\cite{HexBot09}
are hexagonal pivoting modules able to perform restricted moves with the free-space requirements depicted in Figure~\ref{fig:hex_moves}. 
Fracta~\cite{Fractum} 
modules organize in a hexagonal lattice and there they can perform both restricted and monkey moves with the free-space requirements illustrated in Figure~\ref{fig:hex_moves}.
Moreover, the millimeter-scale Catoms~\cite{bkirby-iros07} 
are able to perform pivoting moves with lighter free-space requirements given their cylindrical shape. Metamorphic units~\cite{metamorphic96} 
have the lightest free-space requirements for moving, only requiring the starting and the ending grid positions to be empty.   
Other modular robots 
including 
M-Lattice~\cite{MLattice_planning},
M-TRAN~\cite{M-tran},
PolyBot~\cite{PolyBot},
SuperBot~\cite{SuperBot},
with various shapes and capabilities can 
form meta-modules that organize and move in a hexagonal lattice with little free-space requirements~\cite{HurtadoMRA15}. 
 
XBots~\cite{heuristics-square} and M-blocks~\cite{M-blocks} are pivoting cubes moving in 2D and 3D, respectively. 
They able to perform restricted and leapfrog moves as well as \emph{straight} monkey moves (left half of Figure~\ref{fig:sq_mon}). 
 I(CES)-Cubes~\cite{ICubes} 
 and 
 Microunits~\cite{micro00} 
 can perform restricted moves with the free-space requirements illustrated in Figure~\ref{fig:square_moves}. 
Other units that organize in the square lattice  
have lighter free-space requirements~\cite{EMCube,Chiang01,Vertical98}. 
Using square/cubic meta-modules, multiple other prototypes 
have the lightest free-space requirements, only requiring the starting and the ending grid positions to be empty~\cite{metamodule1,metamodule2}.
}

\paragraph*{Related work and contribution.}
The most fundamental question within the modular robot setting is whether universal reconfiguration is possible. That is, can we have an algorithm to transform any (connected) configuration of $n$ modules into another configuration with the same number of modules? 

Indeed, efficient algorithms are known for universal reconfiguration of modular robots using moves that have significantly lighter free-space requirements~\cite{squeezing11,pushing-squares,Dumitrescu-Suzuki-Yamashita-2004,MeltGrow}. 
Relaxing the connectivity requirement 
has also lead to reconfigurability results~\cite{nadia}.

Unfortunately, the setting of this paper (pivoting robots) has proven to be more challenging. Instead, previous work has revolved around providing sufficient conditions. Nguyen, Guibas and Kim~\cite{density} showed that 
reconfiguration of hexagonal robots using only restricted moves is always possible between configurations without  
the forbidden pattern illustrated in Figure~\ref{fig:patterns} (left). 
Similarly, for pivoting squares, Sung et~al.~\cite{M-blocks} presented an algorithm 
for reconfiguring between configurations without the patterns shown in Figure~\ref{fig:patterns} (right). 
%
%
These algorithms fail to provide reconfiguration guarantees as soon as the configuration contains a single copy of the forbidden pattern. In an attempt to remove global requirements, a recent result~\cite{musketeers} introduced a different type of necessary condition: they provide an efficient algorithm for reconfiguring between any two configurations that have $5$ modules on the external boundary that can freely move (for pivoting squares in the monkey model). Other algorithms to reconfigure pivoting squares and hexagons are heuristics that 
do not provide termination guarantees~\cite{heuristics-square,heuristics}. 
\begin{figure}[thb]
    \centering
\begin{minipage}{.3\textwidth}
\centering
\includegraphics[page=6,scale=0.6]{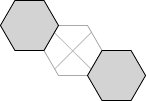}
\end{minipage}
\hfill
\begin{minipage}{.57\textwidth}
\includegraphics[page=8]{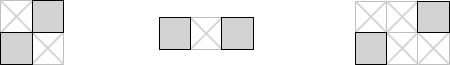}
\end{minipage}
\caption{Forbidden patterns in previous algorithms for hexagonal and square pivoting modules}
\label{fig:patterns}
\end{figure}

Despite the many attempts, 
universal reconfiguration 
remains unsolved in the setting 
of edge-connected pivoting robots. In this paper we answer this question for all five pivoting models for hexagons and squares. Specifically, we answer it positively for the hexagonal monkey model by giving a 
universal reconfiguration algorithm in Section~\ref{sec:algo}. For all other models we show that it is PSPACE-hard to determine if we can reconfigure from one configuration of modules into another one. In the process, we prove a stronger PSPACE-hardness result about a restricted form of motion planning with reversible, deterministic gadgets from~\cite{motionplanning2} (our reduction highly limits in which direction each edge can be traversed, effectively reducing the number of cases to consider). This framework has already proven useful in other swarm robot motion planning models~\cite{balanza2019full, caballero2020relocating} and we believe the improvements here will aid in future PSPACE-completeness proofs. The framework is described in Section~\ref{sec:Balanced} and is afterwards used for hexagonal restricted robots in Section~\ref{hexmodel}, and for all square models in Sections~\ref{sec_squarereduc} and \ref{sec:PSPCAE-sq-lf-monkey}. A summarizing table of our results is shown in Table~\ref{tb:results}.

\begin{table}[htb]
\def\halfdown#1{\smash{\raisebox{-1.5ex}{#1}}}
\def\halfup#1{\smash{\raisebox{1.75ex}{#1}}}
\def\rowspaceup{\rule{0pt}{1.05\normalbaselineskip}}
\centering
\begin{tabular}{l|c|c|c}
{\textbf{Model}}  
& \textbf{Restricted} & \textbf{Leapfrog} & \textbf{Monkey} \\ \hline
{\textbf{Hexagons}}\rowspaceup   &   PSPACE-hard (Thm.~\ref{theo_hexrestrictedhard}) & \cellcolor{gray!30}{N/A}   &  $O(n^3)$ universal (Thm.~\ref{thm:alg}) \\ 
[1ex] \hline
\halfdown{\textbf{Squares}} \rowspaceup  &   PSPACE-hard    &  PSPACE-hard     & PSPACE-hard (Thm.~\ref{thm:square-monkeyleap})  \\
 &  (Thm.~\ref{theo_squarerestrictedhard})   &  (Thm.~\ref{thm:square-monkeyleap})   & $O(n^2)$ if $+5$ modules~\cite{musketeers}  \\ 
\end{tabular}
\medskip

\caption{{\iffull Summary of reconfiguration results for pivoting squares and hexagons.\fi}
{}
The leapfrog moves only makes sense for square modules.}
\label{tb:results}
\end{table}

\iffull
One way to reduce the number of moves (but requiring a constant number of modules on the external boundary that are free to move) would be to try and extend the algorithm in~\cite{musketeers} for pivoting squares in the monkey model to hexagons in the monkey model. Unfortunately, it seems that this technique cannot be directly translated, as we discuss in the conclusions (Section~\ref{sec:conclusion}).

\fi

\later{
\section{Hexagonal models}\label{sec:models}
In this section we give some further insights 
for both hexagonal models. 

\paragraph*{Restricted model.}
Figure~\ref{fig:hexa-rigid1} shows an example of a rigid configuration, i.e., a configuration in which no module can move. It is inspired by the one in~\cite{density}, but it is different in that it does not assume the existence of a fixed non-movable module.

\begin{figure}[hbt]
    \centering
\begin{minipage}{.3\textwidth}
\centering
\includegraphics[page=10,scale=0.4]{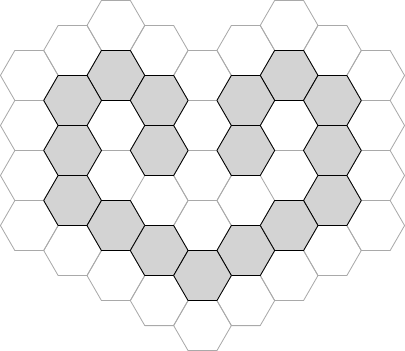}
\subcaption{Rigid configuration.}\label{fig:hexa-rigid1}
\end{minipage}
\hfill
\begin{minipage}{.65\textwidth}
\centering
\includegraphics[page=9,scale=0.4]{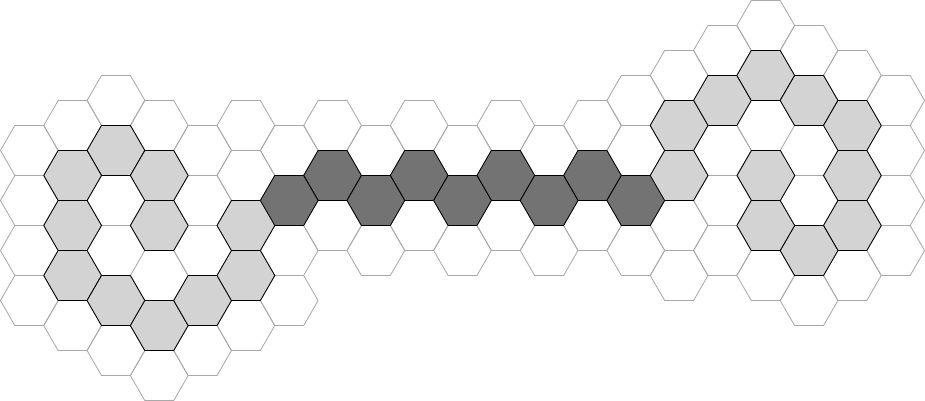}
\subcaption{Constructing different rigid configurations.}\label{fig:hexa-singletons}
\end{minipage}
\caption{Configuration that cannot be reconfigured using only the restricted move.}
\label{fig:hexa-rigid}
\end{figure}

In the restricted model the reconfiguration space is complex. 
For any positive integer $n$, the reconfiguration graph ${\cal G}_n$ in the restricted model 
has a node for each connected configuration with $n$ modules, 
and an edge between two nodes if the corresponding configurations can be reconfigured into each other through a single restricted move.
Using arguments similar to those in~\cite{musketeers} 
we can show that ${\cal G}_n$ contains are exponentially many singletons and exponentially many components of exponential size. 
In the construction illustrated in Figure~\ref{fig:hexa-singletons} 
no module can make a restricted move without disconnecting the configuration. 
Moreover, the monotone path with $p$ dark modules 
can be arranged in $2^{p-1}$ different ways, as every pair of adjacent such modules can be connected in two ways while preserving monotonicity in the horizontal direction. 
The two spirals with light modules have a constant number of modules ($11$ each).
Thus, ${\cal G}_n$ has $\Omega (2^{n})$ singletons. 
A similar idea can be used to lower bound the number of larger connected components in ${\cal G}_n$.
\begin{proposition}\label{prop:exponential exponential}
	The reconfiguration graph  ${\cal G}_n$ of connected pivoting hexagons in the restricted model has an exponential number of connected components of exponential size. 
\end{proposition}

\begin{proof}
	Consider the configuration shown in Figure~\ref{fig:hexa-exponential-exponential}. Notice that each of the spirals has a tip (the pink module) that can only pivot back and forth between two grid positions. Assume that a constant fraction of the $n$ modules (filled in dark gray in the figure) form the path connecting the two rightmost spirals. Then the number of spirals is also a constant fraction of $n$, since each spiral has constant size. Therefore, the number of configurations reachable from the one depicted in Figure~\ref{fig:hexa-exponential-exponential} is exponential in $n$. Furthermore, since the number of (dark gray) modules in the path is a constant fraction of $n$, and every pair of adjacent such modules can be connected in at least two ways (the right module in each pair could be North-East or South-East of the left module), ${{\cal G}}_n$ contains an exponential number of such connected components.
	\begin{figure}[hbt]
		\centering
		\includegraphics[page=14, width=\textwidth]{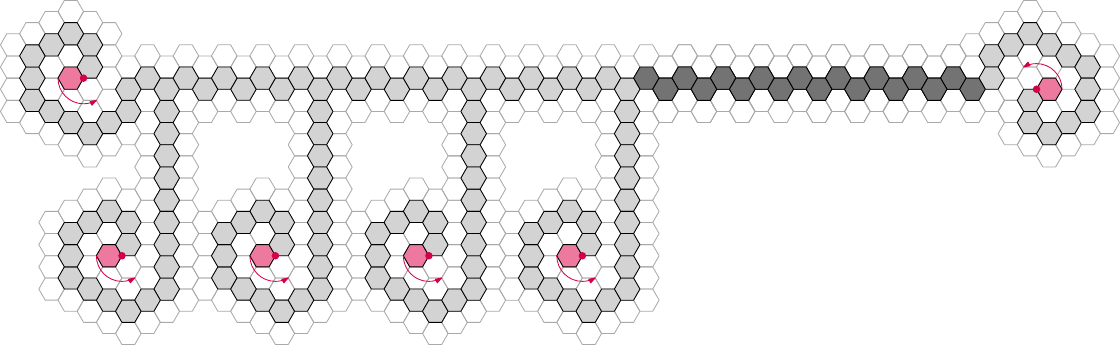}
		\caption{A configuration that in the restricted model has an exponential number of 
		connected configurations in ${\cal G}_n$ (by pivoting the pink modules),
		and also an exponential number of non reachable configurations (by changing the shape of the path of dark modules).}
		\label{fig:hexa-exponential-exponential}
	\end{figure}
\end{proof}

In Section~\ref{sec:PSPACE} we show that the reconfiguration problem in this model is PSPACE-hard.

\paragraph*{Monkey model.}
In contrast to the restricted model, in the monkey model for pivoting hexagons there are no rigid configurations. 
The proof is similar to the one in~\cite{musketeers} for the square monkey model when forbidding certain patterns. For completeness we present it here. 

	Let $G$ be the contact graph of a given connected configuration of hexagons. The \emph{cactus graph} $T(G)$ of $G$ is defined as follows. For each simple cycle $C$ in $G$, consider the set $Region(C)$ of grid positions that lie in $C$ or are enclosed by $C$. A simple cycle $C$ is said to be maximal if $Region(C)$ is maximal with respect to inclusion. We define $T(G)$ as the connected subgraph of $G$ that contains all the leaves of $G$, all the maximal cycles of $G$, and all the connections among them.
Figure~\ref{fig:cactus-graph} illustrates this definition. 
The cactus graph has been used in other reconfiguration papers for square modules~\cite{pushing-squares,musketeers}. 

\begin{figure}[tbh]
	\centering
	\includegraphics[page=16,scale=0.75]{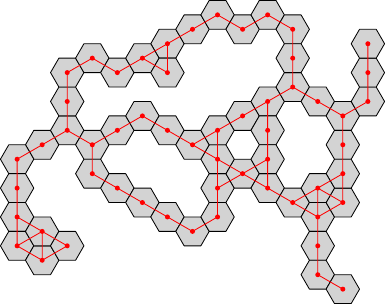}
	\qquad\includegraphics[page=17,scale=0.75]{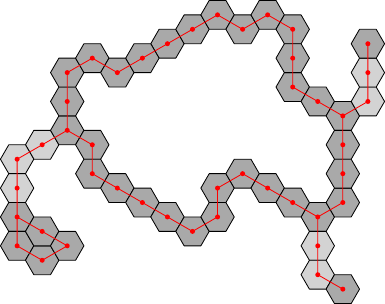}
	\caption{The contact graph $G$ of a connected configuration (left) and the corresponding cactus graph $T(G)$ (right); leaves and maximal cycles are dark-shaded.}
	\label{fig:cactus-graph}
\end{figure}

A \emph{corner} of a connected configuration of hexagons is a module that is adjacent to at least three empty grid positions through three consecutive edges. 

\begin{lemma}\label{lem:corner}
	Let $C$ be any connected configuration of hexagons and let $G$ be the contact graph of $C$. 
	There always exists a corner in $C$ that is not a cut vertex of $G$.
\end{lemma}

\begin{proof}
	Let $T(G)$ be the cactus graph of $G$. 
	If $T(G)$ has a node $m$ of degree one, then $m$ corresponds to a corner in $C$ that is not a cut vertex of $G$, so the lemma holds. Otherwise, we view $T(G)$ as a tree of cycles, and arbitrarily pick a leaf cycle in $T(G)$. 
	The leaf cycle must have a corner $c$ different from the (unique) node that connects it to the rest of $T(G)$. Then $c$ cannot be a cut vertex. 
\end{proof}


\begin{proposition}\label{prop:no singletons}
Let $C$ be any connected configuration of hexagons and let $G$ be the contact graph of $C$. 
There exist a module $c$ in $C$ that can pivot clockwise.  
Moreover, $c$ can keep pivoting clockwise until returning to its starting position.
\end{proposition}

\begin{proof}
	Let $C$ be a connected configuration of pivoting hexagons. By Lemma~\ref{lem:corner}, there exists a module $c$ in~$C$ such that $c$ is a corner and the removal of $c$ does not disconnect $C$. 
	Figure~\ref{fig:hexa-no-singletons} shows the case analysis of the proof that $c$ (the pink module) can pivot clockwise. 
	Moreover, after pivoting, in all the cases it is still a corner. 
	Thus, the invariant is maintained and $c$ can keep pivoting clockwise.
\begin{figure}[hbt]
	\centering
	\includegraphics[page=3,width=\textwidth]{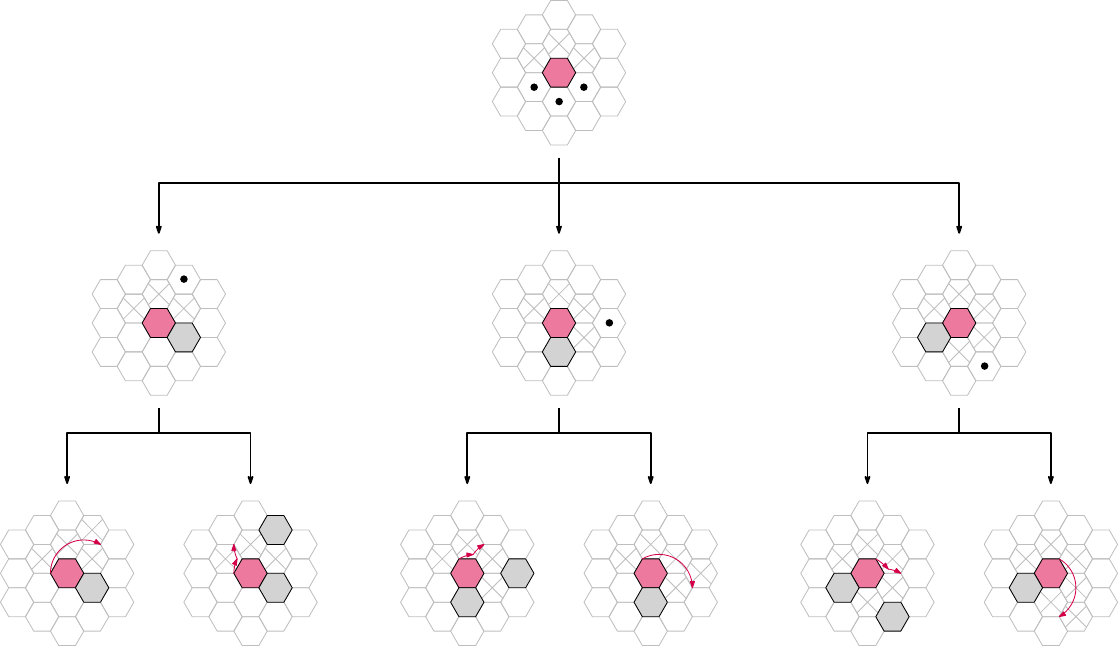}
	\caption{Proof of Proposition~\ref{prop:no singletons}. The cells marked with a cross are empty. The pink cell is the corner $c$ not disconnecting~$C$.}
	\label{fig:hexa-no-singletons}
\end{figure}
\end{proof} 

In Section~\ref{sec:algo} we show that, in the monkey model, not only there is always a module that can pivot but we can always reconfigure. 




}

\section{Polynomial algorithm for the hexagonal  monkey model}
\label{sec:algo}

\newcommand{\deflate}{\textsf{Deflate}}
\newcommand{\bridge}{\textsf{Bridge}}
\newcommand{\bbup}{\textsf{Bubble-Up}}
\newcommand{\inc}{\textsf{Incorporate}}
\newcommand{\lbridge}{\textsf{Local-Bridge}}
\newcommand{\destr}{\textsf{Merge}}
\newcommand{\shift}{\textsf{Shift}}
\newcommand{\infl}{\textsf{Inflate}}
\newcommand{\pass}{\textsf{Pass}}
\newcommand{\pivot}{\textsf{Pivot}}
\newcommand{\crew}{crew}
\newcommand{\cg}{contact graph}

This section describes an algorithm that computes a sequence of $O(n^3)$ moves in the monkey model that transforms a given configuration with $n$ modules into another.
Our approach uses a \emph{canonical configuration} defined as the configuration with $n$ modules whose {\cg} is a path and each module is only adjacent to modules above and/or below it. 
Since each move is reversible, an algorithm that takes a configuration and transforms it into the canonical configuration within $O(n^3)$ moves can be used to compute $O(n^3)$ moves between any pair of configurations. 
The main strategy \iffull of our algorithm \fi is to increase the connectivity of the {\cg}\footnote{Increasing connectivity of the {\cg} of the configuration is a concept that has recently proven useful in the different setting of reconfiguring \emph{sliding squares}~\cite{compacting_squares}.}.
Note that if the {\cg} is 2-connected, every convex corner of the configuration is movable, including the modules that are extremal in a grid direction.
Then, there is a module that can move to become the new topmost module by attaching itself to a previous topmost module.
We proceed in this manner inductively building the canonical configuration.

\medskip\noindent\textbf{Definitions and Preliminaries.}
The {\cg} $G$ is the adjacency graph of the modules in a configuration.
Since connectivity is important for the problem, we use the \emph{block tree} $\mathcal{B}$ of the {\cg} $G$.
A graph is \emph{2-connected} if it contains no cut vertices.
A \emph{block} (also \iffull called \fi 2-connected component) of $G$ is a maximal subgraph of $G$ that is 2-connected.
We call a block containing a single edge a \emph{trivial block}.
We define $\mathcal{B}$ to be a bipartite tree whose nodes are the blocks of $G$ in one partite set, and its cut vertices in the other partite set.
There is an edge between two nodes if the corresponding cut vertex is contained in the corresponding block.
The deletion of a cut vertex $v$ of a connected $G$ splits it into $2$ or more components. A subgraph induced by a such component union with $\{v\}$ is called a \emph{split component} of $v$. 
\iffull The block tree can be computed in $O(|V(G)|)$ time in a planar graph $G$ by finding cut vertices and recursively splitting split components.\fi
Similarly, a \emph{2-cut} is a pair of vertices $\{v_1, v_2\}$ whose deletion increases the number of components of~$G$.
Its \emph{2-split components} are the subgraphs induced by each of components obtained by the deletion of $\{v_1, v_2\}$ union with $\{v_1, v_2\}$.

We now give some more specific definitions used in the algorithm. 
Note that a module correspond to a vertex in $G$. 
We abuse notation referring to them interchangeably.
We label the topmost rightmost module of $G$ the \emph{root}.
We root $\mathcal{B}$ at the node containing the root module.
A cut vertex (2-cut) defines one \emph{parent} (2-)split component, containing the root module, and one or more \emph{child} (2-)split components.
Such a cut vertex (2-cut) is called the \emph{parent} of its children (2-)split components.
A 2-split component $\ell$ is \emph{trivial} if $|V(\ell)|$ is $3$ or $4$.
The parent of such a component is also called trivial.
Note that because $G$ is a subset of the triangular grid, one of its faces is either the external face (whose edges form the  \emph{boundary}), a triangle, or encloses an empty position of the grid, which we call a \emph{pocket}. 
In a 2-connected block $\ell$, if $v$ is a vertex in the boundary of $\ell$ and it is not incident to any pocket, deleting $v$ can cause \textit{only adjacent} vertices to become cut vertices.
We call a 2-cut $\{v_1, v_2\}$ \emph{adjacent} if $v_1$ and $v_2$ are adjacent.
Note that when $\{v_1, v_2\}$ is an adjacent trivial 2-cut, the faces of the  trivial 2-split component are triangles.
Our algorithm uses the following fact about adjacent nontrivial 2-cuts.

\begin{observation}
\label{obs:adj2cut}
If $\{v_1, v_2\}$ is an adjacent nontrivial 2-cut, then $\{v_1, v_2\}$ has only two 2-split components.
Furthermore, if $v_1$ is movable, $\{v_1, v_2\}$ is the only 2-cut containing $v_1$.
\end{observation}

The previous observation comes from the maximum degree of the triangular grid.
For an adjacent 2-cut to have $3$ 2-split components, two of them must be trivial.
The fact that $v_1$ needs 3 adjacent empty positions around it to be movable implies that $v_1$ must be adjacent to 2 modules other than $v_2$. Any cycle through $v_1$ connecting the (2-)split components of $\{v_1, v_2\}$ must go through the two modules adjacent to $v_1$ that are not $v_2$.


The main technical part of our algorithm is a procedure called \destr\ that increases the 2-connectivity of $G$, i.e., decreases the number of nodes in $\mathcal{B}$.
For that, we want to move modules in order to create new paths between blocks of $G$ without destroying previously existing blocks.
We define a \emph{2-free} module to be a movable module whose deletion preserves 2-connectivity in the block containing it (a module that is not in a 2-cut).
A \emph{{\crew}} $c = (m_1, \ldots, m_k)$ is a sequence of modules that induce a connected component of $G$ such that $m_1$ is 2-free, and $m_i$, $i \in \{2,\ldots,k\}$ is 2-free after the deletion of all $m_j$, $j \in \{1,\ldots,i-1\}$.
For  given 2-connected subgraph $\ell$ of $G$, let $\overline{\ell}$ be the induced subgraph of $G$ given by $V(G)\setminus V(\ell)$.
A \emph{bridge} from $\ell$ is a set of modules that were previously a {\crew} that moved to create a path between $\ell$ and~$\overline{\ell}$, thus potentially not being 2-free anymore.
We say a set of modules \emph{bridges from $\ell$} if they move to create a new path between $\ell$ and $\overline{\ell}$.
One of the goals of the algorithm is to get a {\crew} of size $3$ in a group of grid positions called \emph{flower} that is otherwise empty.
That allows us to maneuver the modules in the {\crew} to move them and create a bridge while not creating new blocks.
Let a \emph{flower} be a set of grid positions defined by a \emph{center} cell and the six adjacent positions.
A flower is \emph{adjacent} to a grid position if it does not contain it but it contains a grid position that is adjacent to it.
A flower is \emph{valid} for a 2-connected configuration $\ell$ and a disjoint {\crew} $c$ if it  contains \textit{exactly} the modules in $c$ (all modules in $c$ and no other modules),and  is adjacent to a module in $\ell$.

The following are definitions that help us describe positions in the configuration. We might reflect and/or rotate the configuration in order to fit our description w.l.o.g., and the following definition always refer to the current frame of reference.
A \emph{row} containing a position $p$ are all positions $p + (- \frac{\sqrt{3}}{2}, \frac{1}{2})i$ for some integer $i$.
An \emph{ascending} (\emph{descending}) path in a row $\rho$ is a path  $(m_1, \ldots, m_k)$ induced by modules in $\rho$ such that $m_{i+1}$ is the top-left (bottom-right) neighbor of $m_i$.
A \emph{extreme path} is a path induced by modules that are in the convex hull.
Due to the geometry of the grid, extreme paths can only have $6$ possible directions.
A \emph{SW extreme path} of a configuration $\ell$ is an ascending or descending path in the lower hull of $\ell$.
Given a position $p$ in the grid, we use a sequence of arrow superscripts on $p$ to describe positions nearby. 
For example, $p^{\uparrow\nearrow}$ refers to the position to the top-right of the position above $p$, i.e., $p + (\frac{\sqrt{3}}{2}, \frac{3}{2})$.
We overload this notation to refer to the current positions of modules, replacing $p$ by a module.

\medskip\noindent\textbf{Main algorithm.}
Here we give the general descriptions of our algorithm and the procedures it uses.
We proceed to give details of the procedures and finally we analyse it.
We split the {\cg} into two parts: the canonical path $P$ which is a canonical configuration, and the remainder of the graph $G$.
We initialize $G$ to be the entire {\cg} and $P$ empty.
Let $\mathcal{B}$ be the block tree of $G$ rooted at the block containing the topmost rightmost module.
We divide our algorithm into three phases. 
\emph{Phase 1} is a prepossessing procedure that eliminate all trivial leaves of  $\mathcal{B}$.
Then, assume that every leaf of $\mathcal{B}$ contains at least 3 modules and no further procedure will change that.
\emph{Phase 2} transforms $G$ into 2-connected.
While~$\mathcal{B}$ is not a single node, let $\ell$ be a leaf of $\mathcal{B}$.
We will apply \destr$(\ell)$, outlined in Algorithm~\ref{alg:destroy}, that will cause $\ell$ to merge with other nodes of $\mathcal{B}$ until $G$ becomes 2-connected. 
\emph{Phase 3} builds $P$.
We decrease the size of $G$ while adding modules to $P$ by moving a {\crew} on its boundary so that each of its members in turn move to become the new topmost module in the {\cg}. We use a slightly modified version of \destr\ to produce such a {\crew} without breaking the 2-connectivity of $G$.

\subsection{Phase 1: Removing trivial leaves}
\label{sec:phase1}

Phase~1 
reconfigures a connected configuration into \iffull a configuration\fi without vertices of degree~1 (which are in trivial leaves) in $G$. 
\iffull We note that 
there exist \fi
{}
configurations in which it is not enough to just pivot the degree-1 modules, i.e., this task requires coordination with other modules. 
\iffull Figure~\ref{fig:notrivial} shows such a configuration: $m$ (the degree-1 module is the only trivial leaf) can only pivot to five other grid positions, and in all of them it is adjacent to exactly one other module. \fi
\iffull However, in the next lemma we show that we can remove all trivial leaves by pivoting the modules in them and making some local transformations.\fi 

\begin{figure}[ht]
	\centering
	\includegraphics[page=5]{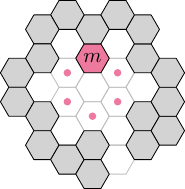}
	\caption{Configuration with one trivial leaf ($m$) that cannot be removed by pivoting it.}
	\label{fig:notrivial}
\end{figure}

\later{
{}
}
\both{
\begin{lemma}
\label{lem:phase1}
A connected configuration of $n>2$ hexagons can be transformed in $O(n^2)$ moves into a configuration without trivial leaves in the contact graph
without breaking connectivity.
\end{lemma}
}

\later{
\begin{proof}
Given a connected configuration $C$ with  a module $m$ of degree $1$ in the contact graph, we will show that in $O(n)$ pivoting moves,  without breaking connectivity,  
we can transform~$C$ into a configuration with strictly less modules of degree $1$.

We start by letting $m$ pivot clockwise. 
If $m$ can reach a position in which it is adjacent to at least two modules we are done. 
However, this might not be the case as we saw before. 
We fix a root module $r$ different from $m$. 
If $m$ can only reach positions adjacent to one module, 
we place $m$ in the grid position adjacent to a module $p$ that is the furthest possible from $r$ in the contact graph. 
Without loss of generality we assume that $p$ is on top of $m$. 

Apart from the cell containing $p$, 
all the other grid positions adjacent to $m$ are empty. 
We now distinguish some cases depending on which grid cells at distance~2 from $m$ are occupied. 
To simplify the description, we give them names, as illustrated in Figure~\ref{fig:deg1} (top). 
We define 
$\alpha=m^{\uparrow\uparrow}$, 
$\beta=m^{\uparrow\nearrow}$, 
$\beta'=m^{\uparrow\nwarrow}$, 
$\gamma=m^{\nearrow\nearrow}$, 
$\gamma'=m^{\nwarrow\nwarrow}$, 
$\delta=m^{\searrow\nearrow}$,
$\delta'=m^{\swarrow\nwarrow}$, 
$\varepsilon=m^{\searrow\searrow}$,
$\varepsilon'=m^{\swarrow\swarrow}$, 
$\zeta=m^{\searrow\downarrow}$, and 
$\zeta'=m^{\swarrow\downarrow}$. 
In all the cases the transformations 
(i) preserve connectivity, 
(ii) a module with degree~1 in the contact graph becomes a module with degree at least $2$, and 
(iii) no module with degree at least~2 in the contact graph of becomes a module with degree~1. 

\begin{figure}[ht]
	\centering
	\includegraphics[page=2,width=\linewidth]{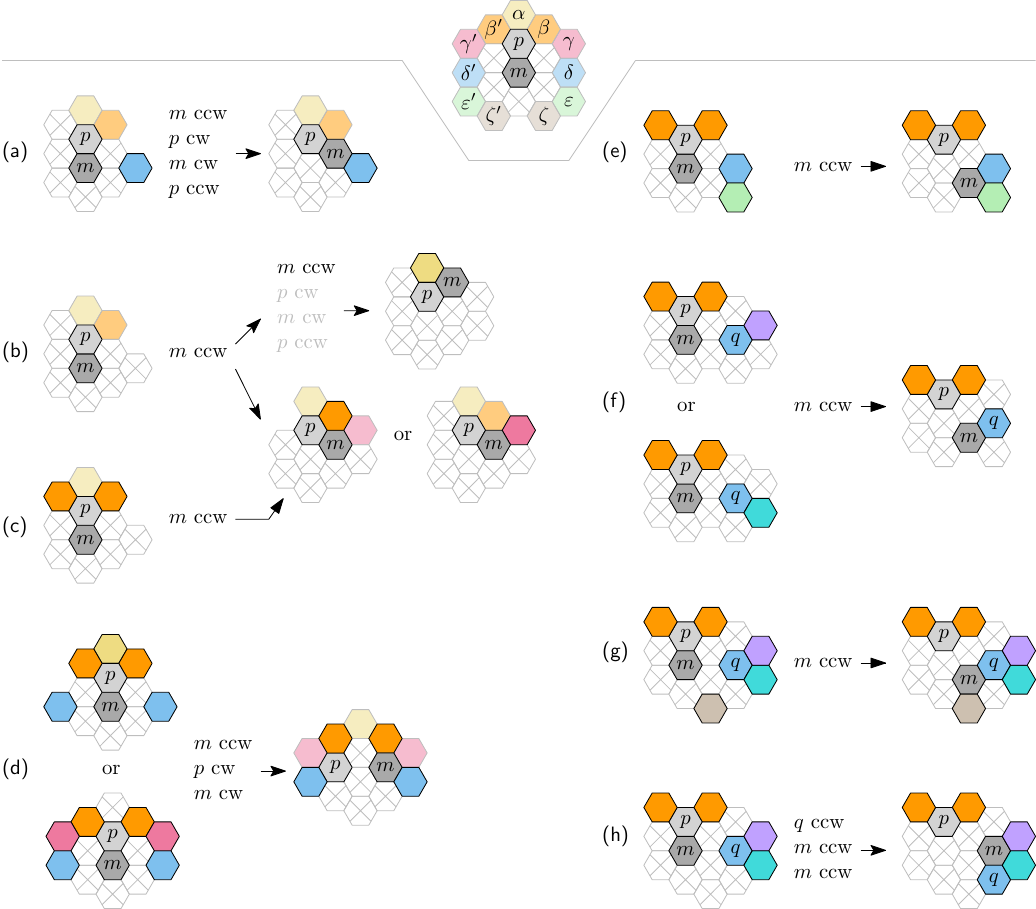}
	\caption{Reducing the number of trivial leaves without disconnecting the configuration.}
	\label{fig:deg1}
\end{figure}

We first consider the cases in which at least one of $\beta$ and $\beta'$ is empty. 
By symmetry, we can assume without loss of generality that $\beta'$ is empty. 
If $\delta$ is occupied, 
pivoting $m$ ccw, $p$ cw, $m$ cw, and finally $p$ ccw, 
effectively moves $m$ between $p$ and the module in $\delta$ without altering the position of other modules; 
see Figure~\ref{fig:deg1}~(a). 
Otherwise, if $\delta$ is empty, 
we pivot $m$ ccw. 
If $m$ becomes adjacent to another module apart from $p$ we are done. 
If that is not the case we know that $\beta$ and $\gamma$ are empty, so $\alpha$ must be occupied. 
We pivot again $m$ ccw. 
This move can be either a restricted move, 
making $m$ adjacent to both $p$ and the module in $\alpha$, 
or a monkey move. 
In the former situation we are done 
and in the latter pivoting $p$ cw, $m$ cw, and $p$ ccw also makes $m$ adjacent to both $p$ and the module in $\alpha$. 
See Figure~\ref{fig:deg1}~(b) for an illustration. 

In the remaining cases both $\beta$ and $\beta'$ are occupied. 
If at least one of $\delta$ and $\delta'$ is empty, 
$m$ has enough space to make a restricted move and become also adjacent to $\beta$ or $\beta'$; see Figure~\ref{fig:deg1}~(c). 
Thus, in the remaining we can also assume that both $\delta$ and $\delta'$ are occupied. 
If possible, we want to place $m$ and $p$ adjacent to the modules in $\beta$ and $\delta$ and in $\beta'$ and $\delta'$, respectively. 
This transformation can be implemented pivoting $m$ ccw, $p$ cw, and $m$ cw; 
see Figure~\ref{fig:deg1}~(d). 
Doing that we do not create any new trivial leaf in the contact graph, 
but we have to be careful not to break connectivity when pivoting $p$. 
If $\alpha$ is occupied it is clear that pivoting $p$ does not disconnect the configuration. 
We now show that if both $\gamma$ and $\gamma'$ are occupied also then pivoting $p$ does not disconnect the configuration.
By the choice of $p$, 
the modules in $\delta$ and $\delta'$ cannot be further to $r$ than $p$ in the contact graph of $C$. 
Thus, neither $p$ nor the modules in $\beta$ and $\beta'$ are in the shortest paths from $r$ to $\delta$ and from $r$ to $\delta'$. 
The modules in these shortest paths together with the modules in $\gamma$ and $\gamma'$ connect $\beta$ and $\beta'$. 
This implies that pivoting ($m$ and) $p$ does not disconnect the configuration. 

For the cases left, 
apart from $\beta$, $\beta'$, $\delta$ and $\delta'$ being occupied, 
we can assume that $\gamma$ and $\gamma'$ are empty. 
Let $q$ be the module in $\delta$. 
If $\varepsilon$ 
is occupied $m$ can pivot ccw 
and become adjacent to this module and to $q$.  
Moreover, after moving $m$, $p$ is still adjacent to at least two modules; see Figure~\ref{fig:deg1}~(e). 
Thus, we can assume that $\varepsilon$ 
is empty and therefore $q$ has at most~2 adjacent modules. 
If $q$ has only one adjacent module it is a trivial in the block graph of $C$ and we can get rid of it by pivoting $m$ ccw; see Figure~\ref{fig:deg1}~(f). 
Otherwise, if $q$ has two adjacent modules,
we distinguish between $\zeta$ being occupied and $\zeta$ being empty. 
In the former case, pivoting $m$ ccw makes it adjacent to $q$ and the module in $\zeta$; see Figure~\ref{fig:deg1}~(g). 
In the latter case, pivoting $q$ ccw and $m$ ccw twice 
results in $m$, $q$, and the modules adjacent to $q$ having all at least two adjacent modules; see Figure~\ref{fig:deg1}~(h).
\end{proof}
}

\subsection{Phase 2: Merging leaves}
\label{sec:phase2}

The goal of Phase 2 is to take a connected configuration with no degree 1 vertices, and transform it into a 2-connected configuration in $O(n^3)$ moves. The main technical tool of this phase is the procedure \destr\, outlined in Algorithm~\ref{alg:destroy}, which allows us to reduce the number of 2-connected components by merging them.
Its input is a child (2-)split component of a cut vertex $v$ (adjacent 2-cut $\{v_1, v_2\}$).
We first apply the necessary rotations so that $v$ ($\{v_1, v_2\}$) is farthest to the row $\rho_0$ containing the extreme SW path of $\ell$.
We then assume that $\rho_0$ does not include $v$ ($\{v_1, v_2\}$) and neither does the row above it except for the base case when $|V(\ell)| = 3$ and $\ell$ is a split component, or when $|V(\ell)| = 5$ and $\ell$ is a 2-split component.
The output of the algorithm is a set of modules that, after $O(|V(\ell)|^2)$ moves, bridges from $\ell$.

Refer to Algorithm~\ref{alg:destroy}.
\destr\ uses several others sub-procedures which we outline here.
\iffull We give their details later. \fi
We call $m$ the \emph{ascending} module, which by its definition in line~\ref{destr:m} is movable.
It is either 2-free, in which case we will try to move it by cw pivots to its highest possible position in $\rho_0$ before it leaves $\ell$;
or it is part of a 2-cut, in which case we make it 2-free using sub-procedures.
The end goal is to either bridge using $m$ while it ascends in $\rho_0$ if it gets blocks by a vertex $m^* \notin \ell$, or accumulate 2-free modules at the top of the configuration where a valid flower will form.
Then, the sub-procedure \bridge\ moves the valid flower around $\ell$ until it hits $\overline{\ell}$ where we create a bridge with the {\crew}.
There are $3$ main Cases given by lines~\ref{merg:1}, \ref{merg:2} and \ref{merg:3}.

Assume we are in Case 1.
If $m$ is in a trivial 2-cut, it will try to move up as explained before.
Let $m'$ be the module at $m^\nearrow$.
By Observation~\ref{obs:adj2cut}, $\{m, m'\}$ is the only 2-cut containing $m$.
If $m$ succeeds in moving up, at least one unit, that leaves $m'$ a cut vertex.
Then, in line~\ref{merg:triv},  we move the (up to $2$) modules that are in the 2-split component of $\{m, m'\}$, restoring 2-connectivity.
During $m$'s ascension in $\rho_0$, we identify whether a valid flower gets formed.
In the positive case, \bridge\ will accomplish our goal.
During its ascension, $m$ might be blocked by a module $m^* \in \overline{\ell}$.
If certain conditions are satisfied, the sub-procedure \lbridge\ uses $m$ \iffull and some nearby modules in $\ell$ \fi to create a bridge to $m^*$.
Else, the sub-procedure \inc\ moves $m$ to the row $\rho_1$ above it, or out of $\ell$, and we can find a new ascending module.

Now assume that $m$ is part of a nontrivial 2-cut (Case 2).
Then, either $m$ is part of an adjacent 2-cut or it is incident to a pocket.
In the case $m$ or an adjacent module is part of an adjacent 2-cut we recuse in the child 2-split component , which makes $m$ 2-free.
Else (Case 3), we either use \deflate, which decreases the number of empty positions enclosed by the pocket, or \bbup, which moves one of such empty positions up.
In some situations, \deflate\ produces a 2-free module in $\rho_0$ that will be the next ascending module.

\renewcommand{\emph}{}
\smallskip
\begin{algorithm}[ht]
\SetAlgoLined
\While{True}{
    Let $m$ be the topmost module in a extreme SW path of $\ell$\;\label{destr:m}
    Let $\rho_{-1}$, $\rho_0$ and $\rho_1$ be the rows below, of, and above $m$ it respectively\;
    \uIf{$m$ is 2-free or part of an adjacent trivial 2-cut}{ \label{merg:1}
        Pivot $m$ cw to the highest position in $\rho_0$ before it leaves $\ell$\;
        \uIf{$m$ was part of a trivial 2-cut}{
            Pivot cw once the other modules in the trivial child\;\label{merg:triv}
        }
        \uIf{$m$ is ever in a {\crew} $c$ of size $3$ in a valid flower $F$ during its ascension}{
            Return \bridge$(F, \ell-c)$\;
        }
        \uElseIf{$m$ bridges from $\ell$}{
            Return $m$\;
        }
        \uElseIf{the requirements of \lbridge$(m)$ are met}{
            Return \lbridge$(m)$\;
        }
        \uElse{
            \inc$(m)$\;
        }
    }
    \uElseIf{$\{m, m^\nearrow\}$ or $\{m^\searrow, m^{\searrow\nearrow}\}$ is a nontrivial 2-cut}{  \label{merg:2}
        Let $\ell'$ be the child 2-split component of the highest such 2-cut\;
        $c' := $ \destr$(\ell')$\;\label{merge:recursion}
        \uIf{$c'$ bridges between $\ell$ and $\overline{\ell}$}{
            Return $c'$;
            \Comment{$c'$ already merges $\ell$ into another block.}\\
        }
    }
    \Else  
    { \label{merg:3}
        \deflate$(m^\nearrow)$ or \bbup$(m^\nearrow)$;
        \Comment{$m^\nearrow$ is empty}\\
    }
}
\caption{\destr($\ell$)}\label{alg:destroy}
\end{algorithm}
\smallskip
\renewcommand{\emph}{\textbf}

\iffull We proceed with details of the sub-procedures.\fi


\medskip\noindent\textbf{\bridge$(\ell, F)$.}
The operation takes a 2-connected $\ell$ and a valid flower $F$ containing a {\crew} $c = (m_1, m_2, m_3)$ where $m_1$ was an ascending module. 
It returns $c$ after a sequence of moves that transforms $c$ into a bridge from $\ell$.
Compute a maximal sequence of flowers $(F_1 = F, \ldots, F_k)$, where each subsequent flower is adjacent to $\ell$, containing no modules except for $c$, and  obtained by moving the previous flower by one grid unit around the boundary of $\ell$.
We choose to move cw or ccw around $\ell$ based on the following condition.
If $\ell$ has a parent cut vertex, then choose arbitrarily.
Else, if $\ell$ has a parent  adjacent 2-cut $\{v_1, v_2\}$ where $v_1$ is movable, we chose the direction towards $v_2$ so that $F_k$ is not adjacent to $v_1$.
Since $G$ is connected and planar, and there are vertices in $\overline{\ell}$, $F_k$ is adjacent to a module $m^*$ in $\overline{\ell}$.
We show how to compute the sequence of moves to bring the the {\crew} with the sequence of flowers $(F_1, \ldots, F_k)$ and finally bridge between $\ell$ and $m^*$ in $F_k$.

\begin{figure}[ht]
    \centering
    \includegraphics[scale = 0.28]{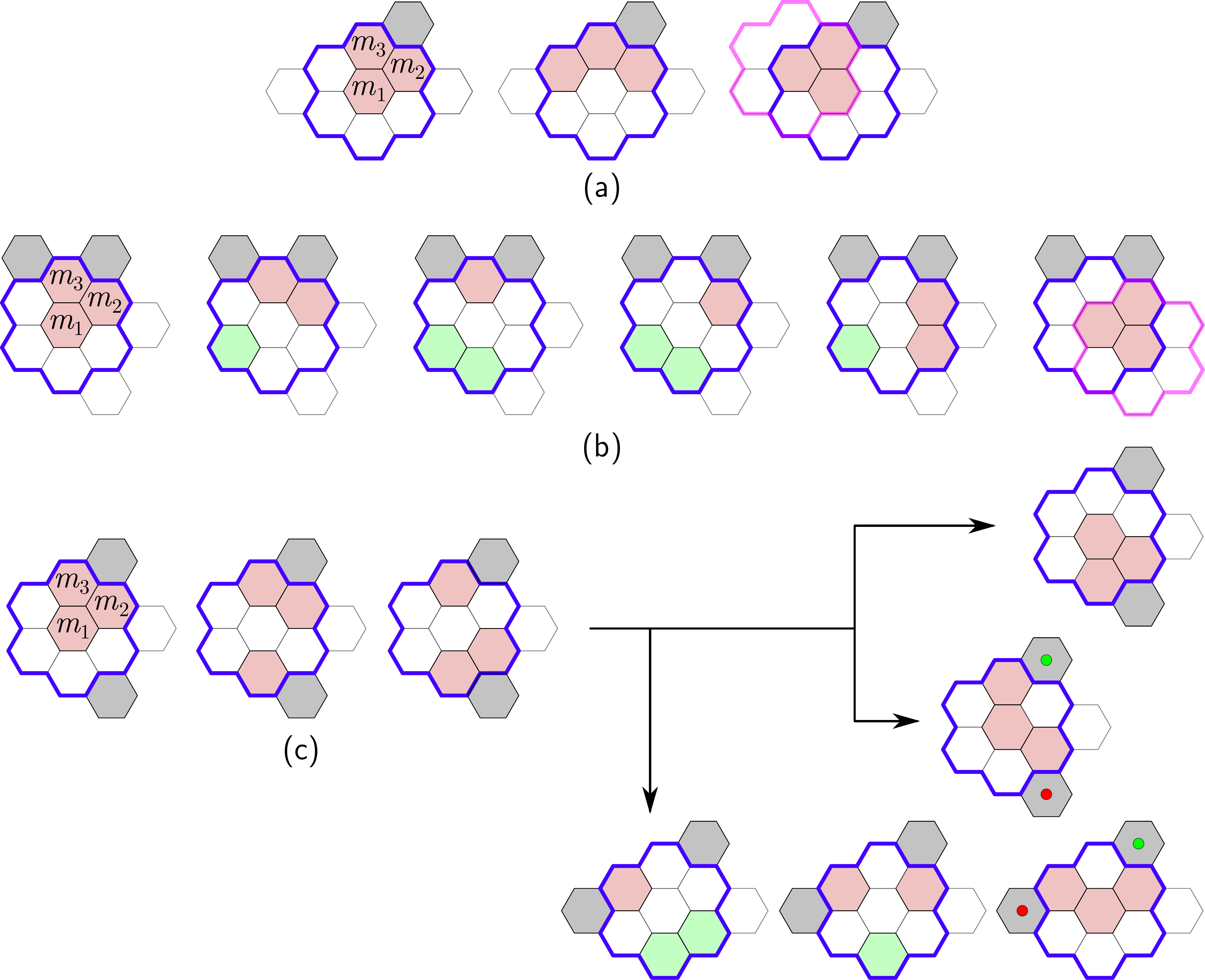}
    \caption{Maneuvers used to rotate 
    around $m_1$ 
    a {\crew} that induces a cycle. A possible next flower is shown in pink.}
    \label{fig:bridge-walk}
\end{figure}

If the modules in $c$ induce a connected graph, this graph is either a triangle, a straight path or a ``bent'' path.
A configuration of $c$ in a valid flower is \emph{useful} if a module of $c$ is adjacent to $\ell$ and $c$ induces a triangle or a ``bent'' path with $m_1$ in the center of the flower (i.e., both endpoints are adjacent to modules outside the flower).
We show how to reach every useful configuration of $c$ in a valid flower $F_i$.
This is enough to accomplish our objective since: (i) by definition, $F_i\cap F_{i+1}$ is adjacent to $\ell$ and there is a useful configuration contained in the intersection of $F_i$ and $F_{i+1}$ (Figure~\ref{fig:bridge-walk}~(a)--(b)); and (ii) if $m^-\in \ell$ is the only module adjacent to $F_k$ and $m^*$ is the only module of $\overline{\ell}$ adjacent to $F_k$ and they are across from the center of $F_k$, e.g. $m^-$ ($m^*$) is at the topmost (bottommost) position adjacent to $F_k$, then we can move $F_k$ one more unit along the boundary of $\ell$, contradicting the maximality of the sequence.
By (i) we can transition between flowers $F_i$ and $F_{i+1}$ through a useful configuration, making both valid.
By (ii) there is a useful configuration at $F_k$ that bridges between $m^-$ and $m^*$.

\later{
We present four \emph{maneuvers} and then argue that they are sufficient.
Refer to Figure~\ref{fig:bridge-walk}.
First note that, after the movement of $c$, the labels $m_1$, $m_2$, $m_3$ change according to the definition of {\crew}.
For a useful crew in $F_i$, we label $m_1$ to be the module in the center of $F_i$.
We first argue for useful crews inducing a triangle.
W.l.o.g., let $m_3$ be a module above $m_1$, and let $m_2$ a module to the top right of $m_1$.
By definition of {\crew}, the position to the bottom right of $m_2$ is empty, and at least one position outside $F_i$ and adjacent to $m_3$ contains $m^-$.
We define this labeling to be \emph{counterclockwise} because $m_2$ appears immediately after $m_3$ in a cyclic order around $m_1$.
Let's consider a different configuration of $c$ in $F_i$.


\emph{Maneuver 1.} Consider the case when the target configuration of $c$ induces a triangle, can be labeled clockwise, and is defined by the rotation of $c$ by $60^\circ$ around $m_1$ (Figure~\ref{fig:bridge-walk}~(a)).
This implies that $m_3^{\swarrow\swarrow}$ is empty.
Then we can pivot $m_1$ and $m_2$ clockwise once.

\emph{Maneuver 2.} 
Consider the case when the target configuration of $c$ induces a triangle, can be labeled counterclockwise, and is defined by the rotation of $c$ by $-60^\circ$ around $m_1$ (Figure~\ref{fig:bridge-walk}~(b)).
This implies that $m_2^{\downarrow\downarrow}$ is empty and that at least one of the positions adjacent to $m_2$ is full, since that will be the target position of the new $m_3$.
Then we can pivot $m_1$ clockwise until it leaves $F_i$ or it moves to the lower left position of $F_i$.
Pivot $m_2$ clockwise until it leaves $F_i$ or it moves to the lower left or bottom position of $F_i$.
Now, $m_3$ can move to the top right position of $F_i$, $m_2$ can move to the position below $m_3$ and $m_1$ can move to the center of $F_i$.

\emph{Maneuver 3.} 
Consider the case when the target configuration of $c$ induces a triangle, can be labeled clockwise, and is defined by the rotation of $c$ by $-120^\circ$ around $m_1$ (Figure~\ref{fig:bridge-walk}~(c)).
Then, we move $m_1$ 1 unit down with a constant number of moves if any of the 3 positions directly below $F_i$ are full.
Else we move $m_1$ outside of $F_i$.
Next, move $m_2$ and $m_3$ 1 unit down by cw and ccw pivots respectively.
If needed, return $m_1$ to $F_i$ to obtain $c'$.

\emph{Maneuvers 4 and 5.}
These maneuvers reach the target useful configurations in Figure~\ref{fig:bridge-walk}~(c) that are paths.
The moves are similar to Maneuver 3.

}

\both{
\begin{lemma}
\label{lem:crewCycle}
Every useful configuration of a {\crew} $c$ in a valid flower can be reached from any useful configuration.
\end{lemma}
}
\later{
\begin{proof}
Note that the rotation by $60^\circ$ of $c$ into a counterclockwise-labeled is possible by the reverse of Maneuvers 3.
Also note that the rotation by $-60^\circ$ of $c$ into a clockwise-labeled configuration is possible by executing the reverse of Maneuvers 3 followed by Maneuver 4. 
Then, the claim is true for when all possible positions can be transformed in $c$ via $\pm 60^\circ$ rotations, because we can use the reverse of Maneuvers 4 or 5 if necessary to go to a configuration where $c$ is a triangle and the rotate incrementally by $60^\circ$.
If such a rotation is not possible there are two possible motives.
After rotating $c$ by $\pm 60^\circ$ around $m_1$, either: no module can be labeled $m_2$, i.e., two modules in the {\crew} would be between two other modules, leaving no space for any to move after deleting $m_1$; or no module can be labeled $m_3$, i.e., the two modules in the {\crew} not in the center of $F_i$ are not adjacent to any module outside $F_i$.
In the first case, Maneuvers 3 and 4 are possible, which connects the current configuration with the next possible useful configuration that is a rotation.
The last case, implies that the only two modules adjacent to $F_i$ are at $m_1^{\uparrow\uparrow}$ and $m_1^{\downarrow\downarrow}$.
Then, we can move $m_2$ $2$ units down, $m_1$ to $m_2$'s initial position, then $m_2$ to the bottom position of $F_i$, $m_1$ moves to the bottom right position of $F_i$ and finally $m_3$ to the center of $F_i$.
That effectively rotates $c$ by $120^\circ$.
\end{proof}
}

\both{
\begin{lemma}
\label{lem:bridge}
\bridge$(\ell,F)$ performs $O(|V(\ell)|)$ moves and bridges from $\ell$ while not breaking connectivity. After its execution, $\ell$ is still 2-connected. If $\{v_1, v_2\}$ is the parent 2-cut of $\ell$ and $v_1$ is movable, then $v_1$ remains movable.
\end{lemma}
}

\later{
\begin{proof}
Only members of the {\crew} move and, by definition, $\ell$ remains 2-connected and no cut vertex is moved.
As argued before, by (i) and Lemma~\ref{lem:crewCycle}, given a valid flower $F_i$ we can make $F_{i+1}$ valid with a constant number of moves.
The number of flowers $k$ is $O(|V(\ell)|)$.
By (ii) and Lemma~\ref{lem:crewCycle}, a bridge in $F_k$ can be formed in $O(1)$ moves.
Furthermore, if the initial flower is at the convex hull of $\ell$, because of the chosen direction of the flower sequence, the bridge is not adjacent to $v_1$ and $v_1$ remains movable.
%
\end{proof}
}

\later{
\subsection{\deflate\ and \bbup}
\label{sec:app-def-bbup}
}

\medskip\noindent\textbf{\deflate$(p)$ and \bbup$(p)$.}
These operation takes an empty position $p$ to the top-right of a module $m$ which is a corner of a 2-connected subgraph $\ell$ of the {\cg}.
We assume that $m$ is a corner of $\ell$ in its SW extreme path, i.e., $m^\nwarrow$, $m^\swarrow$, and $m^\downarrow$ are empty.
Then, $p$ is enclosed by $\ell$ by 2-connectivity.
Refer to Figure~\ref{fig:deflate}.
\deflate\ requires that positions surrounded by a red line in Figure~\ref{fig:deflate}~(a) or (c) are as shown.
In particular, if $m^{\searrow\nearrow}$ is full, then $m^{\uparrow\nwarrow}$ is empty.
Then, the operation fills $p$ with a module adjacent to $m$ and preserve 2-connectivity of $\ell$, effectively reducing its area.
\bbup\ requires that positions surrounded by a red line in Figure~\ref{fig:deflate}~(b) are as shown.
In particular, if $m^{\searrow\nearrow}$ is full, then $m^{\uparrow\nwarrow}$ is full.
We additionally require that $\{m^{\searrow}, m^{\searrow\nearrow}\}$ is not a nontrivial 2-cut.
Then, the operation moves the empty position and $m$ to their top-left position while preserving 2-connectivity.

\begin{figure}[ht]
    \centering
    \includegraphics[width=\textwidth]{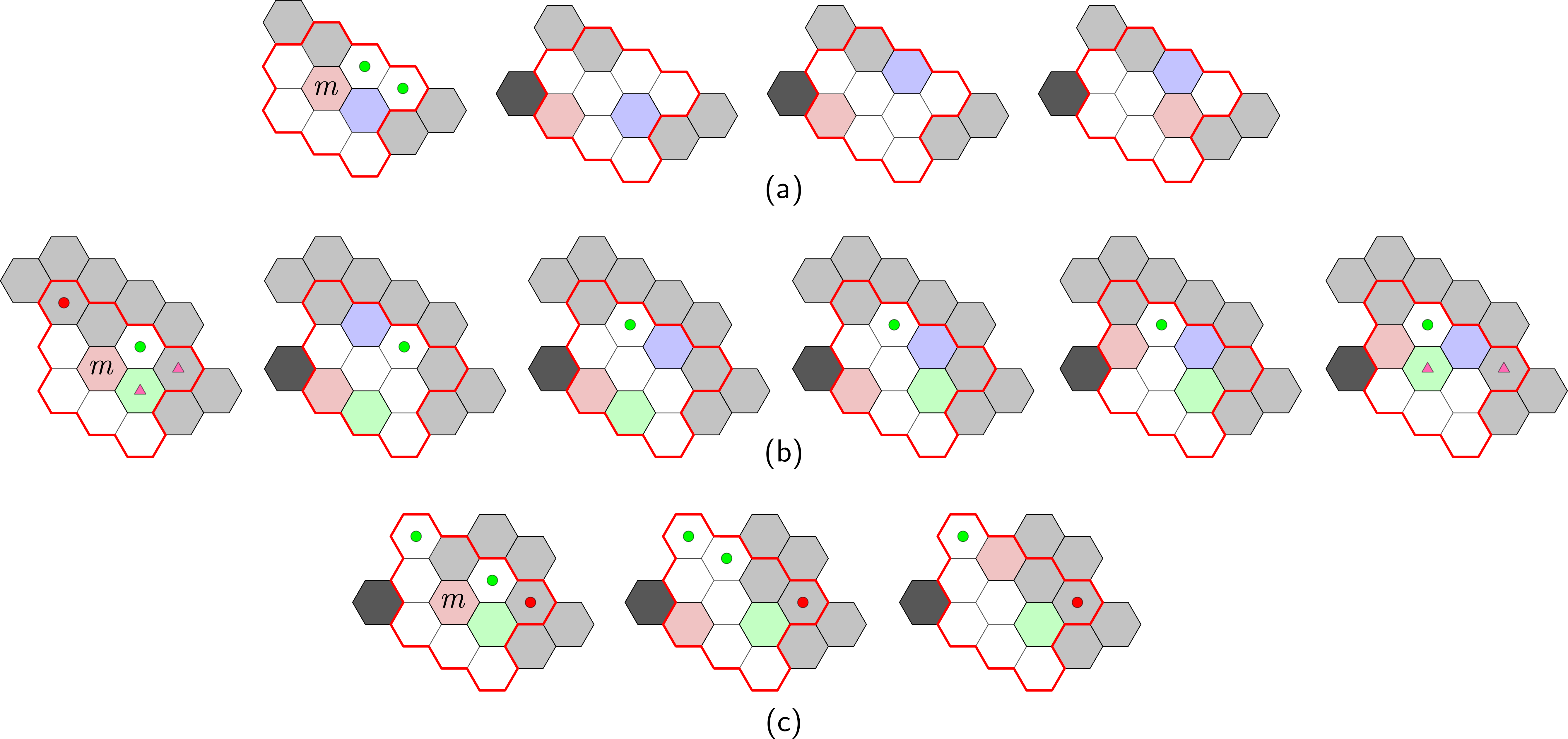}
    \caption{Operations used in \deflate$(m_1)$.}
    \label{fig:deflate}
\end{figure}

{}
\later{
{}
We now describe the operations starting with \bbup.
Let $m_{g}$, $m_b$ and $m'$ be the modules at $m^{\searrow}$, $m^{\uparrow}$ and $m^{\searrow\nearrow}$.
Refer to Figure~\ref{fig:deflate}~(b).
Move $m$ cw until it reaches $m^{\swarrow}$ or it leaves the red area. 
Do the same with $m_g$.
Pivot $m_b$ ccw which brings it to $p$.
Now bring $m_g$ and $m$ to $m_b^{\swarrow}$ and $m_b^{\swarrow\nwarrow}$.
This can be done by pivoting $m_g$ and then $m$ ccw until they reach the desired position if they left the red area, or it requires some coordination between these modules otherwise as is shown in 
Refer to Figure~\ref{fig:deflate}~(b).
If $m'$ becomes a cut vertex, it was part of a trivial 2-cut (marked with pink triangles in the figure).
Then, pivot cw the modules at $m'^\downarrow$ and $m'^\searrow$ to restore 2-connectivity.
We now describe \deflate.
Refer to Figure~\ref{fig:deflate}~(a).
In this setting, \deflate\ pivots $m$ cw once, bringing it to $m^{\nwarrow}$ or $m^{\swarrow}$ depending on the presence of a module in $m^{\swarrow\nwarrow}$.
Then we bring $m'$ up one position. This can be done by either pivoting it cw or pivoting $m'$ cw, bringing $m$ back to its original position and pivoting $m'$ ccw.
We then bring $m$ to the position below $m'$ via one or two ccw moves.
Refer to Figure~\ref{fig:deflate}~(c).
In this setting, \deflate\ does the same as \bbup\ except for the last step. Instead it brings $m_g$ to $m_b^{\downarrow}$ and $m$ to $m_b^{\nwarrow}$.
}

\both{
\begin{lemma}
\label{lem:deflate}
\deflate$(p)$ and \bbup$(p)$ perform $O(1)$ moves, do not break connectivity, and the resulting $\ell$ is 2-connected. \iflater asdf \fi
\end{lemma}
}

\later{
\begin{proof}
We only move modules at $m$, $m^\swarrow$ and possibly $m^\uparrow$ which bound the external face and the pocket containing $p$ and, hence, there is no 2-cut containing only such modules.
Then the configuration remains connected.
By definition, only $O(1)$ moves were used.
It is clear that \deflate\ preserves 2-connectivity.
We now focus on \bbup.
All pockets  pockets that are merged with the outer face by the deletion of $m$, $m^\swarrow$ and $m^\uparrow$ are closed by the final positions of these modules.
The only possible 2-cuts affected would be adjacent 2-cuts.
In particular, $\{m_g, m'\}$, shown with pink triangles in Figure~\ref{fig:deflate}~(b). 
The operation would leave $m^{\searrow\nearrow}$ a cut vertex.
However by the operation's requirement, $\{m^{\searrow}, m^{\searrow\nearrow}\}$ is not a nontrivial 2-cut.
It is is a trivial 2-cut, the 2-split component is either identical to the Figure or there is only one other module adjacent to both $m_g$ and $m'$.
The positions below those modules are empty, or else the 2-cut would not be trivial. Then, these up to $2$ modules have the necessary space to move and restore 2-connectivity.
\end{proof}
}

\later{
{}
}
\medskip\noindent\textbf{\shift$(M,d)$.}
Although not used directly in \destr, this operation is used in following sub-procedures.
The input is a sequence of modules $M = (m_1, \ldots, m_t)$ and a direction $d \in \{$cw, ccw$\}$.
\shift\ applies one pivoting move in the direction $d$ for each module in the order of $M$.
If, after moving $m_i$, there is a degree-1 module $r$ adjacent to $m_{i+1}$, move $r$ in the reverse direction of $d$ before moving $m_{i+1}$.
We of course require that $m_1$ is movable and that $m_i$ becomes movable after the deletion of $m_{i-1}$ and possibly a degree-1 module adjacent to $m_i$. 
We sometimes require the following condition.

\begin{enumerate}
    \item [$(S)$] $M$ is an descending (ascending) path in the boundary of a 2-connected subgraph $\ell$, and $d = $ cw ($d = $ ccw).
\end{enumerate}

\begin{figure}[ht]
    \centering
    \includegraphics[width=\textwidth]{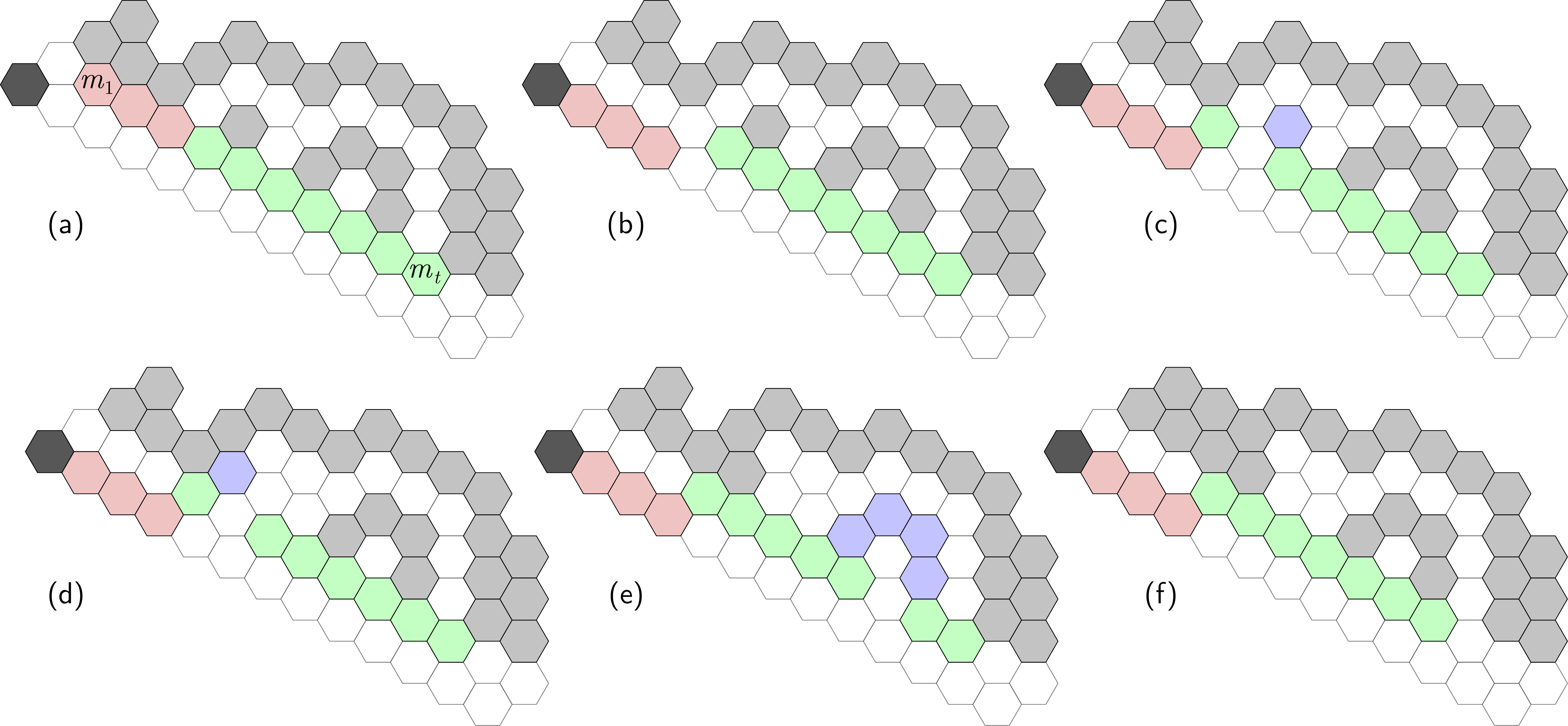}
    \caption{Illustration of \shift$(M,$cw$)$ satisfying $(S)$.}
    \label{fig:shift}
\end{figure}

\both{
\begin{lemma}
\label{lem:shift}
If a sequence of modules $M = (m_1, \ldots, m_t)$ satisfies $(S)$, then it does not break the connectivity and, after it terminates, all pockets of $\ell$ remain intact except for possibly one that has $m_t$ in its boundary.
\end{lemma}
}

\later{
\begin{proof}
Refer to Figure~\ref{fig:shift}.
W.l.o.g., $M$ is descending.
By $(S)$, no module in $M$ has a bottom neighbor except for possibly $m_t$.
If $m_1$ does not Monkey jumps, then, all other subsequent module will occupy the position of its predecessor after it moves. 
Every module above $M$ is adjacent to two of its modules except for a module above $m_1$.
Since $\ell$ is 2-connected, the deletion of $M$ can only generate components that either (i) are adjacent to $m_1$ and $m_t$, (ii) are singletons, or (iii) are adjacent to at least 3 modules in $M$.
During the movement, we can think of a prefix of $M$ have moved up, there is a gap of 1 position in the row and the rest of $M$ is unchanged. Except for the singleton components (case (ii)), all other components are still connected to either the prefix or suffix of $M$.
The singletons become degree-1 vertices that \shift\ moves up, so they too become connected.
It remains to show that the same is true when $m_1$ Monkey jumps.
If it Monkey jumps to the same row, then $m_2$ can take its former position and so on.
If it Monkey jumps to the row below, then $m_2$ either takes $m_1$'s former place, in which case we are done, or it also jumps down a row. 
Let $M_1$ be the prefix of $M$ containing modules with a top neighbor and $M_2$ the remaining suffix.
Then the modules in $M_1$ will jump down and $M_2$ will shift up in the same row.
See Figure~\ref{fig:shift}.
Then, \shift\ does not break any pocket except at $m_t$'s former position, as claimed.
\end{proof}
}

\later{}

\medskip\noindent\textbf{\infl$(m)$.}
This operation uses \shift\ to make a concave corner convex, possibly creating a new empty space enclosed by $\ell$. This will be used in \lbridge.
The input $m$ is an ascending module in a 2-connected $\ell$.
We require that $m^{\uparrow}$ and $m^{\nearrow}$ are full, neither $\{m^{\uparrow}, m^{\uparrow\nearrow}\}$ nor $\{m^{\uparrow\uparrow}, m^{\uparrow\uparrow\nwarrow}\}$ are in adjacent nontrivial 2-cuts, and that at least one position in $\{m^{\uparrow\nwarrow}, m^{\uparrow\nearrow}, m^{\uparrow\uparrow\nwarrow}\}$ is full.
\infl\ moves the module at $m^{\uparrow}$ to $m^{\nwarrow}$ via a series of operations, returning such module. 

\later{
We distinguish between two cases.
Let $m'$, $s$ and $s'$ be the modules at $m^\uparrow$, $m^{\uparrow\nearrow}$ and $m^{\uparrow\uparrow\nwarrow}$ if they exist.
First assume that either $m^{\uparrow\nwarrow}$ or $m^{\uparrow\nearrow}$ are occupied; see Figure~\ref{fig:inflate}~(a).
Then $M_1$ be the maximal descending path ending at $m^{\uparrow\nwarrow}$, and $M_2$ be the maximal ascending path ending at $m$.
Perform (i) \shift$(M_2, $ccw$)$, \shift$(M_1, $cw$)$,  pivot $m'$ ccw.
Then (ii) \shift$(M_1^{-1}, $ccw$)$, pivot $m'$ ccw, and finally \shift$(M_2^{-1}, $cw$)$.
Note that $M_1$ might be empty as shown in Figure~\ref{fig:inflate}~(c).
Now, assume that both $m^{\uparrow\nwarrow}$ and $m^{\uparrow\nearrow}$ are occupied; see Figure~\ref{fig:inflate}~(b).
Let $m''$ be the module at $m^{\uparrow\uparrow}$.
The operation is similar to the previous case and we use the same definitions (note that $M_1$ is empty).
The only difference is that we pivot $m''$ cw between the moves of $m'$, i.e., after (i) we add the additional move and proceed with (ii). 
If after the second case $s'$ is a cut vertex, it was part of a trivial 2-cut. Then there is a path of length 1 or 2 sticking out of $s'$. 
pivot cw the leaf of this path, restoring 2-connectivity.
If after the first case $s$ has degree 1, pivot it ccw.
If that doesn't restore 2-connectivity, then $s'$ does not exist, i.e., the initial position $m^{\uparrow\uparrow\nwarrow}$ is empty.
Then we replace (ii) by pivot $m'$ ccw, pivot $s$ ccw, \shift$(M_1^{-1}, $ccw$)$, and finally \shift$(M_2^{-1}, $cw$)$.
}

\begin{figure}[ht]
    \centering
    \includegraphics[scale=0.25]{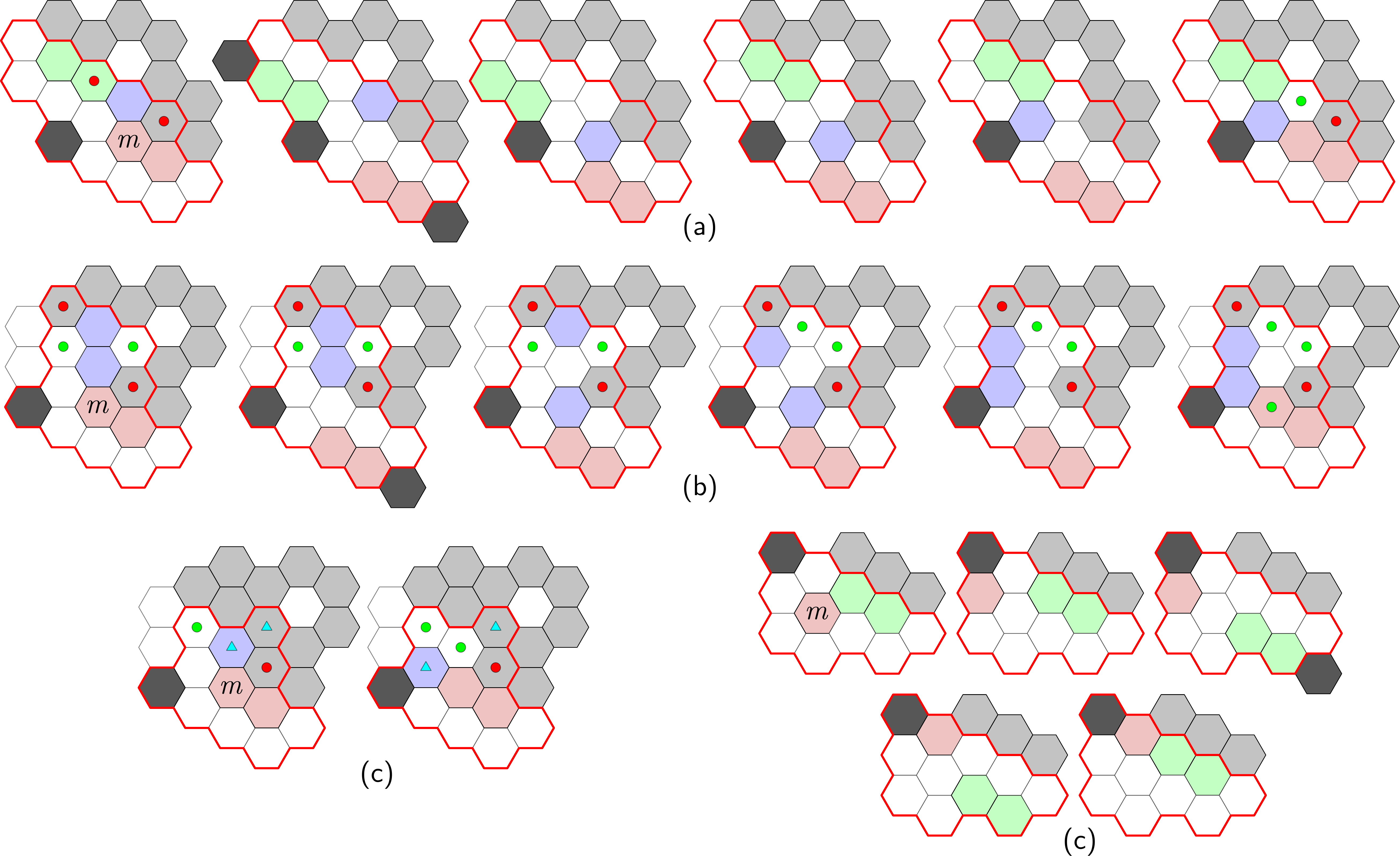}
    \caption{Operations used in \lbridge.}
    \label{fig:inflate}
\end{figure}

\both{
\begin{lemma}
\label{lem:inflate}
\infl$(m)$ performs $O(|V(\ell|)$ moves, does not break connectivity and preserves 2-connectivity of~$\ell$.
\end{lemma}
}

\later{
\begin{proof}
We first assume there are no adjacent 2-cuts in the vicinity.
Then we analyse all the edge cases depending on the presence of each possible adjacent 2-cut.
The size of the {\shift} operations is $O(|V(\ell|)$.
The first {\shift} operations satisfy $(S)$ and, by Lemma~\ref{lem:shift}, maintain connectivity.
The second may not satisfy $(S)$ as $\ell$ is not 2-connected anymore. 
However, since $m^{\nearrow}$ is full, no pocket was destroyed and, so it does not affects the second {\shift} because they can't both move modules of the same adjacent 2-cut, i.e., there is a 2-connected subgraph of $\ell$ containing $M_1$ in which $M_1$ satisfies $(S)$. 
Then $m'$ becomes movable or else it was a cut vertex from the start contradicting $\ell$'s 2-connectivity.
Although the next 2 shifts do not satisfy $(S)$, they are the reverse of the previous shifts and, therefore preserve connectivity.
The only module that effectively moved, comparing the initial and final configurations, is $m'$ (and $m''$ if applicable).
Since all pockets are preserved the only possible affected 2-cuts are adjacent 2-cuts containing $m'$ (and $m''$ if applicable).
Assuming there were no such adjacent 2-cuts, the proof is complete. 
Now assume that $\{m', s\}$ is a trivial 2-cut.
Then, $m$ is in the child 2-split component and there can be only another module in the same component.
Then, without any modification, the operation results in a path from $m'$ to $s$ using the modules in the child component and there are cut vertices since the other end of the path is attached to $\overline{\ell}$ and the path is effectively a bridge.
Now assume that $\{m', m^\nearrow\}$ was originally a trivial 2-cut.
Then, $s$ is the only other modules that can be in the child 2-split component and the normal operation would leave $s$ a degree-1 vertex. However, the modification at the end of the description of the operation restores 2-connectivity.
Now assume that $\{s', m'\}$ is initially a trivial 2-cut.
Note that this is the last case since if $m'$ is in any other adjacent 2-cut it is also in $\{s', m'\}$ due to the limited options of neighbors.
Then the operation leaves a path hanging from $s'$ as described in the operation. The vertex or both vertices of this path were adjacent to $m'$.
Then, the leaf of this path has enough space to move to occupy $m'$ original position, restoring 2-connectivity.
\end{proof}
}

\later{}
\medskip\noindent\textbf{\lbridge$(m)$.}
This operation is used when there is an opportunity to create a bridge when $m$ either gets blocked by $m^*$ on its way to the top of $\rho_0$ or it reaches the top and it would jump to $\overline{\ell}$.
We require that, if a cw pivot brings $m$ to $\rho_{-1}$, then $m^\nearrow$ is full and at least one position in $\{m^{\uparrow\nwarrow}, m^{\uparrow\nearrow}, m^{\uparrow\uparrow\nwarrow}\}$ is full.
We recurse in a child component, calling \destr\, if there is an adjacent nontrivial 2-cut forbidden by {\infl}.
That guarantees that we can apply {\infl} which would create a bridge from $\ell$.
If a cw pivot maintains $m$ in $\rho_{0}$, then it would land on a module $m^*\in \overline{\ell}$. We require that the maximal ascending path $M$ ending in $m^{\uparrow}$ can be shifted down by \shift$(M, ccw)$, i.e., $\rho_{0}$ must contain only $m$ below $M$; see Figure~\ref{fig:inflate}~(c).
Then, we ``squeeze'' $m$ in the space between $m^*$ and $M$ creating a bridge.
We do that by moving $m$ out of $\ell$, \shift\ $M$ down, moving $m$ back and \shift\ $M$ back.

\later{

\renewcommand{\emph}{}
\smallskip
\begin{algorithm}[H]
\label{alg:local-bridge}
\SetAlgoLined
\uIf{the next cw-pivot brings $m$ to $\rho_{-1}$ and $m^\nearrow$ is full}{
    \uIf{$\{m^{\uparrow}, m^{\uparrow\nearrow}\}$ or $\{m^{\uparrow\uparrow}, m^{\uparrow\uparrow\nwarrow}\}$ is a nontrivial 2-cut}{
        Let $\ell'$ be the child 2-split component of the lowest such 2-cut\;
        $c' := $ \destr$(\ell')$\;
        \uIf{$c'$ bridges between $\ell$ and $\overline{\ell}$}{
            Return $c'$;
            \Comment{$c'$ already merges $\ell$ into another block.}\\
        }
    }
    Return \infl$(m)$\;
}

\uElseIf{a cw pivot maintains $m$ in $\rho_0$}{
    Let $M$ be the maximal ascending path ending in $m^{\uparrow}$\;\label{lb:case2}
    \uIf{$M$ is not adjacent to any module in $\rho_0$ except for $m$}{\label{lb:squeeze}
        Pivot $m$ cw, \shift$(M, ccw)$, pivot $m$ ccw, then \shift$(M,cw)$\;
        Return $m$;
        \Comment{$m$ now bridges from $\ell$.}\\
    }
}

\caption{\lbridge($m$)}
\end{algorithm}
\smallskip
\renewcommand{\emph}{\textbf}
}

\both{
\begin{lemma}
\label{lem:lbridge}
\lbridge($m$) bridges from $\ell$, does not break connectivity, and preserves 2-connectivity of $\ell$. It uses $O(|V(\ell)|) + T_m(|V(\ell')|)$ moves where $T_m(|V(\ell')|)$ is the number of operations performed by \destr\ in~$\ell'$.
\end{lemma}
}

\later{
\begin{proof}
If all requirements of the called operations are satisfied, then \lbridge\ successfully returns a bridge.
We now check such requirements. 
The requirements for the case in line~\ref{lb:case2} are already guaranteed by the requirements of \lbridge.
We argue for the other case.
Note that we already require that the appropriate positions are full. 
The only thing we have to check is whether a forbidden 2-cut exists. 
\lbridge\ applies \destr\ to only one such a 2-cut $\{v_1, v_2\}$.
Note that, if both of these 2-cuts are nontrivial, then \infl\ will execute its case 1 and we don't care about the upper 2-cut.
If there is a 2-cut, the procedure calls \destr\ once on $\ell'$, and by Lemma~\ref{lem:inflate} a bridge can be created.
\end{proof}
}

\later{}
\medskip\noindent\textbf{\inc$(m)$.}
Whenever a local bridge was not possible\iffull in $m$'s ascending journey\fi , this operation either incorporates $m$ into $\rho_1$ or leaves $m$ attached to a module in $\overline{\ell}$ with the promise that some module will ascend in $\rho_0$ and bridge (Figure~\ref{fig:incorporate}~(d)).
There are four cases.
{}
In case 1, we check if we can call \deflate\ at position $m^{\uparrow\nearrow}$. In the positive case, we move $m$ and a possible neighbor $m'$ in $\rho_0$ out of the way, call \deflate, and move $m$ and possibly $m'$ back (Figure~\ref{fig:incorporate}~(a)--(b)).
In case 2, $m^{\nearrow}$ is empty and $m^{\uparrow\nearrow}$ is full.
Then, we ``squeeze'' $m$ into $m^{\nearrow}$ by using \shift\ operations, similar to \bridge\ (Figure~\ref{fig:incorporate}~(c)).
In case 3, if we pivot $m$ cw, that brings $m$ to $\rho_1$ and makes its degree 1.
Then, we apply some local movements in order to incorporate $m$ into $\rho_1$ while maintaining 2-connectivity (Figures~\ref{fig:incorporate}~(e) and (g)).
In case 4, we are not in the previous cases and we simply pivot $m$ cw.
Note that $m$ might leave $\ell$ (Figure~\ref{fig:incorporate}~(d)).
We explore this case from now.
As shown in the proof of Lemma~\ref{lem:inc}, there is a guarantee that a subsequent module $s$ in $\rho_0$ will ascend.
There are three possible cases, either (i) $s$ creates a bridge using $m$ (as in Figure~\ref{fig:incorporate}~(d)), in which case nothing need to be done; (ii) $s$ calls \lbridge\ or \bridge.
Then, pivot $m$ twice counterclockwise before \bridge\ or after \lbridge;
or (iii) $s$ calls \inc. Then, there is either another module in $\rho_0$ or we can move $s$ back to $\ell$ and apply \lbridge.

\begin{figure}[ht]
    \centering
    \includegraphics[scale=0.25]{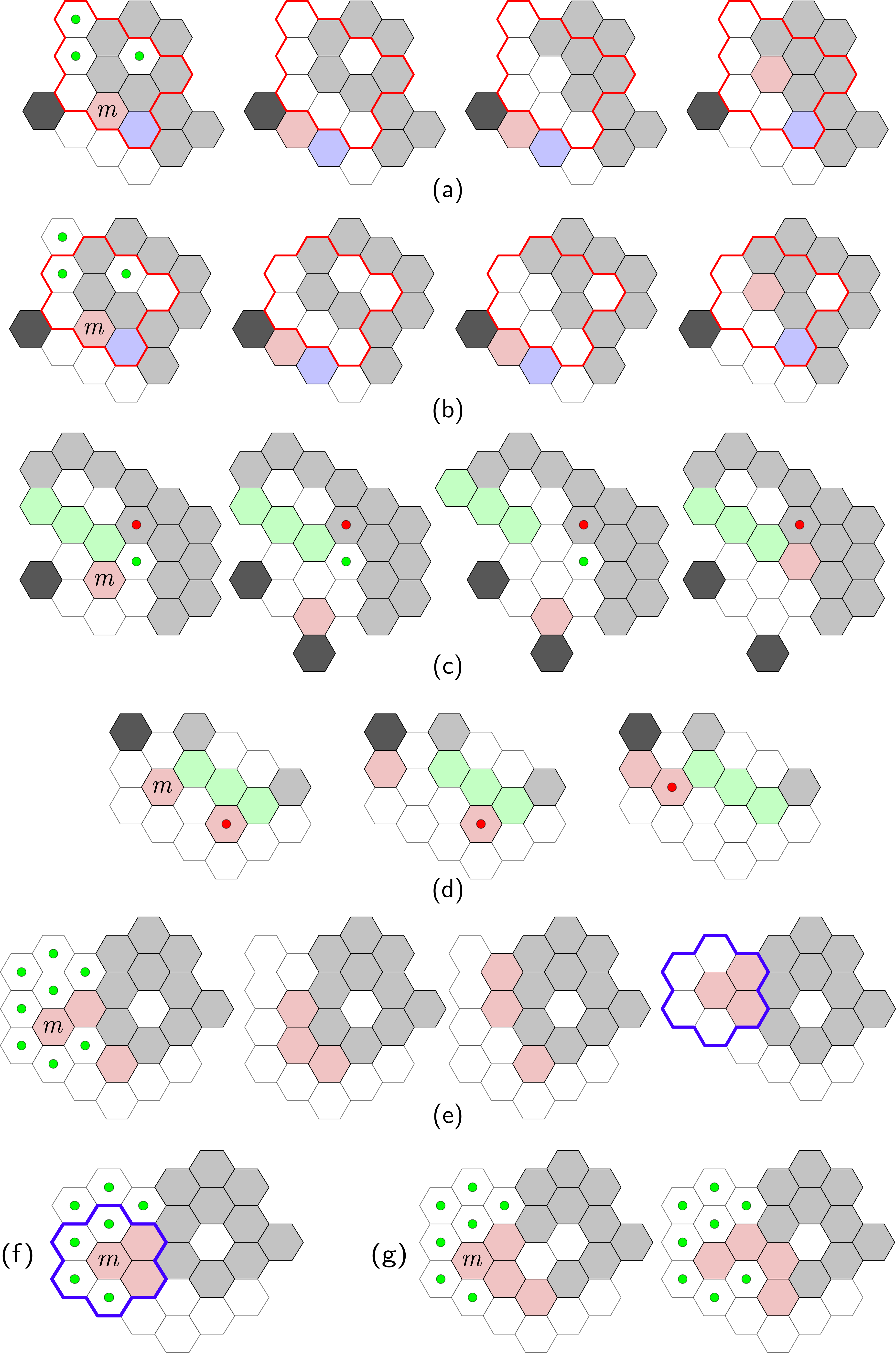}
    \caption{Operations used in \inc.}
    \label{fig:incorporate}
\end{figure}

\later{

\renewcommand{\emph}{}
\smallskip
\begin{algorithm}[H]
\label{alg:inc}
\SetAlgoLined
\uIf{$m^{\uparrow\nearrow}$ is an empty position enclosed by $\ell$ and can be deflated after deleting $\rho_0$}{\label{inc:case1}
    Let $m'$ be the module at $m^\searrow$ if it is full\;
    Pivot $m$ (and $m'$ if it exists) ccw out of the way;
    \deflate$(m^{\uparrow\nearrow})$;
    \Comment{Leaving $m^{\uparrow}$ empty (Figure~\ref{fig:incorporate}~(a)--(b)).}\\
    Move $m'$ (if applicable) and $m$ cw to bring them back\;
}
\uElseIf{$m^{\nearrow}$ is empty and $m^{\uparrow\nearrow}$ is full}{\label{inc:case2}
    Let $M$ be the maximal descending path ending at $m^\uparrow$\;
    Pivot $m$ ccw twice, \shift$(M, \text{cw})$, pivot $m$ cw twice, then \shift$(M^{-1}, \text{ccw})$; \label{inc:noDef}
    \Comment{This integrates $m$ into $\rho_1$ (Figure~\ref{fig:incorporate}~(c)).}\\
}
\uElseIf{a cw pivot would leave $m$ degree-1 in $\rho_1$}
    {\label{inc:case3}
        \uIf{$m^{\searrow}$ is full}{\label{inc:def}
            Let $p = m^{\searrow\nearrow}$\;
            Move $m$ out of the way, \deflate$(p)$, move $m$ back;
            \Comment{Figure~\ref{fig:incorporate}~(g).}\\
        }
        Let $m'$ the module adjacent to $m$\; \label{inc:restore2-con}
        Pivot $m$ ccw, then pivot $m'$ and $m$ cw;
        \Comment{Figure~\ref{fig:incorporate}~(e).}\\
    }
\Else{
    Pivot $m$ cw\;
}
\caption{\inc($m$)}
\end{algorithm}
\smallskip
\renewcommand{\emph}{\textbf}
}

\both{
\begin{lemma}
\label{lem:inc}
\inc($m$) uses $O(|V(\ell)|)$ moves and brings $m$ to $\rho_1$ in every situation that $m$ would go to $\rho_{-1}$ by pivoting cw to which \lbridge\ does not apply.
It does not break connectivity and maintains 2-connectivity of $\ell$.
Any created degree-1 module outside $\ell$ can be reincorporated in $\ell$, thus, no new block is created.
\end{lemma}
}

\later{
\begin{proof}
The operation has $3$ main cases we will call cases 1, 2 and 3 (resp., lines~\ref{inc:case1}, \ref{inc:case2} and \ref{inc:case3}).
We show that cases 1 and 2 cover all situations not contemplated by \lbridge\ when $m$ would go to $\rho_{-1}$.
There are two possibilities.
Either all three positions in $\{m^{\uparrow\nearrow}, m^{\uparrow\nwarrow}, m^{\uparrow\uparrow\nwarrow}\}$ are empty or $m^\nearrow$ is empty.
In the former case, deflate is always possible as shown in Figure~\ref{fig:incorporate}~(a)--(b).
Note that $m$ is 2-free and once it moves away, $m'$ becomes free by the 2-connectivity of $\ell$.
In the latter case we apply case 2.
Then, case 3 includes the cases when $m$ would remain in $\rho_0$ or would go to $\rho_1$.
We first address the second case.
We have to show that the conditions for \deflate\ are met.
Then, by Lemma~\ref{lem:deflate}, connectivity is not broken and the result is 2-connected.
If $p$ is empty, then the empty spaces guaranteed by the definition of the case are enough for the requirements of \deflate\ (Figure~\ref{fig:incorporate}~(g)).
Note that $m^{\downarrow\searrow}$ is also empty because $m$ is in an extreme path.
If $p$ is full, then the flower centered in $m$ is valid (Figure~\ref{fig:incorporate}~(f)) and \destr\ would not have called \inc.
Now we address the case when $m$ would remain in $\rho_0$.
The operation just pivots $m$ to $m^*$. Note that that preserves 2-connectivity in $\ell$, but $m$ might have degree 1.
If there were no other modules in $\rho_0$, then the case fits the conditions in \lbridge\ and \destr\ would have called \lbridge\ instead of \inc.
Then, a new module $s$ ascends in $\rho_0$ attached to $m$ by a path in $\rho_1$.
If (i) $s$ bridges with $m$, the claim is clearly true.
If (ii) it bridges somewhere else in which case the operation returns $m$ to $\ell$. In particular, after \lbridge\ or before \bridge\ the vicinity of $m$ was not affected before $s$ is below the $m$'s initial position, or else $s$ would have reached $m$.
Then, we can pivot $m$ ccw at least twice maintaining $m$ in $\rho_0$.
In particular, at such position $m$ is adjacent to 2 modules in $\ell$ or else the component of $\ell$ induced by $\rho_1$ adjacent to $m$ is a single module and the conditions for \lbridge\ would have been satisfied.
Then, $m$ is in the same block as these 2 modules.
If (iii) $s$ calls \inc\ and is the last module in $\rho$.
Then the conditions of \lbridge\ are satisfied if we roll $m$ back to $\ell$ since $M$ in line~\ref{lb:squeeze} of \lbridge\ would not have any bottom neighbor.

The bottleneck of this sub-procedure is case 2 and the \shift\ operations make $O(|V(\ell)|)$ moves.
\shift$(M, \text{cw})$ satisfies $(S)$ by the maximality of $m$ in $\rho_0$.
So it and its reversal do not break connectivity by Lemma~\ref{lem:shift}.
Lemmas~\ref{lem:deflate} and \ref{lem:shift} complete the proof for cases 1 and 2.
In line~\ref{inc:def}, the conditions of deflate are met by the fact that $m$ would become degree-1 in $\rho_1$.
In line~\ref{inc:restore2-con}, 
Figure~\ref{fig:incorporate}~(e) shows the positions that must be empty because $m$ would become degree-1 in $\rho_1$. Then, $m'$ is 2-free and has the space to pivot down and $m$ takes its place. This restores 2-connectivity.
\end{proof}
}

\later{}

\iffull\newpage\fi
\medskip\noindent\textbf{Analysis.}
\iffull
We now ready to analyse \destr. 
\fi

\both{
\begin{lemma}
\label{lem:destr}
If $\ell\neq G$ is a leaf block of $\mathcal{B}$, \destr$(\ell)$ performs $O(|V(\ell)|^2)$ moves merging $\ell$ and a subset of nodes of $\mathcal{B}$ into a single block while not creating any other new blocks.
\end{lemma}
}

\later{
\begin{proof}
We begin with the correctness proof by induction. 
The base case of the recursion is when $|V(\ell)| = 3$ if $\ell$ is a block and $|V(\ell)| = 5$ if $\ell$ is a 2-split component since we don't recurse on trivial children.
If $|V(\ell)| = 3$, $\ell$ induces a triangle with enough empty spaces around it for a valid flower. Then the algorithm calls \bridge\ and terminates.
If $|V(\ell)| = 5$ there are only two options.
Either there is a valid flower as the previous case, or the $3$ children modules are adjacent to the 2-cut and lie in the same row $\rho_0$.
Then, a normal run of the algorithm will either call \lbridge\ or from a valid flower since we have 3 2-free modules in $\rho_0$ that can potentially ascend.
We proceed with the induction step.
If $m$ is not 2-free because it is part of an adjacent 2-cut, then, by Observation~\ref{obs:adj2cut} and the induction hypothesis, $m$ becomes 2-free after the recursive call in line~\ref{merge:recursion}.
Else, \bbup\ brings an empty position up. 
All other operations either destroy empty spaces or creates bridges, in which case we return; never an empty space goes down.
Each position can only travel at most $|V(\ell)|$ times.
\deflate\ destroy an empty position and creates a 2-free module if the empty space is surrounded by modules.
There can be up to $|V(\ell)|$ empty spaces in each row.
Then, there will be a 2-free ascending module.

We show that: $(\star)$ only up to $3$ 2-free ascending modules can pass through the same section of perimeter of $\ell$ until we return a bridge.
For that, we need to look at the situations we apply \inc, since it's the only case in which we don't return.
We use the cases defined in the proof of Lemma\ref{lem:inc}.
First assume that an ascending module $m$ gets stopped by a module $m^*\in \overline{\ell}$ in $\rho_{-1}$ (cases 1 and 2).
If the next ascending module $s$ in $\rho_0$ reaches the initial position of $m$, then it will also be blocked by $m^*$. 
However, in each case, the presence of $m$ creates conditions sufficient for us to call \lbridge.
In both cases, when $s$ reaches its highest position, $s^\nearrow$ and $s^{\uparrow\nearrow}$ are full.
If there is no other ascending module in $\rho_0$, the next ascending 2-free module that passes through $m^*$ in $\rho_1$ will be $m$.
That is because $m$ was 2-free and it either remains 2-free, or it is between modules in $\rho_1$ but is not part of a 2-cut and its top neighbor is full.
In the first case, $m$ continues its ascension from where it left of.
In the second case, it will not be part of any \bbup\ or \deflate\ operations of preceding ascending modules in $\rho_1$ and, when it becomes an ascending module in $\rho_1$ it is 2-free and can continue its ascension from where it left of.

Now assume that $m$ would leave $\ell$ to a position in $\rho_0$ (case 3). 
Refer to Figure~\ref{fig:incorporate}~(d). 
Since we did not apply \lbridge, there is another module $s$ in $\rho_0$ that will become the next ascending module, then, if it reaches the same position once occupied by $m$, it will connect $m$ to the rest of $\ell$ as required.
Else, it would not reach $m$'s initial position and charge the same section of perimeter. 
By the proof of Lemma~\ref{lem:inc}, there is a guarantee that the algorithm will bridge while it processes $\rho_0$.

Finally, assume that $m$ rises to $\rho_1$ after its ascension by case 3. 
As before, a subsequent ascending module $s$ that reaches $m$'s initial position will also rise up to a position adjacent to the current position of $m$.
Then, it will either be the next module in $\rho_1$ beside $m$, or $m$ will move up by case 3.
In the first case, when the third ascending module comes along, it will pass through the position below both $m$ and $s$ forming a valid flower.
In the second case, it will rise to $\rho_1$ as in Figure~\ref{fig:incorporate}~(e).
This concludes the correctness proof. 

We now proceed with the analysis of the number of moves.
Every move in the ascension of $m$ when $m$ is 2-free can be charged for a section of the original perimeter of $\ell$.
By $(\star)$, a section of perimeter is never charged more than 3 times.
Note that \deflate\ and \bbup\ can only pass through a section of the perimeter once until a 2-free ascending module $m$ is created.
Similar as in the proof of $(\star)$, if a subsequent ascending module $s$ reaches the region where a \deflate\ and \bbup\ were performed before the creation of $m$, then it can also ascend to be adjacent to $m$.
The number of moves in a recursion level is then proportional to the perimeter of $\ell$ which is $O(|V(\ell)|)$.
Every operation called makes $O(|V(\ell)|)$ moves except for the recursive calls and \lbridge.
It is clear that we return after \lbridge\ is called.
We claim that \destr\ does at most one recursive call.
Fist, assume that the recursive call was because $\{m, m^\nearrow\}$ is a nontrivial 2-cut.
Then, $m^{\nearrow\nearrow}$ is empty and in the outer face of $\ell$.
It follows that $m^{\searrow}$ and $m^{\searrow\searrow}$ are full and also in adjacent 2-cuts because the boundary of $\ell$ is being ``pinched'' between $m^{\nearrow\nearrow}$ and $\rho_{-1}$.
After the recursive call returns, there is a new path from the child of $\{m, m^\nearrow\}$.
Then, either $m^{\searrow}$ and $m^{\searrow\searrow}$ are not in 2-cuts anymore, or they are in trivial 2-cuts.
By the proof of $(\star)$, the algorithm either returns before the ascension of $m^{\searrow\searrow}$ or the $3$ modules will form a valid flower and the algorithm calls \bridge\ and there are no more ascending modules.
Then the worst case bound in the number of total moves $T_m(V|(\ell)|)$ is given by:
\[T_m(V|(\ell)|) = O(V|(\ell)|) + T_m(V|(\ell')|) + T_m(V|(\ell'')|)\]
where $\ell'$ and $\ell''$ are the split components that get recursed on in \destr\ and \lbridge.
Since $\ell'$ and $\ell''$ are smaller than $\ell$, the recursion solves to $O(V|(\ell)|^2)$.
\end{proof}
}

\both{
\begin{corollary}\label{lem:phase2-works}
$G$ can be made 2-connected in $O(n^3)$ moves.
\end{corollary}
}

\later{
\begin{proof}
While $|\mathcal{B}| > 1$, Phase 2 finds a leaf $\ell \in \mathcal{B}$ and performs \destr($\ell$). 
By Lemma~\ref{lem:destr}, this strictly decreases the number of blocks in $\mathcal{B}$. Since $|\mathcal{B}| \leq n$, after at most $n$ iterations $\mathcal{B}$ will contain a single block, i.e., $G$ is 2-connected. Since $\forall \ell \in G$, $|\ell| \leq n$, by Lemma~\ref{lem:destr}, each iteration takes at most $O(n^2)$ moves.
\end{proof}
}

\subsection{Phase 3: Building the canonical path $P$.}
\label{sec:phase3}

In the final phase, we will show that once the configuration is 2-connected, we can start moving modules onto the end of our path $P$ at a cost of $O(n^2)$ moves per module. 

\begin{lemma}
\label{lem:phase3}
If $G$ is 2-connected, in $O(n^2)$ moves we can produce a 2-free module on an extreme path of $G$ while maintaining the 2-connectivity of $G$.
\end{lemma}
\begin{proof}
We apply a subset of operation \destr\ to $G$.
Then, this proof becomes a special case of Lemma~\ref{lem:destr}, where $\ell=G$. 
In \destr, our goal is to bridge between $\ell$ and $\overline{\ell}$ maintaining $\ell$ 2-connected.
Here, $\overline{\ell}$ is empty and \lbridge\ will never be called since there are no obstacles for ascending modules.
Then, we are always able to produce a \crew\ in an extremal position.
Moving the \crew\ to $P$ does not affect 2-connectivity of $G$ by definition.
\end{proof}

Our main theorem as a direct consequence of Lemma~\ref{lem:phase1}, Corollary~\ref{lem:phase2-works} and Lemma~\ref{lem:phase3}.

\begin{theorem}
\label{thm:alg}
Any connected configuration of $n$ hexagonal modular robots can be reconfigured to any other with $O(n^3)$ pivoting moves in the Monkey model, while maintaining connectivity. 
\end{theorem}
\iffull
\begin{proof}
A direct consequence of Lemma~\ref{lem:phase3} is that given a configuration inducing a 2-connected \cg, $O(m^3)$ moves in the Monkey model are sufficient to transform it into a canonical path $P$.
Then, by Lemma~\ref{lem:phase1} and Corollary~\ref{lem:phase2-works}, within $O(m^3)$ moves we can go from any configuration to $P$.
Since the moves are reversible, we can get reconfigurations between any two configuration.
\end{proof}
\fi

\section{PSPACE-hardness reductions}
\label{sec:PSPACE}

In this section we show PSPACE-hardness for all other models. Our reduction follows the framework introduced in~\cite{motionplanning2}. We reduce from a reachability problem: given an agent that moves along a graph-like structure whose traversability changes in response to the agent's actions, is there a series of moves which takes the agent from a start to a target location.

\begin{theorem}\label{theo_hexrestrictedhard}
Given two configurations of $n$ hexagonal modules, it is PSPACE-hard to determine if we can reconfigure from one to the other using only restricted moves.
\end{theorem}

In Section~\ref{sec:prelim} we describe the reachability problem introduced in \cite{motionplanning2} and the pieces we need to simulate to create the reduction. We introduce a few modifications to this problem and show it remains PSPACE-hard in Section~\ref{sec:Balanced}. In Section~\ref{hexmodel} we discuss how to simulate each of the gadgets with hexagonal modules. Reduction for other models are in Section~\ref{sec_squarereduc}.

\subsection{Preliminaries} \label{sec:prelim}

We reduce from a variation of the $1$-player motion planning with the locking $2$-toggle (L2T) \cite{motionplanning2}. This restricted variant is called 1-toggle-protected motion planning with the locking 2-toggle and described in Section~\ref{sec:Balanced}. In the \emph{1-player motion planning} problem we want to decide whether an agent has a series of moves which will take it to a target location. The constructs we use in this problem are \emph{gadgets} which have \emph{locations} (entrances and exits), \emph{states}, and \emph{transitions}. The agent is always at some location. Transitions are an ordered pair of state and location pairs. If an agent is at some specific gadget location and the gadget is in a state matching the first pair, then the agent can move to the location in the second pair which changes the state of the gadget to the state in the second pair (see Figure~\ref{fig:basic gadgets}). A \emph{system of gadgets} is a set of gadgets and \emph{connections} between locations in those gadgets. The agent can freely move between locations which have connections. Some gadgets transitions form a matching - we call these matched pairs \emph{tunnels}.


\begin{figure}
  \centering
  \begin{subfigure}{0.5\textwidth}
    \centering
    \includegraphics[scale=0.75]{PL2T_labeled}
    \caption{The locking 2-toggle gadget (L2T). \iffull In the top state 3, you can traverse either tunnel going down, which blocks off the other tunnel until you reverse the initial traversal.\fi}
    \label{fig:L2T}
  \end{subfigure}\hfil\hfil
  \begin{subfigure}{0.4\textwidth}
    \centering
    \includegraphics[scale=0.75]{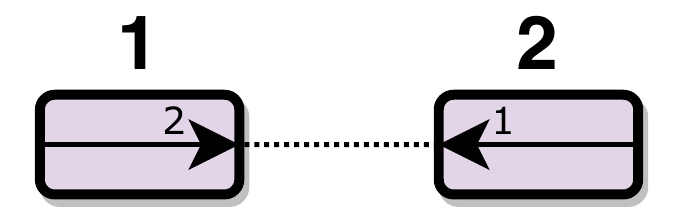}
    \caption{ The 1-toggle gadget. \iffull Traversing the tunnel reverses the direction that it can be traversed.\fi}
    \label{fig:1-toggle}
  \end{subfigure}
  \caption{Examples of reversible, deterministic gadgets. Purple boxes are states of the gadget, labeled with a number outside the box. Transitions are arrows from one location to another with a small number indicating the state it transitions to. Dotted lines help visualize which states are connected by transitions in the gadget.}
  \label{fig:basic gadgets}
\end{figure}



In order for us to reduce from this problem, we need to represent the agent, the gadgets (specifically a locking 2-toggle and a branching hallway gadget), connections between locations, and a goal location with modular robots. 
In order to reduce from 1-toggle-protected motion planning with the locking 2-toggle we need to create the following constructions: 

\begin{itemize}
    \item \emph{Wires} which allow the modules to travel between parts of the configuration simulating the connection graph edges which allow the agent to travel between locations.
    \item \emph{Branching hallways} which connect three wires together and allow the modules to travel down any of them.
    \item \emph{Locking 2-toggle} which is a 3 state, 4 location gadget shown in Figure~\ref{fig:L2T}. The gadget has two tunnels which are both traversable in state 3. After taking either transition, the only option is returning back and restoring the gadget to it's prior state.
    \item \emph{Win gadget} which can only be reconfigured if two additional modules reach it, simulating the goal location in the motion planning problem.
\end{itemize}


\iffull
In the following sections we explain how to build each of these constructions. Our four reductions show that deciding whether one can reconfigure between two module configurations is PSPACE-complete in the remaining four models. 

\fi

\subsection{1-toggle-protected motion planning}
\label{sec:Balanced}

In this section we strengthen the result from \cite{motionplanning2} to show that motion planning with reversible, deterministic gadgets with interacting tunnels is PSPACE-complete even when connections can only be traversed as though they are 1-toggles. We will consider only branchless systems of gadgets, but we will allow the branching hallway gadget. \iffull This coincides with the original model in \cite{Toggles_FUN2018}. \fi In a \emph{branchless} system of gadgets, the connections between locations form a matching \cite{iogadgets}. The \emph{branching hallway gadget} is a 1-state, 3-location gadget with traversals among all three pairs of locations.

An instance of \emph{1-toggle-protected motion planning} with a set of gadgets $G$ is an instance of branchless $1$-player motion planning with $G$ as well as the branching hallway gadget and the 1-toggle, where one end of every connection is a location on a 1-toggle. Intuitively, this requires that every edge in the connection graph acts as a 1-toggle.

\later{


We will now show that, with reversible deterministic gadgets, the motion planning problem remains hard even in the 1-toggle-protected case.

A  gadget is \emph{reversible} if for every transition there exists a transition from the destination state and location to the prior state and location. A gadget is \emph{deterministic} if every state and location pair has at most one transition from it.

First notice that a system of gadgets formed by connecting a locking 2-toggle or a the 1-toggle to another 1-toggle with a matching orientation behaves the same as the original 2-toggle or 1-toggle. In general we call gadgets which have the same behavior when connected to 1-toggles \emph{1-toggle agnostic}.

}

\both{
\begin{theorem}\label{thm:balanced_pspace}
1-toggle-protected planar 1-player motion planning problem with a reversible, deterministic, on-tunnels gadget with interacting tunnels is PSPACE-complete.
\end{theorem}
}

\later{
\begin{proof}
To prove this result we will take the PSPACE-completeness construction from \cite{motionplanning2}, replace all branching connections with branching hallway gadgets, add 1-toggles between all connections, and then show that this new construction is still a correct reduction. Conveniently, this modification strictly reduces the agent's mobility with respect to the original reduction. Thus we only need to show that correct usage still exists after the modification.

The original reduction is from planar Non-deterministic Constraint Logic (NCL)~\cite{GPCBook09}. NCL is a reconfiguration problem on a weighted, directed graph where vertices have constraints which require a minimum amount of weight be pointed towards them at any given point in time. Moves consist of flipping the orientation of an edge, subject to the vertex constraints, and the problem asks if there exists a series of moves which allows a target edge to be flipped. This problem is PSPACE-complete and remains so even when the graph is planar, the edge weights are restricted to be either $1$ or $2$ (called ``red" and ``blue" edges respectively), the vertices all have the same constraint of at least $2$ total weight, and vertices have either exactly three adjacent weight $2$ edges or exactly one adjacent weight $2$ edge and two adjacent weight $1$ edges (called ``OR vertex" and ``AND vertex" respectively).

From the construction in \cite{motionplanning2} we will need to examine the crossover gadget, the edge gadget, the planar OR gadget, the planar AND gadget, and the structure for choosing which edge to flip.

First, observe that a branchless connection with one end connected to a locking 2-toggle will not change behavior if a 1-toggle is inserted into that connection. We align the 1-toggle with the tunnel of the locking 2-toggle to which it is connected. If we traverse that tunnel then both the direction of the 2-toggle and the 1-toggle flip, staying the same direction. If we cross the other tunnel in the locking 2-toggle, then we cannot cross the tunnel connected to the 1-toggle. When traversability is restored the 1-toggle and locking 2-toggle are once again pointed in the same direction.

If we have a degree 3 connection, this is replaced by a branching hallway connected to the three original locking 2-toggles. Again, these locking 2-toggles will each restrict the location they are connected to to either not be traversable or behave as a 1-toggle. Adding an aligned 1-toggle across these connections will not change the traversability.

Thus we only need to be concerned with higher degree connections between systems of gadgets where we must connect multiple branching hallways together. We will now analyze the parts of the reduction with connections of degree higher than three. 

\paragraph*{Locking 2-toggle self simulation}
In the planar setting, there are three types of locking 2-toggles based on whether the tunnels cross, run parallel, or run anti-parallel. One major aspect of planarizing is showing that the three planar types of locking 2-toggles---parallel, anti-parallel, and crossing---simulate each other planarly. These gadgets are shown in Figure~\ref{fig:L2T types}, and a cycle of simulations are show in Figures~\ref{fig:APL2T sim CL2T}, \ref{fig:CL2T sim PL2T}, and \ref{fig:PL2T sim APL2T 1-toggle}. The only constructions which need to be altered because of high degree connections are are Figure~\ref{fig:APL2T sim CL2T} and Figure~\ref{fig:PL2T sim APL2T 1-toggle}.

\begin{figure}
  \centering
  \begin{subfigure}{.3\textwidth}
      \centering
      \includegraphics[scale=.65]{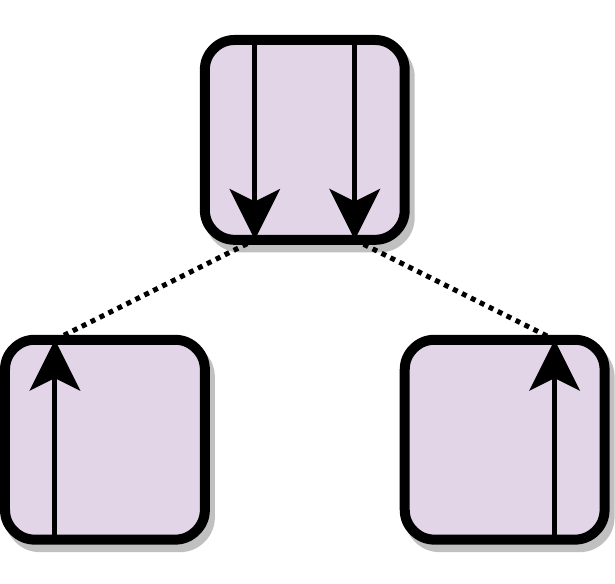}
      \caption{A parallel locking 2-toggle (PL2T).}
      \label{fig:PL2T}
  \end{subfigure}
  \begin{subfigure}{.3\textwidth}
      \centering
      \includegraphics[scale=.65]{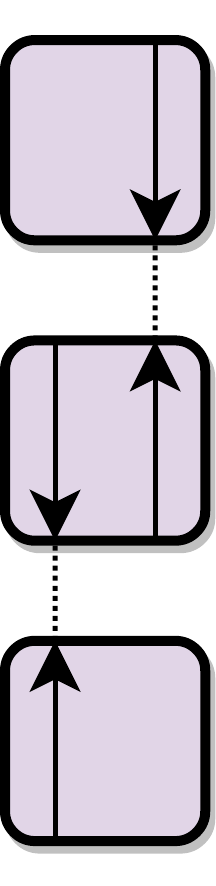}
      \caption{An antiparallel locking 2-toggle (APL2T).}
      \label{fig:APL2T}
  \end{subfigure}
  \begin{subfigure}{.3\textwidth}
      \centering
      \includegraphics[scale=.65]{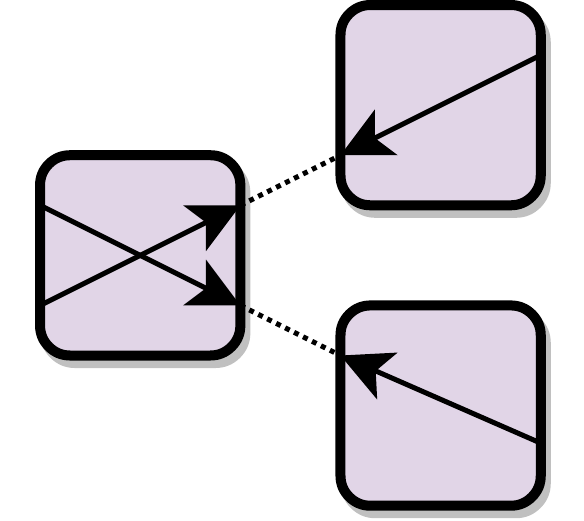}
      \caption{A crossing locking 2-toggle (CL2T).}
      \label{fig:CL2T}
  \end{subfigure}
  \caption{State diagrams of the three types of locking 2-toggles in planar mazes. Purple boxes are different states with arrows as transitions. Since the gadgets are reversible and deterministic, we can denote the state changes of the transitions with the dotted lines.}
  \label{fig:L2T types}
\end{figure}

\begin{figure}
  \centering
  \begin{minipage}{0.43\linewidth}
    \centering
    \includegraphics[scale=0.65]{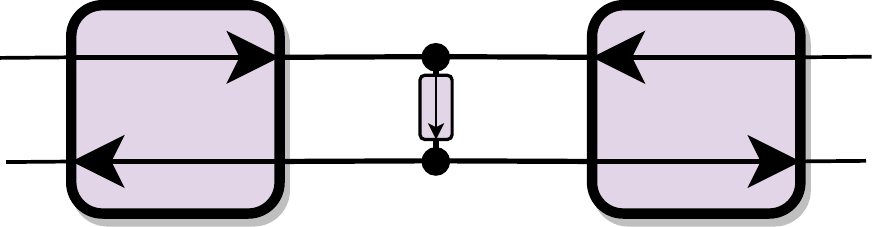}
    \caption{APL2T simulating CL2T.}
    \label{fig:APL2T sim CL2T}
  \end{minipage}\hfil\hfil
  \begin{minipage}{0.43\linewidth}
    \centering
    \includegraphics[scale=0.65]{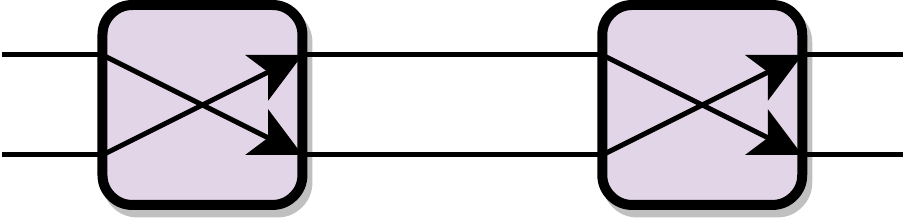}
    \caption{CL2T simulating PL2T. The single 1-toggle suffices for both of the upper to lower traversals in the simulation of CL2T by APL2T.  This figure shows the middle state with two transitions.}
    \label{fig:CL2T sim PL2T}
  \end{minipage}
\end{figure}



\begin{figure}
  \centering
  \includegraphics[scale=0.75]{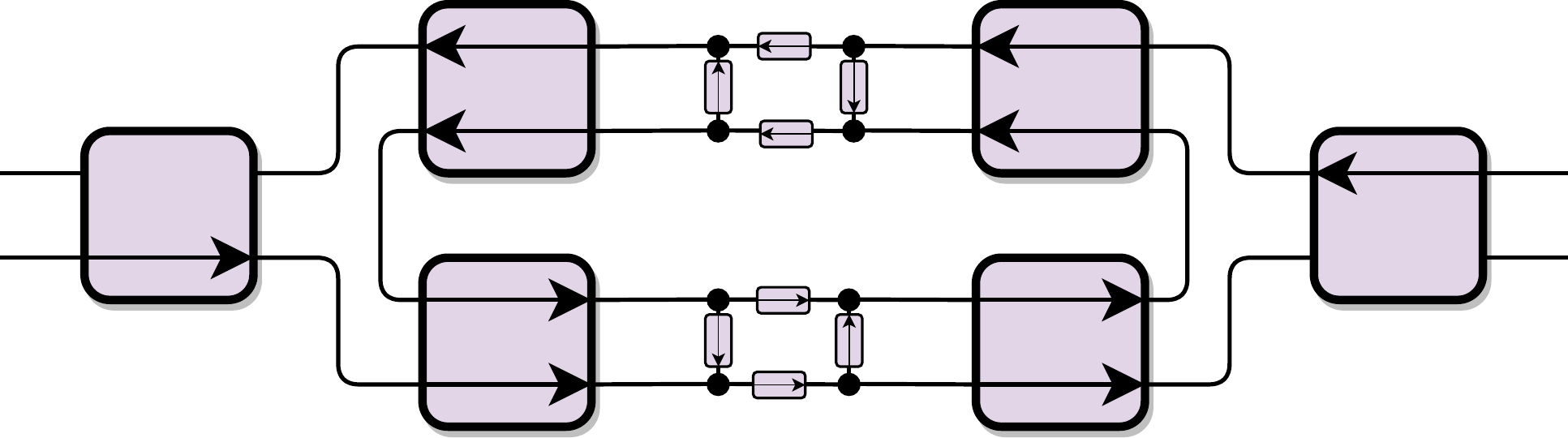}
  \caption{Inserting 1-toggles to separate the degree-4 connections in the simulation of CL2T by APL2T.  This figure shows the middle state with two transitions.}
  \label{fig:PL2T sim APL2T 1-toggle}
\end{figure}

\paragraph*{Crossover gadget}
Now that we have access to the three types of locking 2-toggle, the actual construction of the A/BA crossover does not require any connections higher than degree 3. The construction and behavior of an A/BA crossover is shown in Figure~\ref{fig:A/BA crossover}.

\begin{figure}
  \centering
  \begin{subfigure}{.45\textwidth}
    \centering
    \includegraphics[scale=.75]{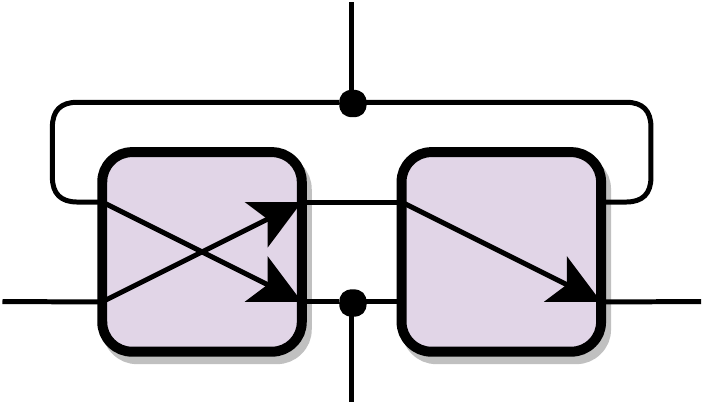}
    \caption{Simulating an A/BA crossover using CL2Ts.}
    \label{fig:CL2T sim A/BA}
  \end{subfigure}
  \begin{subfigure}{.45\textwidth}
    \centering
    \includegraphics[scale=.75]{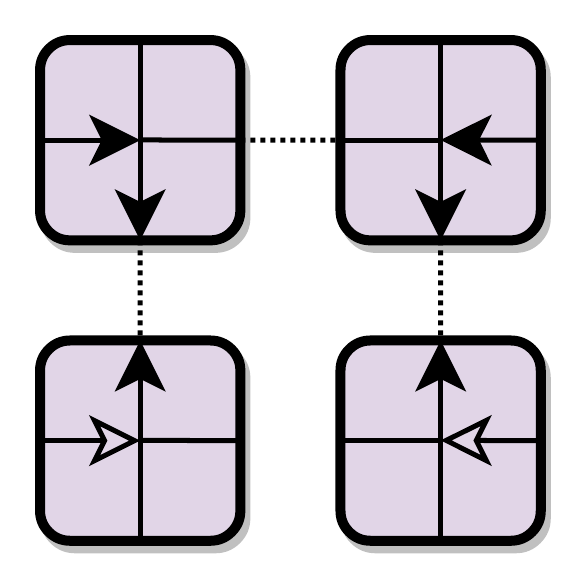}
    \caption{A state diagram and notation for the A/BA crossover.}
    \label{fig:A/BA notation}
  \end{subfigure}
  \caption{An A/BA crossover gadget: the agent can traverse top to bottom (A), or traverse left to right (B) and then top to bottom. Thinking of the gadget as a crossing pair of 1-toggles, the vertical 1-toggle is always traversable, and the horizontal 1-toggle is traversable when the vertical one is pointing down.}
  \label{fig:A/BA crossover}
\end{figure}

\paragraph*{Planar OR and AND gadgets}
The planar AND gadget, as seen in Figure~\ref{fig:planar NCL AND vertex} has connections of degree at most 3. 

The planar OR gadget, shown in Figure~\ref{fig:planar NCL OR vertex} is more complicated because we have multiple pairs of edges to choose to lock into place pointed into the vertex. However, at any point in time only one edge needs to be locked pointing towards the gadget giving us enough flexibility.

When A is locked, we have toggle $1$ pointed to the right and toggle $2$ pointed to the left and both toggles $3$ and $4$ pointed to the left. When C is locked, we still have toggles $3$ and $4$ pointed to the left but toggles $1$ and $2$ are reversed. The agent unlocking A and locking C will flip toggles $1$ and $2$ and will not effect toggles $3$ and $4$ giving us the configurations desired. Unlocking C and locking A is symmetric. 

When locking either A or C and unlocking B, we have a choice of exiting along the top path through toggles $1$ and $3$ or along the bottom path through toggles $2$ and $4$. Either pathway is fine, as the next time we visit the vertex we will begin by unlocking B by going over that same path. This leaves both $1$ and $2$ pointed outwards so exiting through either A or C is possible and leaves the system in the desired configuration.

\begin{figure}
  \centering
  \begin{subfigure}{.4\textwidth}
    \centering
    \includegraphics[scale=.74]{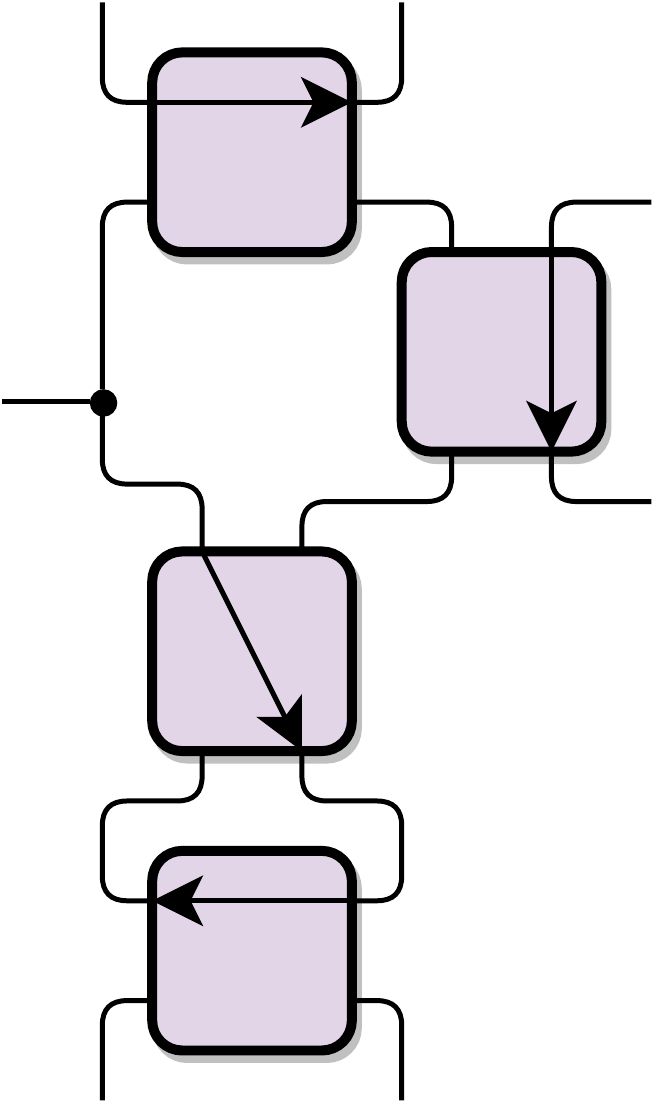}
    \caption{An AND vertex for planar graphs. Currently the weight-2 edge, connected at the bottom PL2T, is directed towards the vertex and locked, and both weight-1 edges are directed away. If the weight-1 edges become directed towards the vertex, the module can visit the vertex gadget and traverse a loop through all three PL2Ts, locking the weight-1 edges and unlocking the weight-2 edge.}
    \label{fig:planar NCL AND vertex}
  \end{subfigure}\hfil\hfil
  \begin{subfigure}{.45\textwidth}
    \centering
      \includegraphics[scale=0.6]{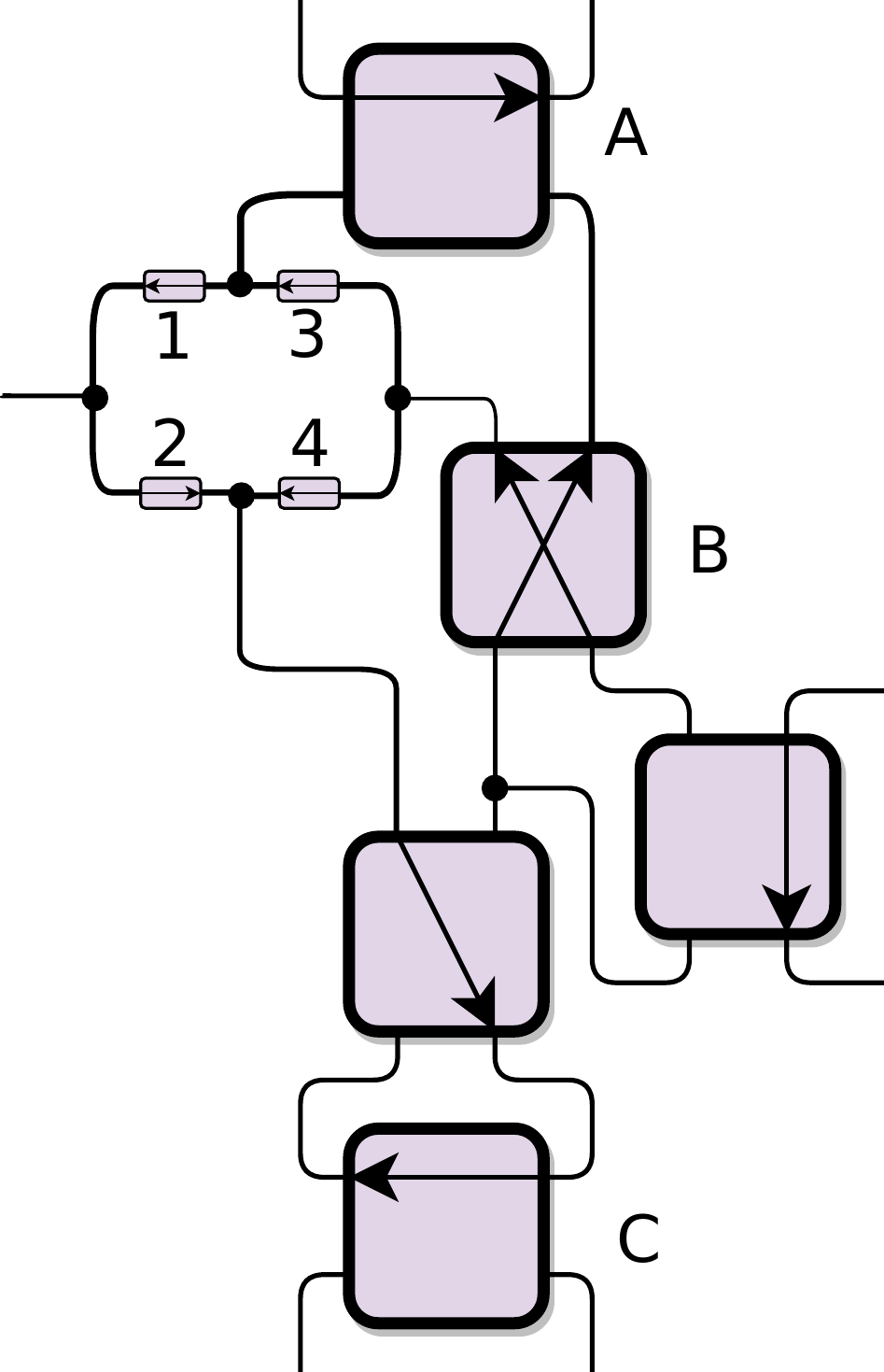}
  \caption{The augmented planar NCL OR gadget where edge C is currently locked pointed into the vertex, and the other edges are directed away. If multiple edges are ever directed towards the vertex, the robot can visit the vertex gadget, unlock the locked edge, and lock another edge.}
    \label{fig:planar NCL OR vertex}
  \end{subfigure}
  \caption{NCL vertex gadgets for planar graphs. In each gadget, each of the three PL2Ts is also part of an edge gadget. The agent enters at the line on the left, called the \emph{entrance}, traverses loops that enforce the NCL constraints, and then leaves at the entrance.}
  \label{fig:planar NCL vertex gadgets}
\end{figure}


\paragraph*{Accessing NCL edges}
In the original construction, we pick a rooted spanning tree of the dual of the planar NCL instance and use the edge-with-crossing gadget shown in Figure~\ref{fig:planar NCL edge gadget 2} for edges which are crossed by our spanning tree. The edge and vertex gadgets are accessible to at least one face, and are connected to that vertex of the spanning tree. When we use this tree to access an edge or vertex we can start at the root, flip the edge or vertex gadget, and return to the root ensuring that we reset any 1-toggles along the tree path. Since we don't need to worry about cycles, it is easy to replace the high degree nodes in the tree with multiple degree 3 nodes.


\begin{figure}
  \centering
  \includegraphics[scale=.75]{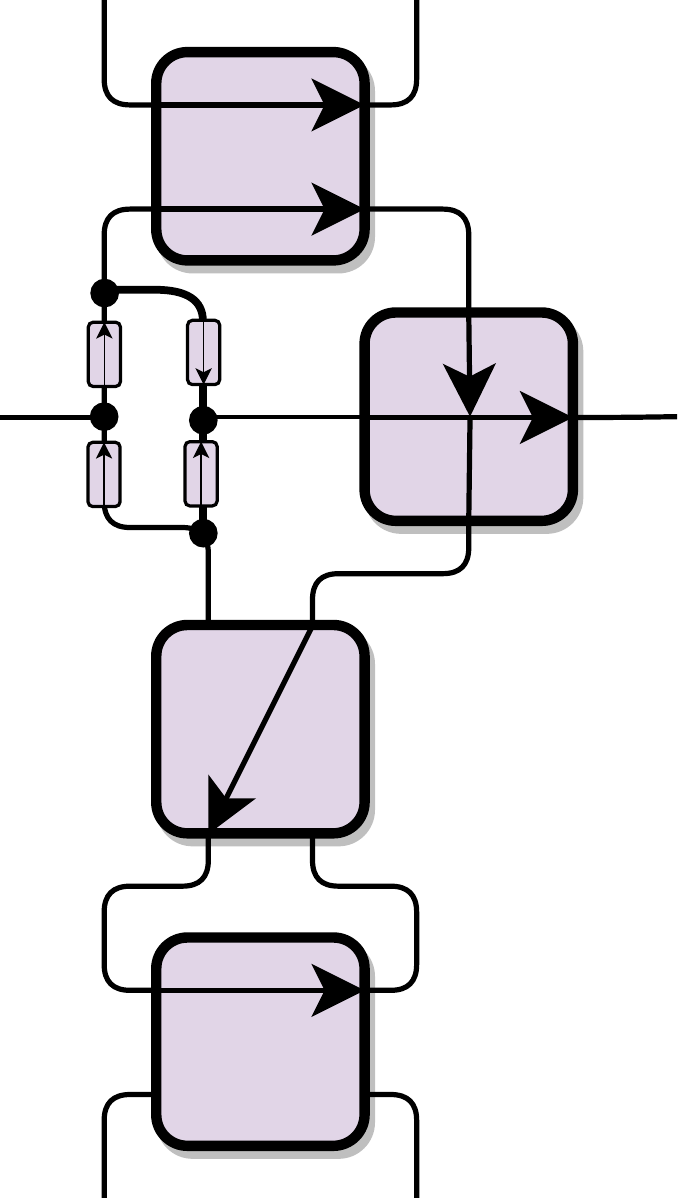}
  \caption{An edge gadget for planar graphs, currently unlocked and directed up. 
  We call the line on the left the \emph{entrance} and the line on the right, on the other side of the A/BA crossover, the \emph{exit}. Toggles have been added to the planar edge with crossing. Note that we will always cross back over the edge before flipping the edge, and thus it is not a problem that a single left to right crossing will put the gadget in a state in which the edge is unable to be flipped.}
  \label{fig:planar NCL edge gadget 2}
\end{figure}

\paragraph*{Simulating the locking 2-toggle}
Now that we've shown the reduction proving PSPACE-completeness of 1-player motion planning still holds with the added 1-toggles, all that remains is to show that the simulation of the locking 2-toggle by reversible, deterministic gadgets with interacting tunnels can also be augmented. 

First, all connections here are only degree $2$, which simplifies the matter, but since our gadgets might not be 1-toggle agnostic we cannot just use our prior observation about degree $2$ connections. 

The first part of the construction was to build a 1-toggle out of the gadgets if no 1-toggle was available. Since we are now inserting 1-toggles, we can jump straight to Figure~\ref{fig:arb gadget sim L2T} where the interacting tunnels are sandwiched together to ensure that each side either flips direction or closes when the opposite tunnel is crossed. Paired with the 1-toggles, this ensures that the crossed tunnel flips direction and the opposite tunnel closes. Given the presence of 1-toggles in the construction it is obvious that the addition of more around the gadgets does not change their function.

\begin{figure}
  \centering
  \includegraphics[scale=.75]{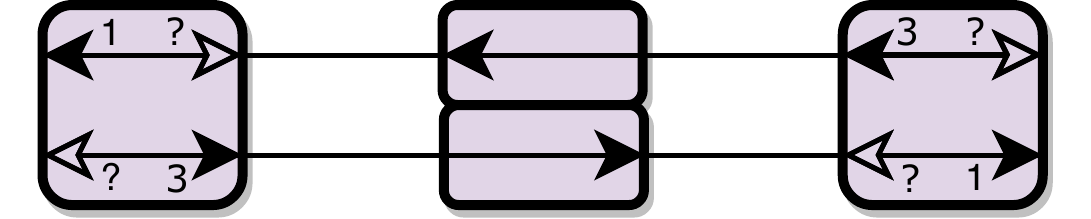}
  \caption{An arbitrary interacting-$k$-tunnel reversible deterministic gadget and a 1-toggle simulate a locking 2-toggle.}
  \label{fig:arb gadget sim L2T}
\end{figure}

\end{proof}

}

\subsection{Reduction for hexagonal modules}\label{hexmodel}

We now focus on describing how to simulate each of the pieces with hexagonal modules. The agent is represented by two modules and while these could go different ways, our instance contains several obstacles that can only be crossed by two modules working together. For simplicity we refer to the two modules that form the agent as {\em the agent modules}.

We simulate wires with sequences of modules in line segments. We also need to be able to turn without letting corner modules move. We simulate these turns using two types of corners: \emph{protected} which can be crossed, and \emph{blocked} which cannot be crossed - neither can move. Protected corners are drawn as green modules, while blocked corners are drawn as black. Wires and corners are discussed in detail in Section~\ref{wires}.

\later{
\subsubsection{Properties of the restricted model} \label{sec_propert}

In our reduction, the agent modules will be represented by the two modules that will be able to move in the configuration. Before getting into the details of the gadgets we discuss some important properties of our model. These properties will allow us to rule out other modules from moving. 

Consider the case shown in Figure~\ref{fig:hexmoves} (left): we have a leaf module (i.e., a module that is adjacent to only one other module in the configuration). If two nearby positions (shown as red dots in the figure) are occupied, then the leaf module cannot move under the restricted model.

Our reduction will also create a few vertices of degree three whose neighbors are consecutively placed along the boundary (Figure~\ref{fig:hexmoves} right). Similar to the previous case, we can prevent that module from being able to move by placing two modules in its neighborhood (red dots in the Figure). Whenever a module is in either of these two configurations (possibly after rotation), we say that it is {\em locked}.

\begin{figure}[ht]
	\centering
    \includegraphics[width=0.4 \textwidth]{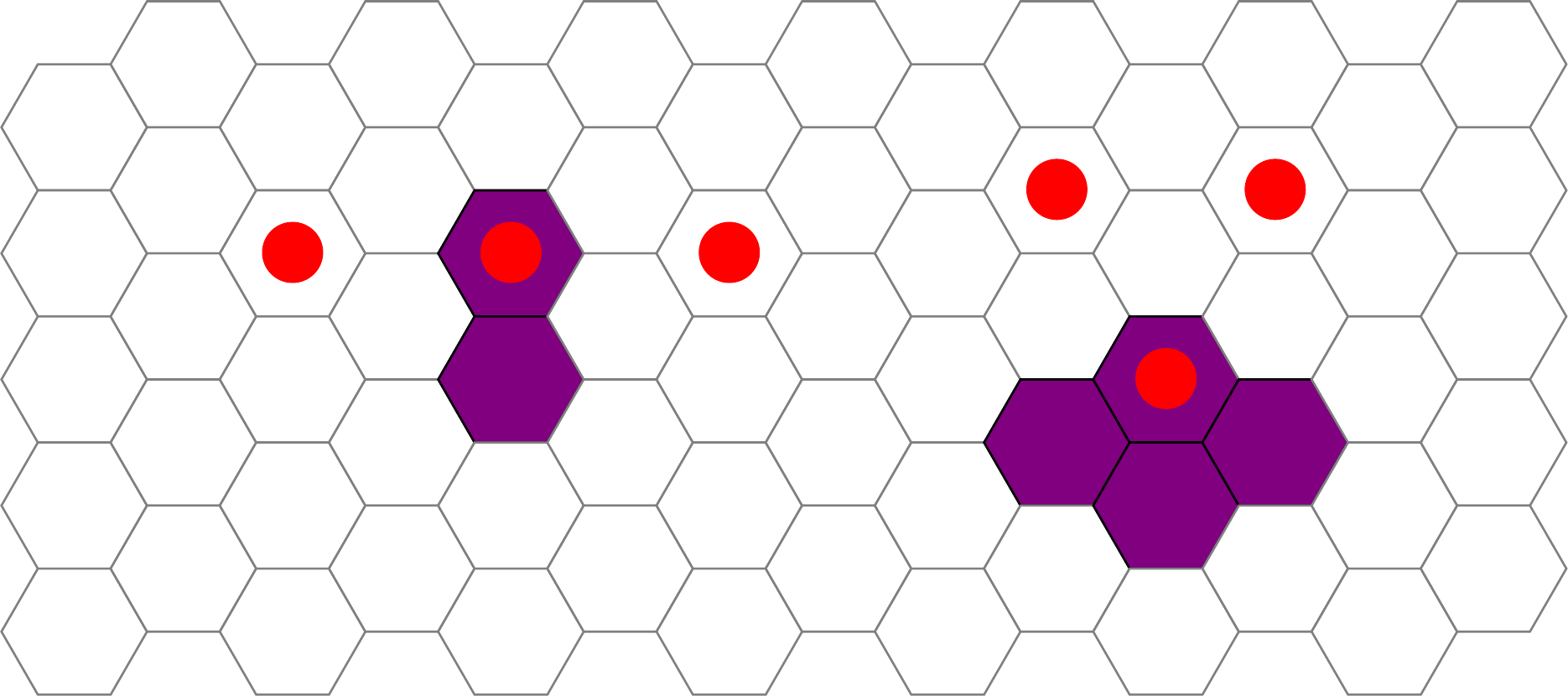}
    \caption{Two examples of modules (shown purple with the red dot) that cannot move because of the neighboring modules. The blocking modules are shown by red dots. Note that neither is needed for connectivity, but if nearby locations are occupied (red dotted positions) 
    they are not movable under the restricted model.
    In the right, the position immediately under the highlighted purple module need not be occupied to prevent movement.}
    \label{fig:hexmoves}
\end{figure}

In order to guarantee that the two agent modules move together, we introduce \textbf{gaps}. More specifically, we will be using $2$-gaps. A $2$-gap consists of three pairwise adjacent positions that are empty and appear along the boundary of a region that the agent is traversing (see Figure \ref{fig:hexgap}). The most important property of these $2$-gaps cannot be crossed by a single or even two modules.

This is because the first module can essentially only place itself in the the nearest empty space of the gap. Once there, the second module can place itself on top of the first one, but cannot cross without using the monkey move (a move that is not allowed under the restricted model). Note that three modules can potentially cross over a $2$-gap (or alternatively, having two modules on one side if there is a third module on the other side that can help them).

\begin{figure}[ht]
    \centering
    \includegraphics[width=0.4 \textwidth]{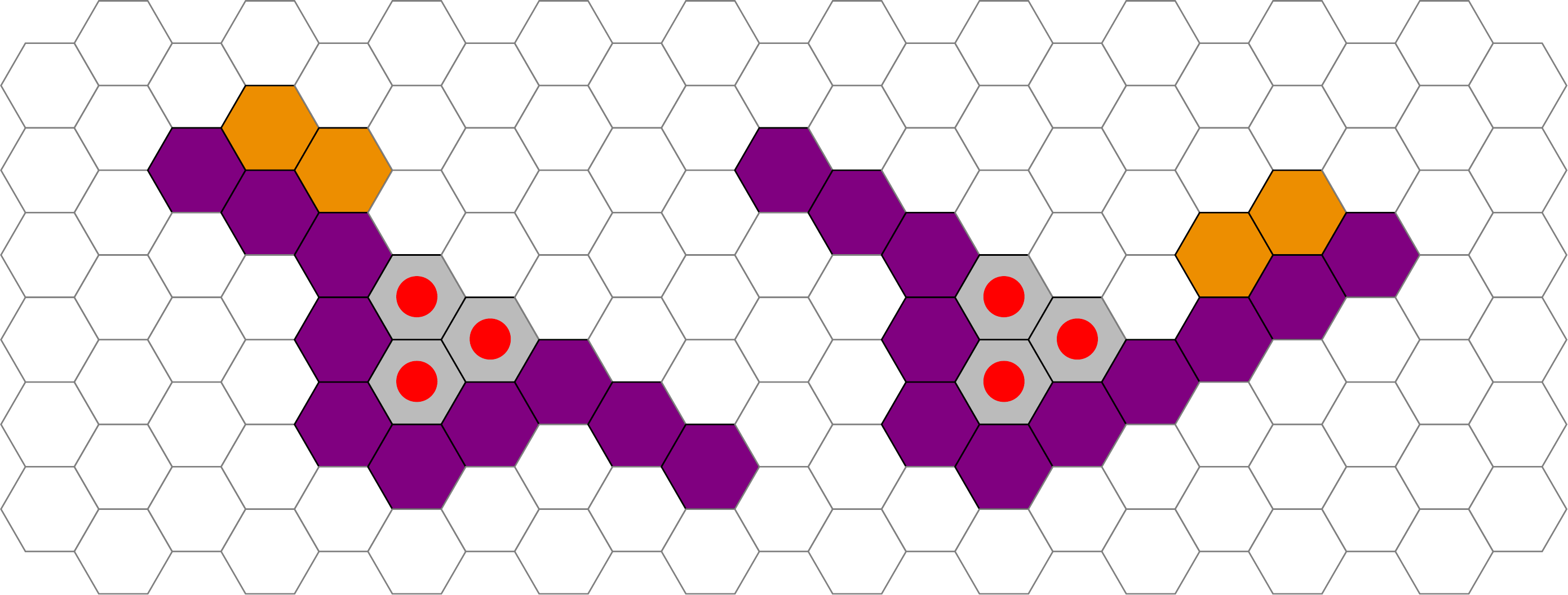}
    \caption{Examples of $2$-gaps (shown as gray dotted hexes). Note that two modules cannot cross the gap, but three modules together could cross both versions of the gap.}
    \label{fig:hexgap}
\end{figure}

The $2$-gap on the left in Figure \ref{fig:hexgap} can be crossed from both sides as long as there are $3$ modules there. It does not matter whether the modules are all on one side or split, they can all traverse the gap in either direction.

The $2$-gap on the right in Figure \ref{fig:hexgap} is slightly different. In this case, traversing the gap from left to right could require a large number of modules (the exact amount will depend on how many modules are placed in a row). But, if we have a single module on the other side of the gap (on the right), then we only need $2$ incoming modules (on the left). Traversing this same gap from right to left only requires $3$ modules, but only $2$ of them can cross and one must remain on the right side.

\subsubsection{Wires and wire corners} \label{wires}

We now turn our attention to wires: wires are simply represented by a sequence of modules forming a line segment. It is easy to see that no module in the middle of the segment cannot move (simply because the two neighbors prevent that). Although the gadget for wires is very simple, we need to add three minor tweaks: corners, wire side switches and wire cuts. The remainder of this section is dedicated to explaining why these tweaks are needed and how they work.

The simplest of the three is the corner: whenever the wire has to bend, we need to insert an additional junction module connecting the two wires. In order to prevent this module from moving we insert a small number modules forming a spiral to prevent the move from happening (see Figure \ref{fig:hexcorners}). If needed, the same spirals can be used to make sure the endpoints of a segment cannot move.

We introduce two different types of corners: \textbf{protected} and \textbf{blocked} (see Figure \ref{fig:hexcorners} - (left) and (right)). Both of them prevent the corner module from moving, but in addition they have extra properties. 

\begin{lemma}
None of the modules in the neighborhood of a blocked corner can move, even if the agent modules can reach the corner.
\end{lemma}
\begin{proof}
The only module of the gadget that is not critical for the connectivity is the one at the end of the spiral. Since that module is locked, we conclude that none of the modules involved in the corner can move. The second statement follows from the fact that the two agent modules cannot move past the $2$-gap (as discussed in Section~\ref{sec_propert}, at least three modules would be needed).
\end{proof}

The protected corner is very similar to the blocked one. This one is used when we want the agent to continue along the wire on the same side (but still prevent the corner module from moving). 

\begin{figure}[ht]
	\centering
	\includegraphics[width = 0.6 \linewidth]{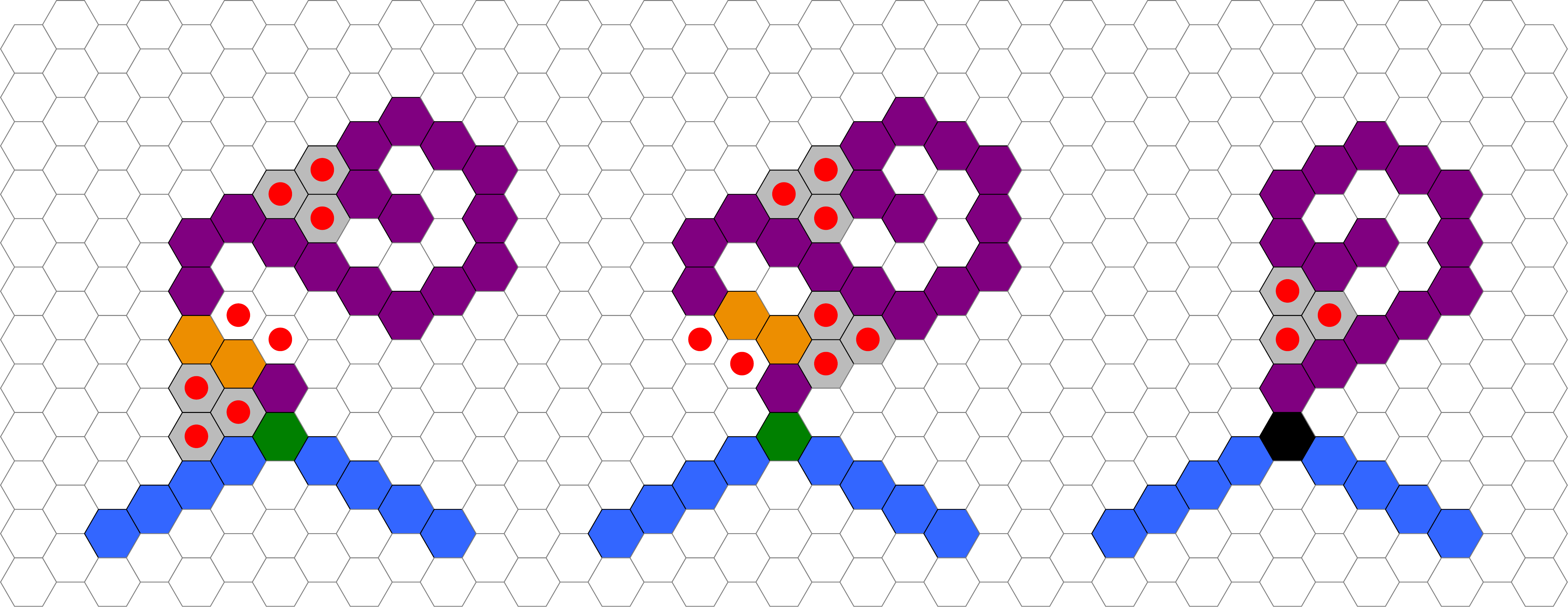}
	\caption{(left and center) Examples of left and right leaning protected corners (corner module shown in green). Note that no module in the spiral can move, nor can the agent modules pass through because of the $2$-gap. However, the agent modules can traverse through by placing themselves in the highlighted slots. By doing so, they will slightly alter the position of the spiral.
	(right) Similar construction for the blocked corner. In this case, nothing can move nor can the agent pass through.}
	\label{fig:hexcorners}
\end{figure}

\begin{lemma}
None of the modules in the neighborhood of a protected corner can move. Moreover, when an agent reaches the protected corner, it can cross the gadget by toggling the status from left to right leaning depending on the situation (or vice versa). Finally, if a gadget is left (resp. right) leaning, then it cannot be crossed by an agent that reaches it from the left (resp. right) side.
\end{lemma}
\begin{proof}
Proof of the first claim is identical to the protected gadget: all modules are critical to the connectivity except a single leaf module that cannot move because it is locked. Thus we focus on explaining how an agent can cross the protected corner.

As in the blocked corner case, the agent can move around but will not be able to cross the $2$-gap. The main difference to the blocked corner is that there is a specific location that they can attach (shown as red dots in Figure~\ref{fig:hexcorners}), on the left and in the center. The left image is a left leaning protected corner and the center image is a right leaning one.

Once the agent modules have attached, the two orange modules in the spiral can move and continue along the other wire (the one that the agent did not enter the gadget on). Overall, what happened is that two modules entered the gadget from one wire, and two modules left the gadget on the opposite wire. The modules that exit the spiral are different, but conceptually we see it as the agent having successfully traversed through the protected corner. In either of these states of the corner, the agent cannot cross in the same direction a second time. The only way these agents can be traversed is backwards, resetting the corner to the other configuration.

Finally, we observe that the behavior of the protected corner is exactly the same as that of a $1$-toggle. This means that as long as we initially set up the corner in the correct direct, its behavior will be as intended. This completes the proof.
\end{proof}

We say that a protected corner is left or right leaning (depending on whether it can only be crossed from left to right or the reverse), like in Figure \ref{fig:hexcorners}. From now on, in each of our constructions, we use a black module to denote a blocked corner (the spiral is not shown for simplicity). Similarly, green modules represent protected corners, where we need to be careful about their initial state.

\subsubsection{Wire cut gadget}\label{sec_wirecut}
Next, we introduce the wire cut gadget. This gadget is not directly required by the 1-toggle-protected motion planning with the locking 2-toggle, but it is needed for the proper working of the reduction.

A key aspect that provides correctness of many parts our reduction is that many modules in our gadgets are critical to connectivity of the graph. This argument holds on a local level (when we look at each gadget independently). However, we must also rule out the presence of large cycles in the overall construction (i.e., a cycle involving more than one gadget). This is where the wire cut gadget comes into play.

\begin{figure}[ht]
	\centering
	\includegraphics[width=0.8\linewidth]{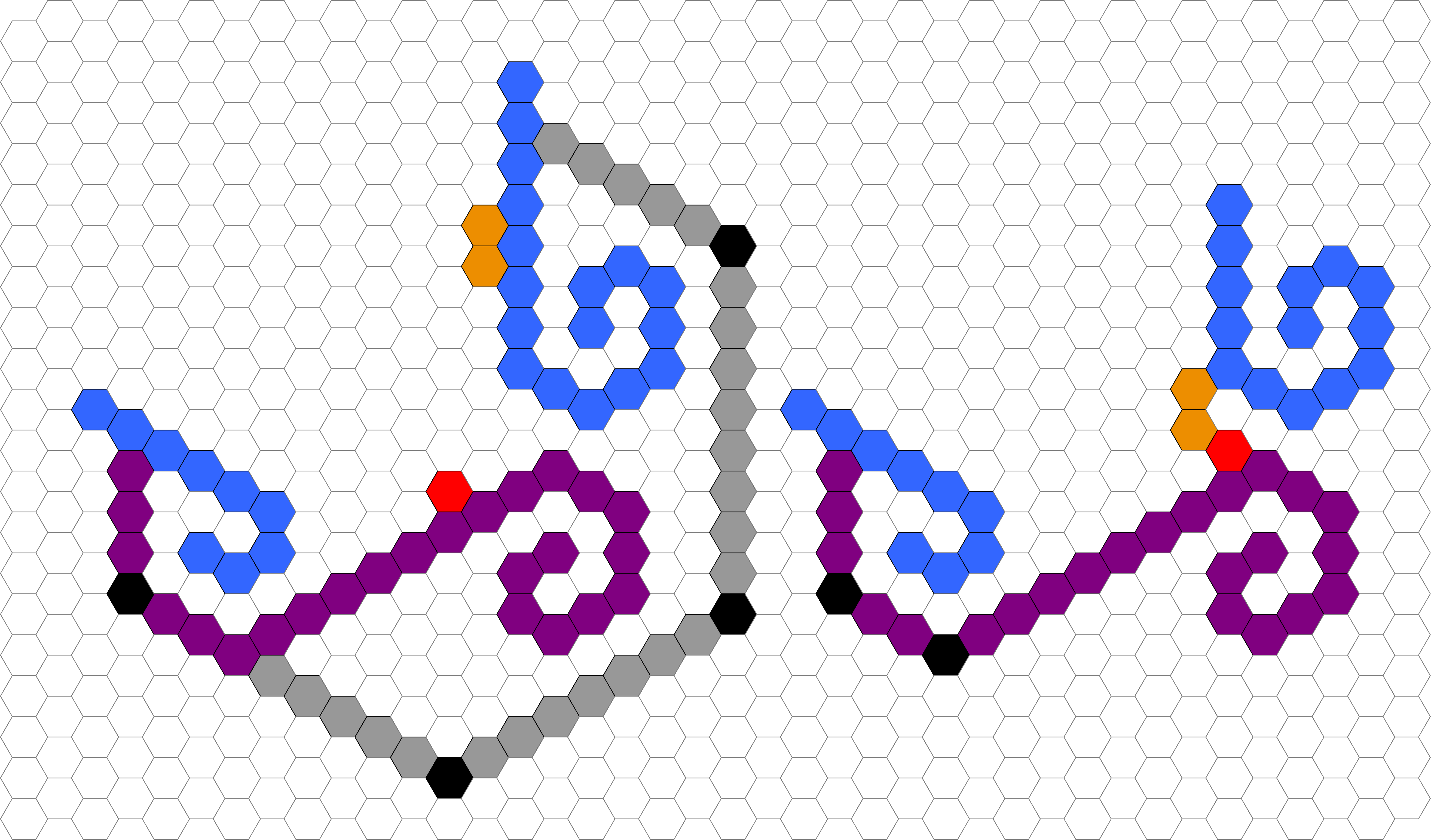}
	\caption{(left, without the gray modules connecting the wires) Wire cut gadget with the auxiliary module highlighted in red and the agent coming from the right side.
	(left, with the gray modules connecting the wires) The connected wire cut gadget acts as a $1$-gap. Note that a single module cannot traverse the gadget, but two modules together can traverse it in either direction.
	(right) The two agent modules and the auxiliary module can form a bridge so that the three modules get onto the wire cut gadget. Notice that, in the process a cycle is created. This cycle can be either local to the gadget (if the modules form a bridge on the left side), or it can be global, like in our example. The gray modules are omitted.}
	\label{fig:hexwirecut}
\end{figure}

For each cycle in the underlying connection graph, pick any of its edges and place a wire cut gadget on it (shown in Figure~\ref{fig:hexwirecut} left, without the gray modules). This will ensure that the only cycles present in the final problem instance are those that are contained within a single gadget.

We observe a few things about the wire cut gadget: first of all, the gadget is not connected. This is by design, as we know that the two portions of the wire gadget are still connected globally.

The wire cut gadget contains a specific module (that we call {\em auxiliary}) that can move. However, notice that it cannot leave the wire cut gadget because of the $2$-gaps present on both sides. No other module in the gadget can move because of the same reasons as in the protected and corners: each module is either critical for connectivity or locked (recall that the two vertices shown as black in Figure~\ref{fig:hexwirecut} contain a blocked spiral).

Further note that the wire cut gadget changes the direction of the original wire gadget (initially it was a segment, but now the two portions of the segment form an angle of 60 degrees). This can be fixed by adding wires and corners as discussed in Section~\ref{wires} so that the two wires are extensions of each other. Note that the agent must pass through those corners, so corners must be protected (effectively making them 1-toggles).

\begin{lemma}
An agent can pass through the wire cut gadget, leaving it disconnected. Moreover, after passing, an auxiliary module will still remain directly attached to the gadget.
\end{lemma}
\begin{proof}
Consider the case in which the agent modules come from the right side (the left side is symmetric). The presence of the $2$-gap makes it so that the agents cannot cross alone. However, together with the auxiliary module they can form a bridge (shown in Figure~\ref{fig:hexwirecut} right). Once the bridge is built the agent modules can move onto the wire cut gadget.  By doing the same moves in reverse, they can form a bridge and allow {\em two} modules to cross to the left side. Indeed, out of the three modules, the one that is connected to the wire cut gadget will not be able to move (because of the location of its neighbors). Movement must start from the middle module, and once that module has moved the module on the wire will be isolated.

Thus, when two agent modules come from one side we can move them to the other side while at the same time guarantee that one module is confined to be within the wire cut gadget. This, together with the fact that the gadget is disconnected guarantees that it works as intended: allowing the agents to cross while at the same time reduce the cycles in the configuration.
\end{proof}

Although the gadget is good for breaking cycles and lowering the connectivity of the instance, there is one situation in which it cannot prevent a cycle from being formed: when the agent is crossing the wire cut cycle, they form a a bridge and at that bridge creates a cycle. In the following we explain how to handle this situation.

First notice that if the bridge is formed in the left side, it only involves vertices of the wire cut gadget. Given the small size, one can verify that only the modules that form the wire can move. A more interesting case happens when the bridge is in the right side. In this case, the cycle can span several gadgets (we call this situation a \emph{global cycle}). 

We design our gadgets so that very few modules can move during a global cycle. Specifically, in the gadgets we introduced up to now no module will be able to move because of a global cycle (the argument for moves not being possible in a wire gadget was based on the location of the neighbors, not connectivity. Also, spirals cannot be part of a global cycle because they are connected to a single vertex). 

In the upcoming gadgets, the modules that can move during a global cycle are those that we want to move when the agent interact with the gadget. The main difference is that when the agent comes close, two modules of the gadget can move whereas in a global cycle only one module can move. We call this situation the {\em solo agent}. 

We will show that a solo agent cannot move beyond its gadget. This is done by adding a connected wire cut that acts as a $1$-gap to every wire gadget (see Figure~\ref{fig:hexwirecut} left, with the gray modules). This $1$-gap can be easily crossed by two agents alone, but a single module will not be able to traverse it. These $1$-gaps along the construction will also ensure that the two modules that form the agent cannot go in different directions. In each of the remaining gadgets, we will argue that a solo agent will not be able to change the gadget's state in a meaningful way.


\subsubsection{Side switch gadget}\label{sec_sideswitch}

Now we turn our attention to the side switch gadget. In 1-toggle-protected motion planning with the locking 2-toggle the agent moves along the wire. However, in the module problem setting, our agent is represented by two modules that must be on one side or the other of a wire gadget. This change on its own is not of big significance: we just have to make sure the agent is always on the correct side of the incoming wire into every gadget (say, when the module approaches a wire cut gadget, it should do so on the side that allows it to interact with the corner gadget).

In the initial condition we determine the side on which the agent starts. From there on we keep track of the only side of the wire that the agent can be. Sometimes we will need to force the agent to switch sides. Whenever this is needed, we will introduce the side switch gadget, shown in Figure \ref{fig:hexwireside}.

\begin{figure}[ht]
    \begin{minipage}[t]{0.5\textwidth}
        \centering
        \includegraphics[width=\linewidth]{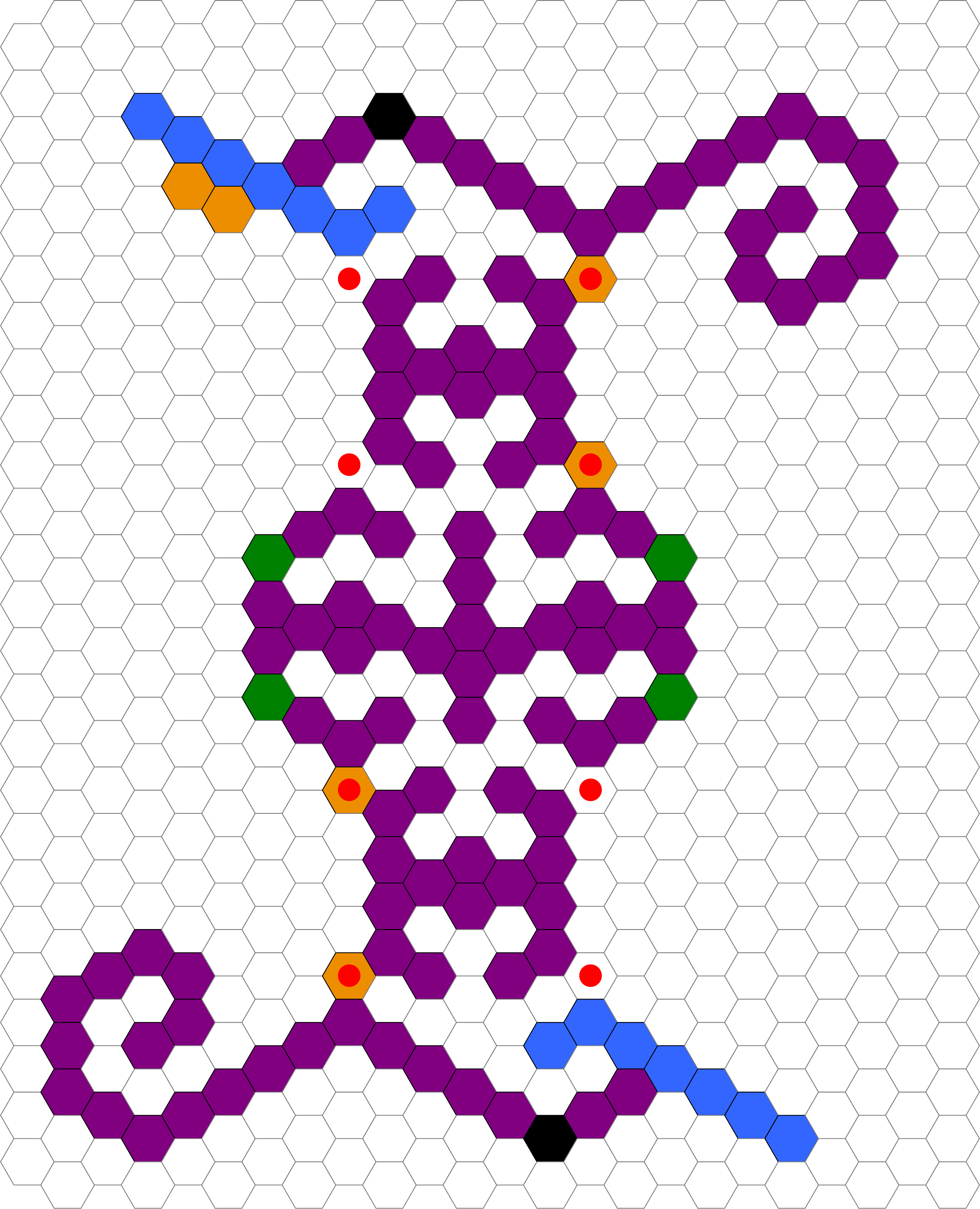}
    \end{minipage}
    \hfill
    \begin{minipage}[t]{0.5\textwidth}
        \centering
    	\includegraphics[width=\linewidth]{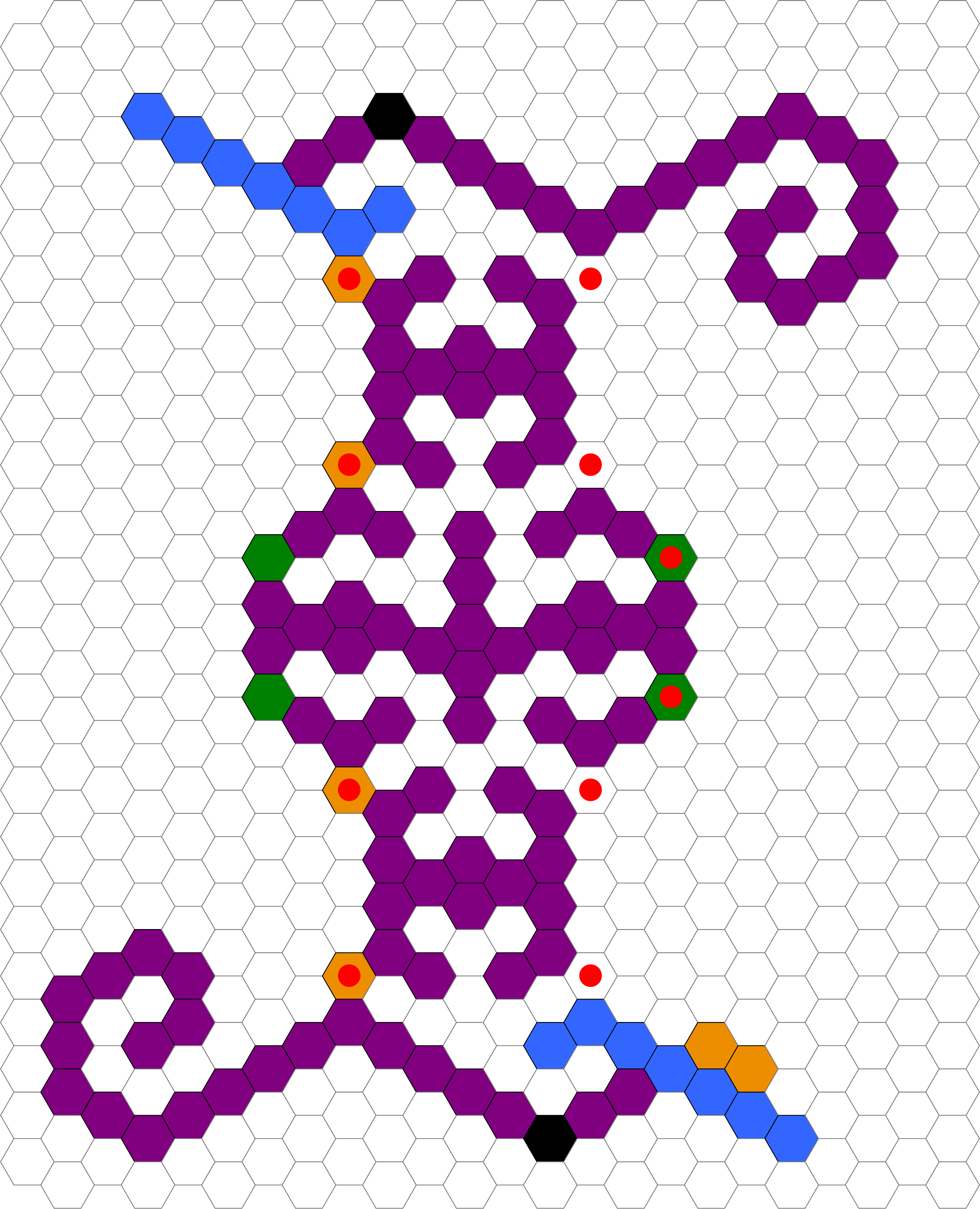}
    \end{minipage}
    \caption{(left) Wire side switch gadget. This gadget is added in the middle of a wire segment. If the agent (shown in orange) enters the wire below the upper half of the wire, then it can move to the nearest highlighted red dots. Recall that modules shown as green in the figure are protected corners, with the spiral not shown for clarity. The protected corners on the left can be traversed bottom up and the ones on the right can be traversed top down.
    (right) When the agent modules are placed on the designated locations, they create two cycles, which in turn allows for two other modules on the opposite side to move. These two modules can proceed down along the gadget, and exit on the other side of the wire gadget.}
\label{fig:hexwireside}
\end{figure}

\begin{lemma}
When the switch gadget is not part of a global cycle, none of its modules can move. Moreover, if the agent enters the gadget along one wire, two other modules can exit on the the other wire.
\end{lemma}
\begin{proof}
Consider first the gadget alone (i.e., without the agent). As in the previous gadgets, no module can move because either $(a)$ it is critical for connectivity or $(b)$ it is locked (as defined in Section~\ref{hexmodel}). This also applies to the two new endpoints that are created when we split the wire gadget into two (the leaves shown in blue). 

Now consider the case in which the agent enters the gadget (say, from under the left half of the wire gadget - the other case is symmetric). These modules can move along the top left side of the gadget, but will not be able to cross the protected corners on the left, since they are leaning up. Thus, no significant change can happen to the configuration.

The agent modules can do a more significant change by placing themselves on the top left positions marked with red dots. By doing so they each create a local cycle (i.e., constrained within the gadget). By design of the gadget, only two new modules will be able to move (shown as orange red dotted modules in the Figure): every other module in these cycles is either {\em locked} or its immediate neighbors prevents it from moving. These two modules can now go downward, cross the two protected corners (since they are leaning down) and reach the upper side of the right half wire segment (see the result after the traversal in Figure \ref{fig:hexwireside} right).

As in the case of the protected corner, two modules entered from the bottom half of the left wire, and two (different) modules have exited the other half of the wire but from the other side. In addition, the state of the two protected corners on the right have changed from leaning down to leaning up. On a macro scale, this corresponds to the agent switching side of the wire and changing the state of the gadget.

Notice that if the agent tries to traverse this gadget left to right again, it will not be able to. On the other hand, the agent can traverse it in the opposite direction (by simply undoing all steps).
\end{proof}

In order to complete the description, we need to discuss what happens when this gadget is part of a global cycle (see  Section~\ref{sec_wirecut}) and how to handle solo agents. In a global cycle we know that there is a path connecting the two halves of the segment wire outside the gadget (that path together with the path inside the gadget creates a cycle). Theoretically, this allows other modules to move (since none is critical for the connectivity anymore). However, by the way in which the gadget was designed, the only modules that will have space around themselves to move are the four modules highlighted in orange with a red dot in Figure \ref{fig:hexwireside}. Note that any of the four modules can move, but as soon as one moves, we have broken the cycle and none of the other three can move. The module that moves will become a solo agent. 

This single module can move up and down its side of the gadget and possibly exit through one of the halves of the wire segment. However, because of the $1$-gap that we add to each wire segment, it will not be capable of reaching other gadgets. Thus, a other than changing its location (and partially changing the state of protected corners), the only thing that a solo agent can do is return to its place and allow the module agent that are currently forming a cycle to break the cycle and continue traversing along the problem instance.

Overall, we conclude that the gadget works properly, and that with it we can keep track of the single side of the wire that the agent is allowed to be on. If needed, we can use this gadget to have the agent switch sides of the wire as well. 

}

\both{
\subsubsection{Branching hallway}


The branching hallway, shown in Figure~\ref{fig:hexbranch}, allows an agent that arrives on any of the three wires to leave on any of the other two. This construction actually acts as a branching hallway with 1-toggles on two of its wires.} \both{We can implement the toggle on the third wire by adding a restricted corner.

\begin{figure}[ht]
    \centering
	\includegraphics[width=0.4 \linewidth]{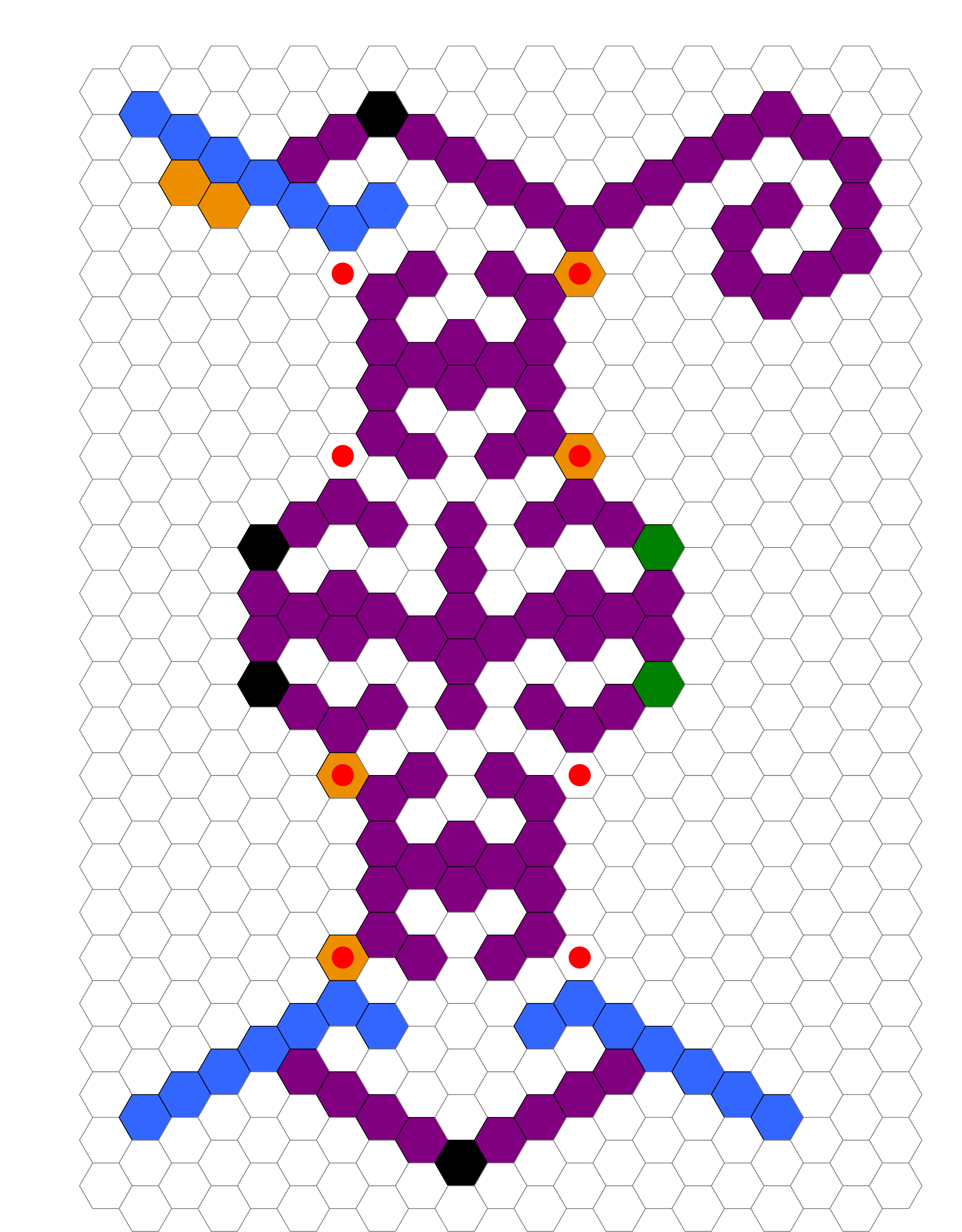}
	\caption{The branching hallway gadget simulated with hexagonal modules.}
	\label{fig:hexbranch}
\end{figure}

\begin{figure}[ht]
    \begin{minipage}[t]{0.3\textwidth}
        \centering
        \includegraphics[width=\linewidth]{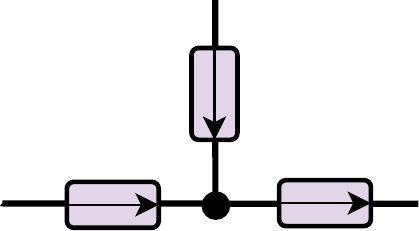}
    \end{minipage}
    \hfil
    \begin{minipage}[t]{0.3\textwidth}
        \centering
    	\includegraphics[width=\linewidth]{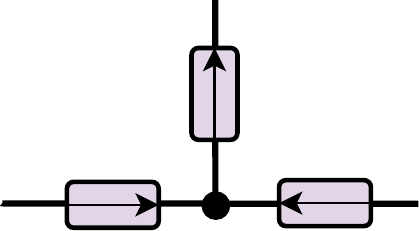}
    \end{minipage}
    \hfil
        \begin{minipage}[t]{0.3\textwidth}
        \centering
        \includegraphics[width=\linewidth]{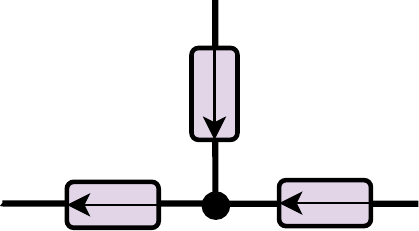}
    \end{minipage}
    
    \vspace{5mm}
    
    \begin{minipage}[t]{0.3\textwidth}
        \centering
    	\includegraphics[width=\linewidth]{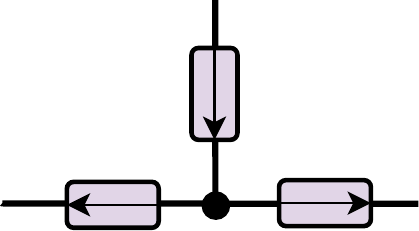}
    \end{minipage}
    \hfil
        \begin{minipage}[t]{0.3\textwidth}
        \centering
        \includegraphics[width=\linewidth]{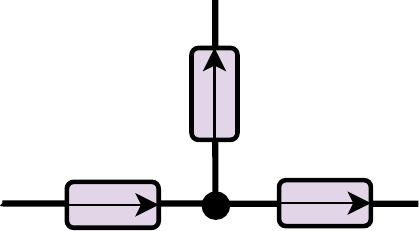}
    \end{minipage}
    \hfil
    \begin{minipage}[t]{0.3\textwidth}
        \centering
    	\includegraphics[width=\linewidth]{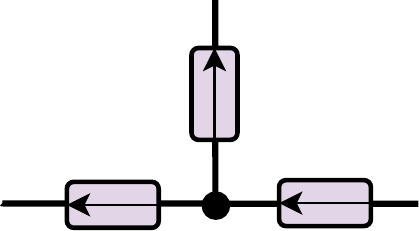}
    \end{minipage}
    \caption{The six configurations of a branching hallway with endpoints connected to 1-toggles.}
\label{fig:1-toggle-hallway}
\end{figure}
}

\later{
This along with the protected corners which are needed for turning wires, acting as 1-toggles are the reason for reducing from 1-toggle-protected motion planning with the locking 2-toggle rather than the unrestricted version of motion planning with locking 2-toggles. Note that the leaning of the protected corners is determined by the direction of the $1$-toggle on the top left wire.

This gadget works in a similar way as the side switch gadget. Again there are four pairs of locations that are critical for the working of the gadget (marked with red dots). Two of the four must be occupied for the gadget to be connected (one pair above and one pair below). As long as the agent does not come close to the gadget nothing can move because of connectivity issues.

\begin{lemma}
The construction shown in Figure~\ref{fig:hexbranch} correctly simulates the behavior of the 1-toggle protected branching hallway gadget.
\end{lemma}
\begin{proof}
We start by discussing the upper pair of critical positions. When the right positions are occupied (as shown in in Figure~\ref{fig:hexbranch}) the agent can come from the upper left wire and occupy the upper left empty slots. This frees the orange modules on the right which can traverse downwards. These two orange modules can directly exit the gadget through the lower right wire, or they can place themselves at the last critical pair of positions and allow the other two orange modules to exit from the lower left wire. Overall, the agent entered from the top wire and can exit from either of the two wires. A key aspect of our gadget is that out the two pairs of upper critical positions the right one was initially occupied, and thus the agent could do a meaningful change. This is the equivalent of the upper portion of the 3 branching hallway pointing down. 

A similar argument holds for the lower half of the reduction and the 2 critical positions in the lower half. This time, notice that the position of the orange modules is flipped in the upper and lower halves of Figure~\ref{fig:hexbranch}. This represents that the left portion of the 3 branching hallway is pointing outwards (and thus the agent should not be able to cross). Indeed, this is the case: if the agent were to come from the left side, it can go near the occupied critical positions, but no meaningful cycle is created (any cycle it creates will not allow new modules from moving). Thus, the agent is prevented from interacting with the gadget as intended.

Overall we have that the upper two pairs of critical positions model the direction of the upper 1-toggle and the lower pairs simulate the left 1-toggle. The blocked spirals in the right side prevent an agent that comes from either of the two wires to go onto the other half of the gadget.

If the agent were to come from the right, it can place itself in either of the two right critical positions (if either of them is empty) and exit from one of the two left wires. In the process we would also switch again the state of the corresponding portion of the gadget. 

Finally it remains to discuss the situation in which a global cycle is created. Recall that this can only happen when the agent is crossing through a wire cut gadget. In this case, any of the four orange modules will be able to move (but as soon as that module moves neither of the other ones nor the agent will be able to move or connectivity will be broken). The orange modules on the left cannot do any meaningful change because the blocked spirals prevent it from going upward. An orange module in the left side can traverse up/down along the gadget and move onto another empty critical position (if any), but this just again allows another orange module to move. The orange module could also potentially exit the gadget via one of the wire gadgets, but the $1$-gap that we add into each wire will prevent it from reaching another gadget. Thus, we conclude that the gadget behaves as desired.

\end{proof}
}

\both{
\subsubsection{Locking 2-toggle (L2T)}

The other main gadget needed for the reduction is the locking $2$-toggle shown in Figure~\ref{fig:hexl2tHex}. \iffull This gadget receives 4 wires and has three states: it can be crossed along either of the wires (from top to bottom) or it can be crossed by only one of the wires from bottom to top. We model this gadget with the configuration of modules shown in Figure~\ref{fig:hexl2tHex}. \fi

\begin{figure}[htb]
    \begin{minipage}[t]{0.48\textwidth}
        \centering
        \includegraphics[width=\linewidth]{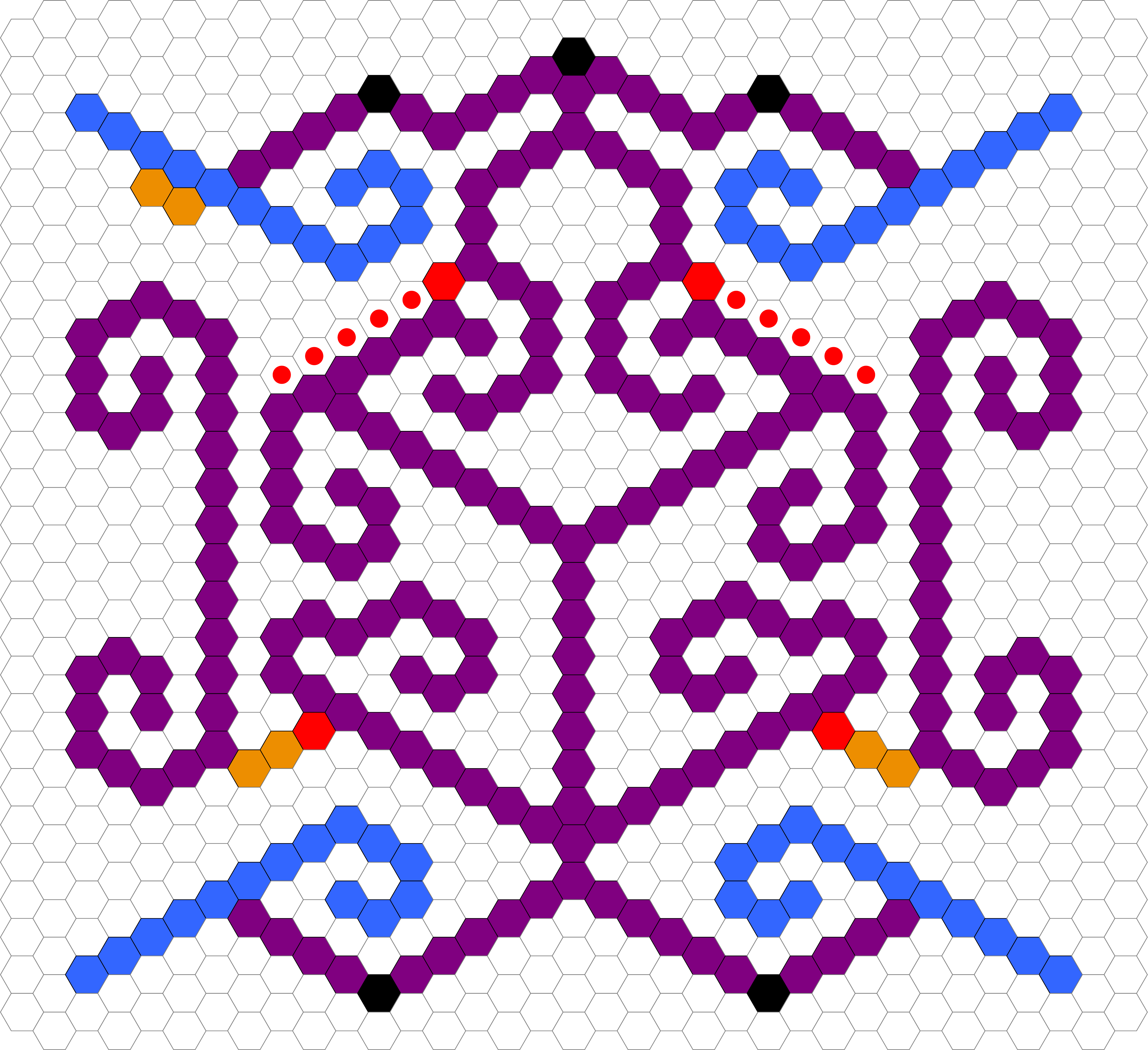}
    \end{minipage}
    \begin{minipage}[t]{0.48\textwidth}
        \centering
    	\includegraphics[width=\linewidth]{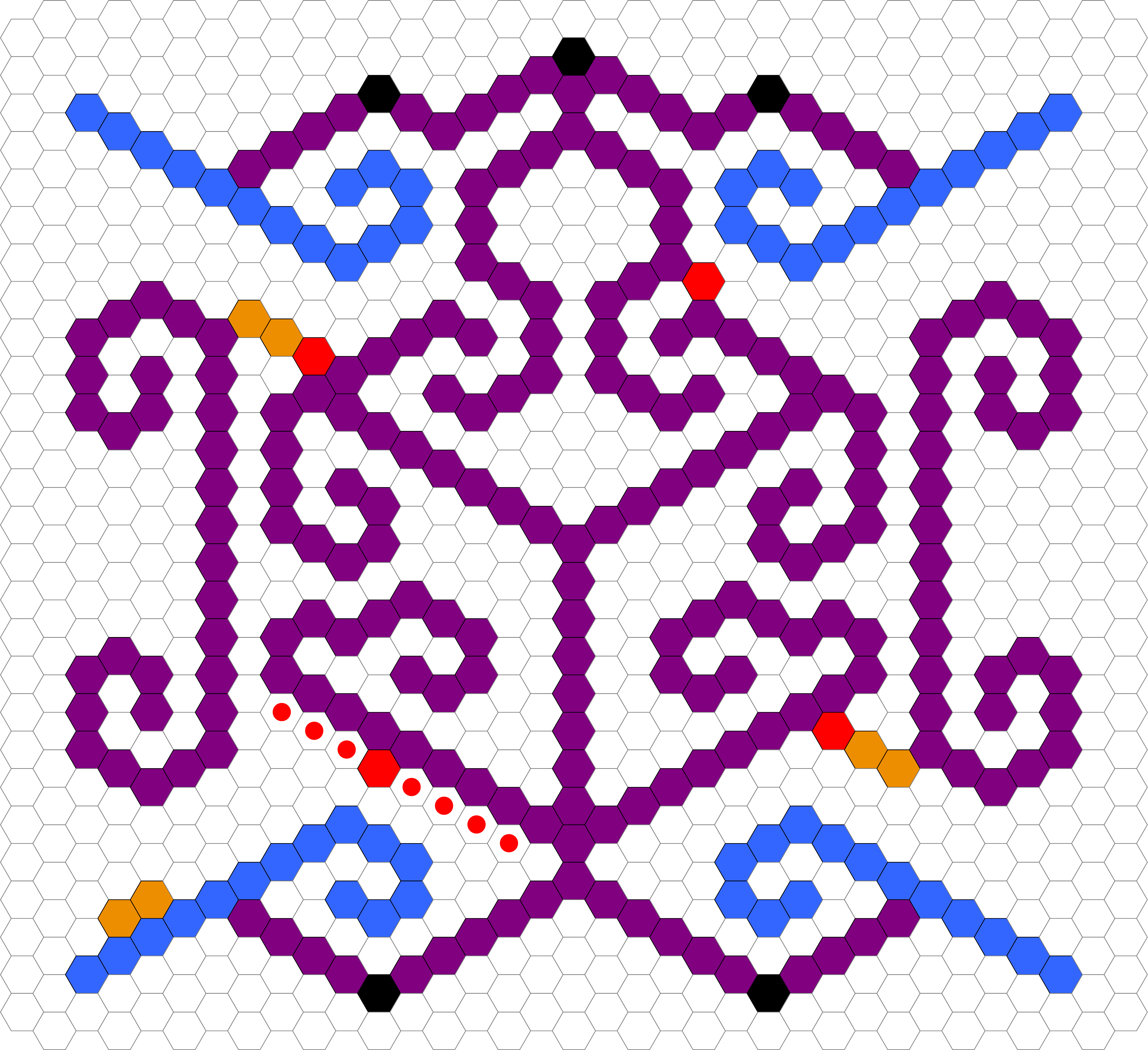}
    \end{minipage}
    \caption{(left) The locking $2$-toggle simulated with hexagonal modules (in the top state 3 in Figure \ref{fig:L2T}). Other than the agent, the modules that can move are the two red modules towards the top of the gadget (but only one at a time).
    (right) Once the agent forms a bridge with the auxiliary left module, they create a cycle allowing three modules at the bottom part of the gadget to move. Two of the three modules can exit the gadget along the bottom wire. Now only the lower left module can move. This changes the state of the $2$-toggle which can now only be crossed in reverse.}
\label{fig:hexl2tHex}
\end{figure}
}

\later{
As expected, this gadget is connected to 4 wire gadgets (shown in blue). The blue wires at the top left and top right are the two entrances and the blue wires at the bottom left and bottom right are the two exits. The L2T in Figure~\ref{fig:hexl2tHex} (left) is open and can be traversed on either side from the top to the bottom. 

The gadget contains 8 special modules (shown in orange and red in the Figure) split into 4 groups. The groups of 3 modules maintain connectivity with the spirals on the sides whereas the two single modules (which we call {\em auxiliary}, connect the upper and lower portions of the gadget. Note that, although two modules to perform the connection only one is needed. 

\begin{lemma}
The only modules that can move in a L2T gadget are the auxiliary modules, even in a global cycle. Moreover, this configuration of modules properly models the behavior of the L2T toggle.
\end{lemma}
\begin{proof}
First we argue about which modules can move: when the agent is not present only the two auxiliary modules can move, but the range of positions is very limited (possible places are shown with red dots in shown in Figure~\ref{fig:hexl2tHex} left). Outside a global cycle only one of the two can move (since the other must remain to preserve connectivity). Observe that no other module can move, even during a global cycle (as in previous cases, neighboring positions prevent all modules that can be involved in a cycle from moving). Thus, as long as the agent is not present no significant change can happen to the gadget.

Now we consider the case in which the agent reaches the gadget (and the gadget is in the state of Figure~\ref{fig:hexl2tHex} left). If the agent comes from either of the lower wires, it will not be able to reach the middle of the gadget, nor alter significantly its state. 

If the agent reaches the gadget from either of the upper wires the situation changes: in this case, the auxiliary module of that side can move and help the agents form a bridge. Notice that the bridge creates a cycle contained within the gadget, but no new of the cycle can move (for the same reasons as when a global cycle is created).

Once on the gadget, the agent and the red auxiliary module can now bridge over to the ``almost disconnected`` configuration (modules on the side of the configuration that are connected to the rest of the configuration by the other three module group at the bottom of the gadget, on the same side). This creates a cycle and allows the bottom three modules to move. By doing the reverse steps, we can transfer two of the three modules onto the lower wire. Notice that the presence of the $2$-gap that is not on a line forces one of the modules to remain in the gadget (and thus only two modules move  onto the wire). See the traversed state, in Figure \ref{fig:hexl2tHex} right.

Overall, we have that: the gadget alone cannot change its state. Moreover, if the agent comes from any of the lower wires it cannot alter its state (nor cross the gadget). If the agent instead comes from either of the upper wires, it can cross to the same lower wire and the state of the gadget has changed. 

By using the same arguments, we can see that once the gadget is in the situation of Figure \ref{fig:hexl2tHex} right, the agent can only alter the state of the gadget by coming back from the same wire and undoing all operations. If the agent comes from any of the other three wires, the $2$-gap (and lack of auxiliary modules) will prevent the agent from interacting with the gadget. This matches the behavior of the 2-toggle as defined in Section~\ref{sec:prelim}.
\end{proof}

Notice that in the argument above it was critical that both agent modules interacted with the gadget. If a single module reached the target (say, because the agent modules decide to split or a global cycle allows one module from a wire switch gadget to move), that module will not be able to change the state of the L2T in any way.

\subsubsection{Win gadget}
As mentioned in Section \ref{sec:prelim}, the last gadget we need is the win gadget. This gadget is only used to mark our finished state. If the agent can reach it, we can change its state then reverse all other moves made to return to the agent's initial start location

\begin{figure}[ht]
    \centering
    \begin{minipage}[t]{0.4\textwidth}
        \centering
        \includegraphics[width=\linewidth]{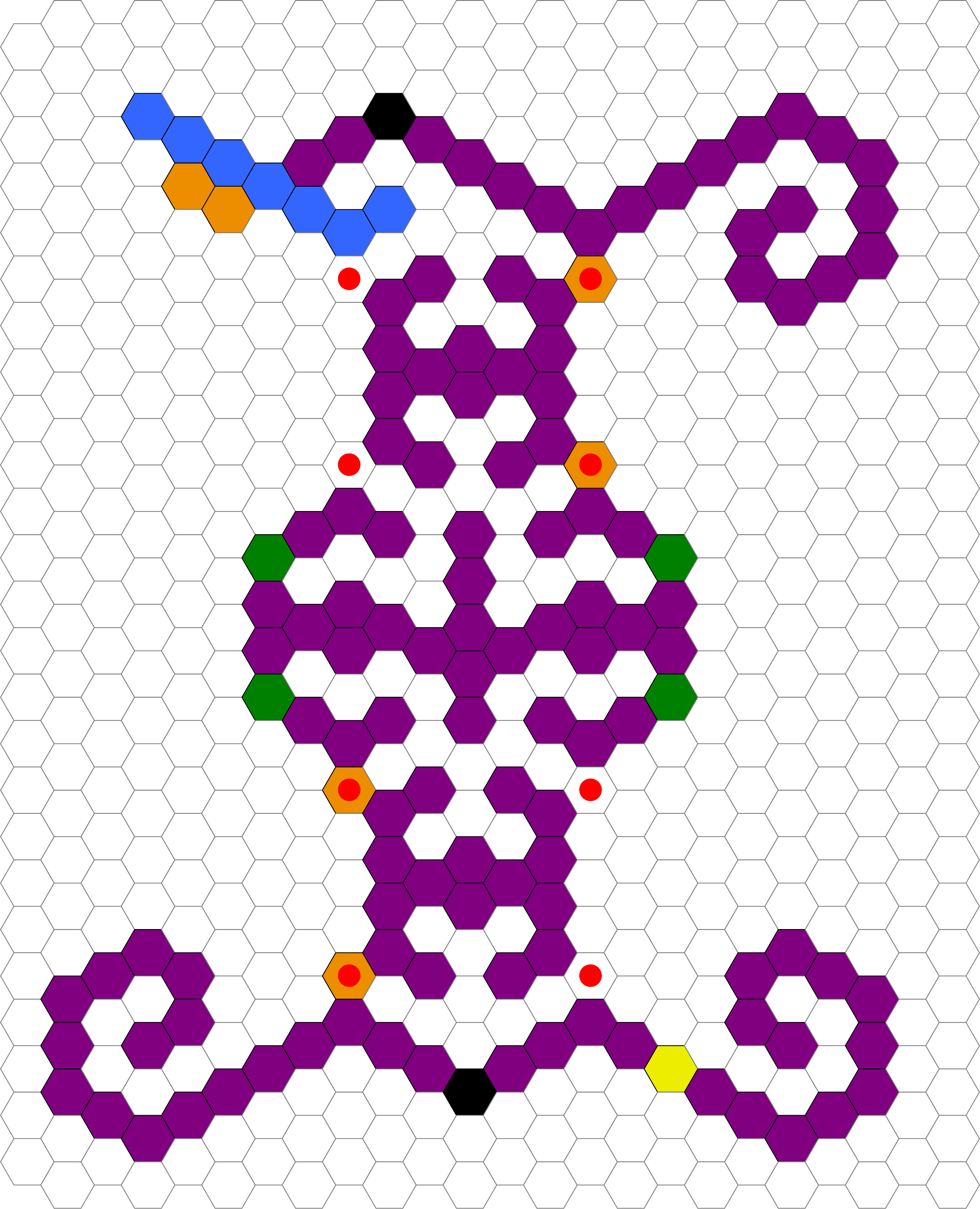}
    \end{minipage}
    \qquad
    \begin{minipage}[t]{0.4\textwidth}
        \centering
    	\includegraphics[width=\linewidth]{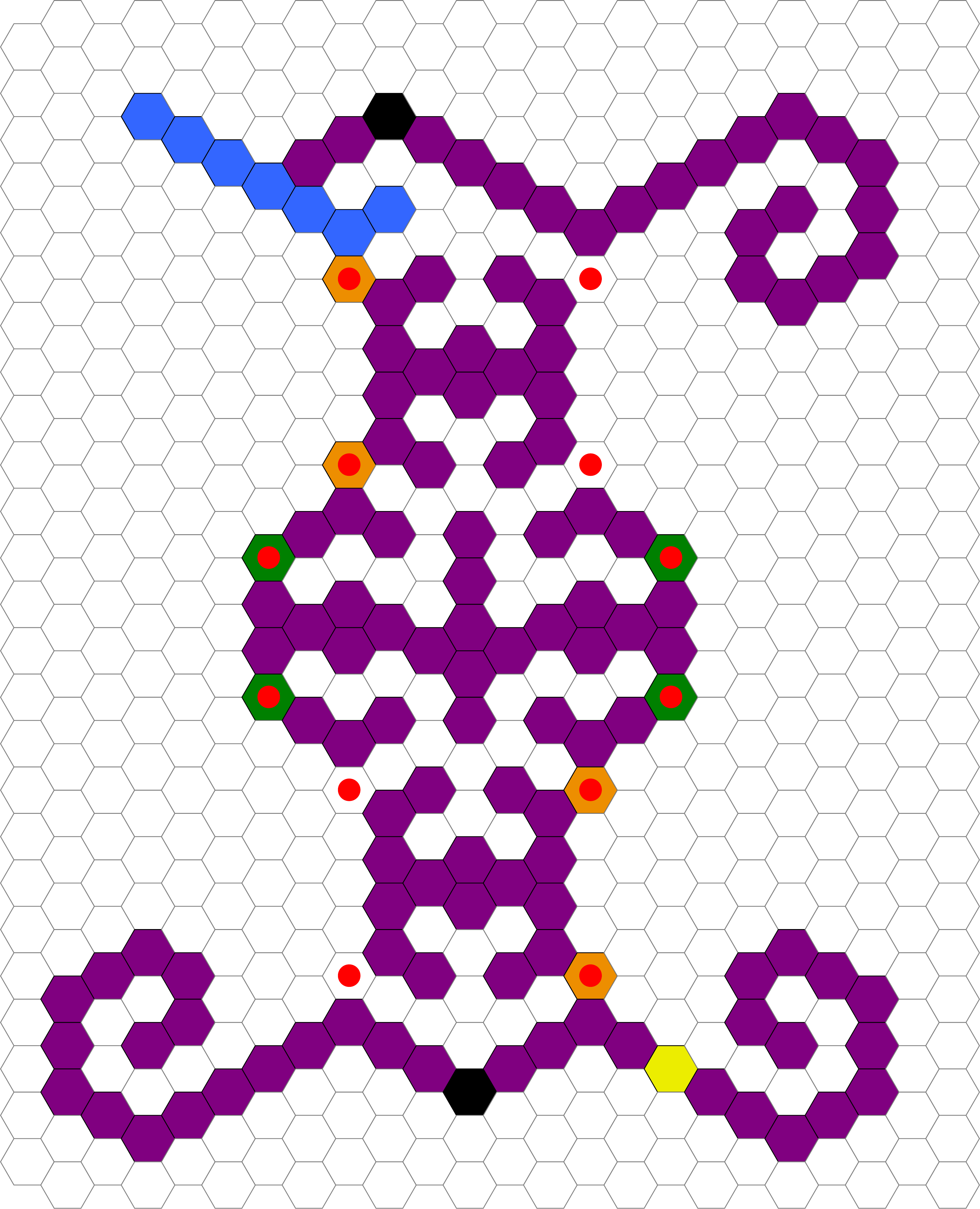}
    \end{minipage}
    \caption{(left) Win gadget in the inactivated state. The protected corners on the left are leaning up, while the ones on the right are leaning down.
    (right) Once the agent has placed itself in the highlighted positions, the partner modules in the right can move. They can traverse down, reach the goal destination, and allow the two other highlighted modules on the bottom left side to return to the wire. The overall result is that the location of 4 modules has changed, and the four corners have changed their state. We call this situation the finished state.}
\label{fig:hexwin}
\end{figure}

This gadget receives a single wire and must change its state if the agent can reach it. We represent this gadget with a slightly modified version of the side switch gadget (see Figure~\ref{fig:hexwin} left). In the gadget there is a highlighted module (shown in yellow) that conceptually is the position that must reached by the agent. Note that this module has no special meaning or purpose: we will track if the agent can reach that yellow module by monitoring the position of nearby modules. 

\begin{lemma}
The win gadget can change its state if and only if it is reached by the agent.
\end{lemma}
\begin{proof}
As usual in our constructions, no module can move until additional modules enter the gadget. Further note that, since the win gadget is connected to a single wire gadget, it cannot be involved in a global cycle. Recall that we add $1$-gaps in the middle of all wire gadgets, thus we conclude that solo agents cannot reach the win gadget. 

Since the status of the gadget cannot change without extra modules and a solo module cannot enter the gadget, the only way in which this gadget can change its state is when 2 modules (i.e., the agent) reaches the gadget. In this case we have a similar situation as in the side switch gadget: if the modules place themselves on the spots designated with a red dot, the orange modules highlighted with a red dot on the other side of the gadget can move. These modules can proceed down and place themselves in the lower part of the gadget (see Figure~\ref{fig:hexwin} right). In turn, this allows the two other orange modules to move. These two modules can move upwards and return to the wire gadget. 

Overall, if two modules reach the gadget, they can change the position of a few modules and afterwards two different modules will exit the gadget. This simulates the changing the state of the win gadget.
\end{proof}
}

\subsubsection{Finishing steps}

\iffull Now that we have all pieces we can proceed to prove Theorem~\ref{theo_hexrestrictedhard}.\fi

\begin{proof} Our reduction follows the framework in~\cite{motionplanning2}. Given a problem instance for 1-toggle-protected motion planning with the locking 2-toggle, we embed in a way that all edges are drawn with polylines that are multiples of $60^{\circ}$, replacing gadgets with the corresponding module configurations (adding side switch and wire cut gadgets as needed, as well as $1$-gaps to all wire segments). Finally, we place two additional modules at the initial position to define the agent. Since each gadget takes constant space, the problem instance will have polynomial size. Our goal configuration is the same configuration with only one change (the state of the win gadget). 

If the problem instance is solvable, there is a way for the agent to reach the win gadget, change its state, and then return back to the initial position in the exact reverse path. By doing so we reset every gadget except the win gadget back to its original state and reaching the agent's original start position. If the problem instance is not solvable, the agent cannot to reach the win gadget and thus the reconfiguration problem will also be unfeasible.
\end{proof}

\subsection{Square modules with the restricted move model}\label{sec_squarereduc}

\iffull We now show the same hardness for square modules. \fi

\begin{theorem}\label{theo_squarerestrictedhard}
Given two configurations of $n$ square modules, it is PSPACE-hard to determine if we can reconfigure from one to the other using only restricted moves.
\end{theorem}

Our reduction is analogous to the hexagonal reduction. We quickly list the pieces and a small description for each, but for brevity the proof of correctness of each single gadget is removed. The arguments are analogous to the hexagonal counterpart and we present a full list of our gadgets in the Appendix in section~\ref{sec:sqgadgetsAppendix}.

\later{
\subsubsection{Gap, agents, wires, wire cuts and side switch gadgets} \label{sec:sqgadgetsAppendix}
As with the hexagonal reduction, the agent will be represented by two modules. The way we use to lock leaves is shown in Figure~\ref{fig:sqmoves}: the purple module marked with a dot is a leaf of the configuration, but will not be able to move if the other highlighted positions are occupied.

\begin{figure}[ht]
	\centering
	\includegraphics[width=0.4 \linewidth]{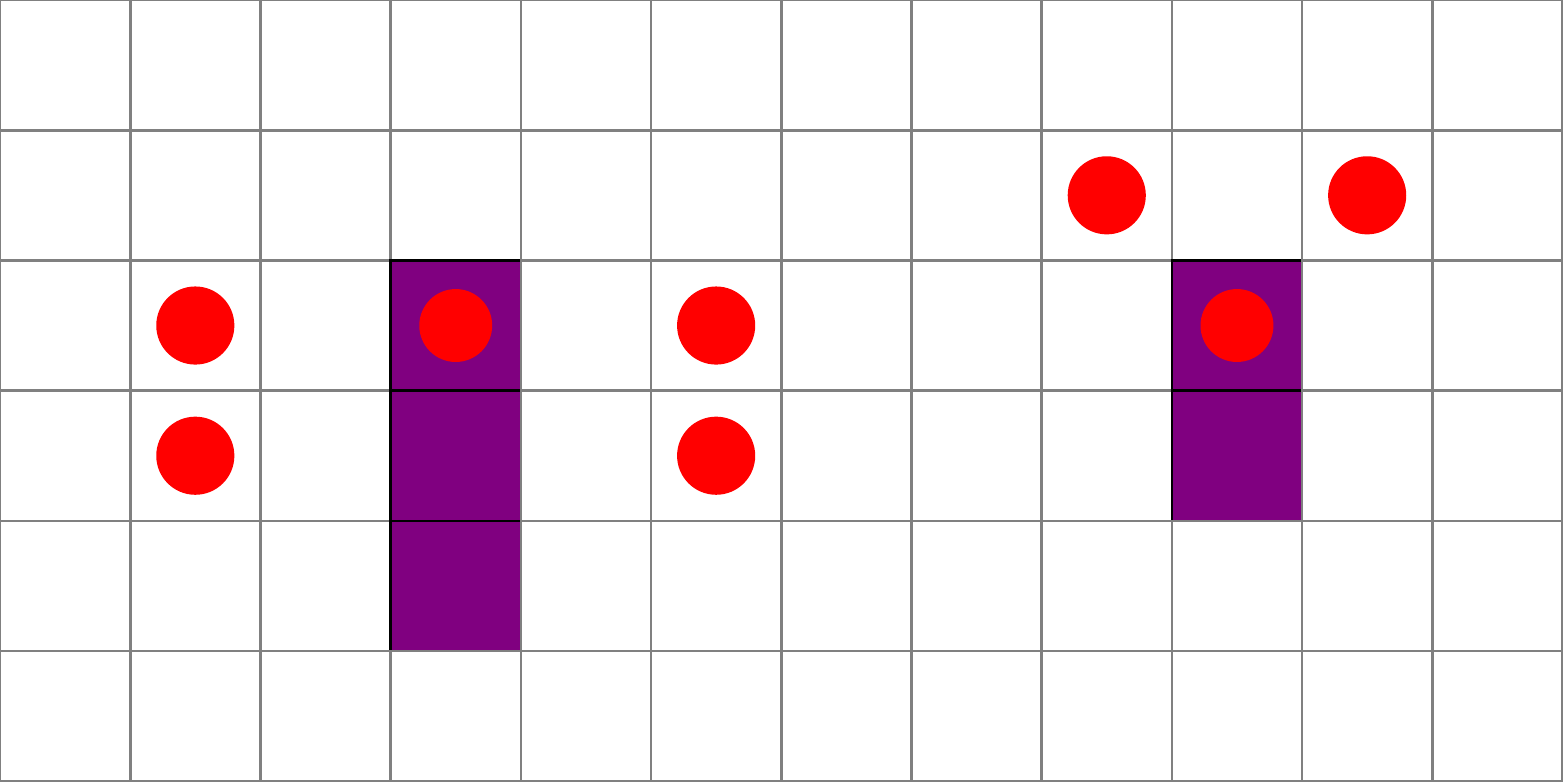}
	\caption{Two different ways in which we can prevent a square leaf module from moving under restricted model. (left) as long as two of the four positions are occupied (one on each side) the highlighted module will not be able to move. (right) We can also block a leaf from moving by occupying two diagonally adjacent locations.}
	\label{fig:sqmoves}
\end{figure}

As in the hexagonal case, we need some local way of preventing the agent from moving into certain locations. In only the restricted move is allowed, then a single hole suffices (see Figure~\ref{fig:sqgaps} left). As in the hexagonal case, we need a gap that can only be traversed by the two agents if there is anauxiliary module. This is shown in Figure~\ref{fig:sqgaps} right: two modules cannot cross the gap alone, but if we have an auxiliary module at the bottom, the agent can move from the top down (and in reverse direction), but one module must always remain attached to the lower edge (it is not possible to move all three modules to the upper portion).

\begin{figure}[ht]
	\centering
	\includegraphics[width=0.55 \linewidth]{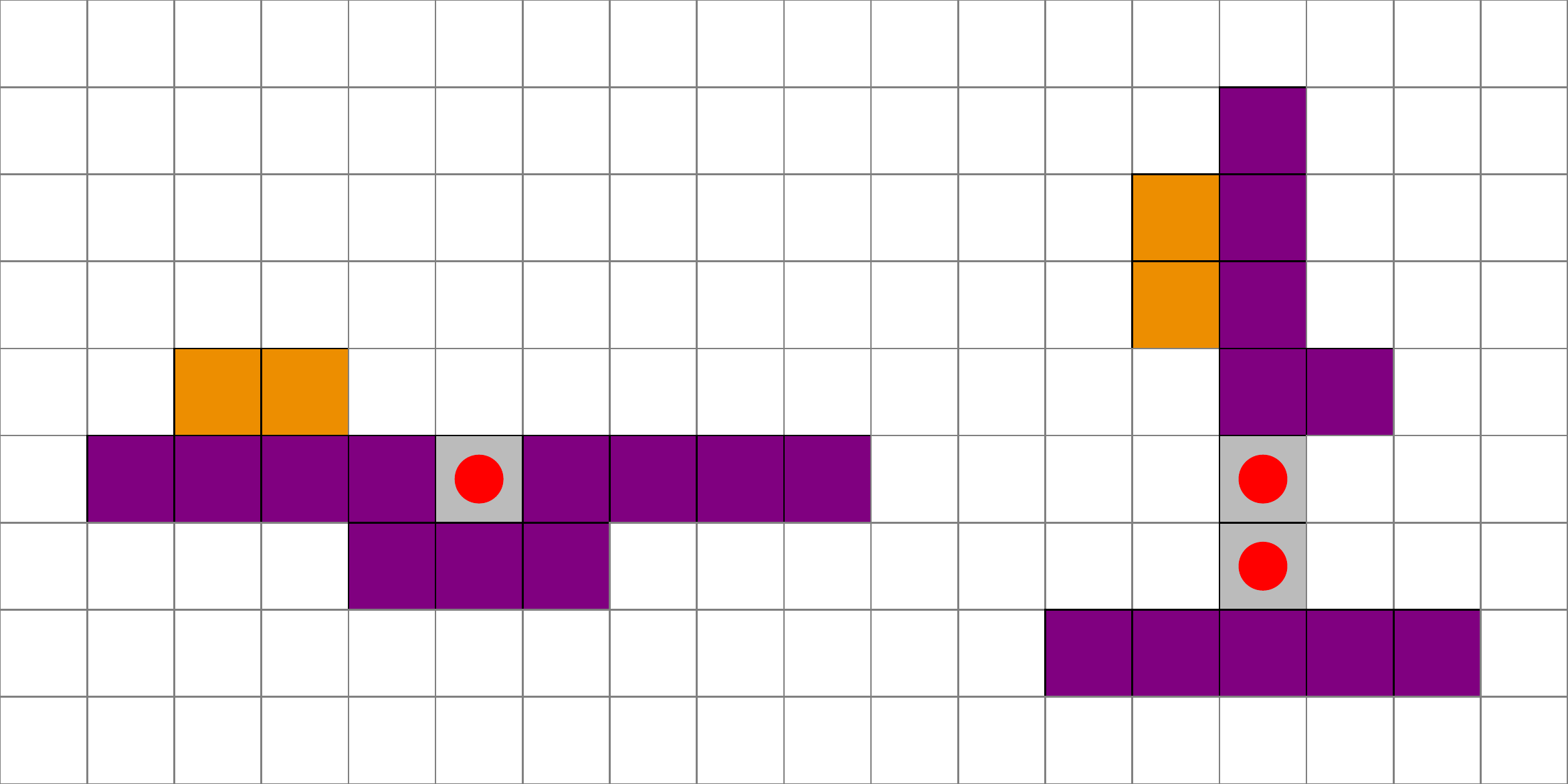}
	\caption{A possible way in which we can represent gaps in the square model. The left image shows a gap that cannot be crossed with 2 modules (in the restricted model). The right image shows a gap that cannot be crossed unless an auxiliary third module is present.}
	\label{fig:sqgaps}
\end{figure}

Wires are also represented by a sequence of modules in a line. This time, we will draw isothetic edges (i.e., either horizontal or vertical). Protected and blocked corners are shown in Figure~\ref{fig:sqcorners}. As it happens with the hexagonal case, blocked corners cannot be traversed by the agent, and protected ones can be traversed while at the same time enforce the fact that the gadget acts as a 1-toggle (see Figure~\ref{fig:sqcornerstates} for the possible states of the protected corner).

\begin{figure}[ht]
	\centering
	\includegraphics[width=0.65\linewidth]{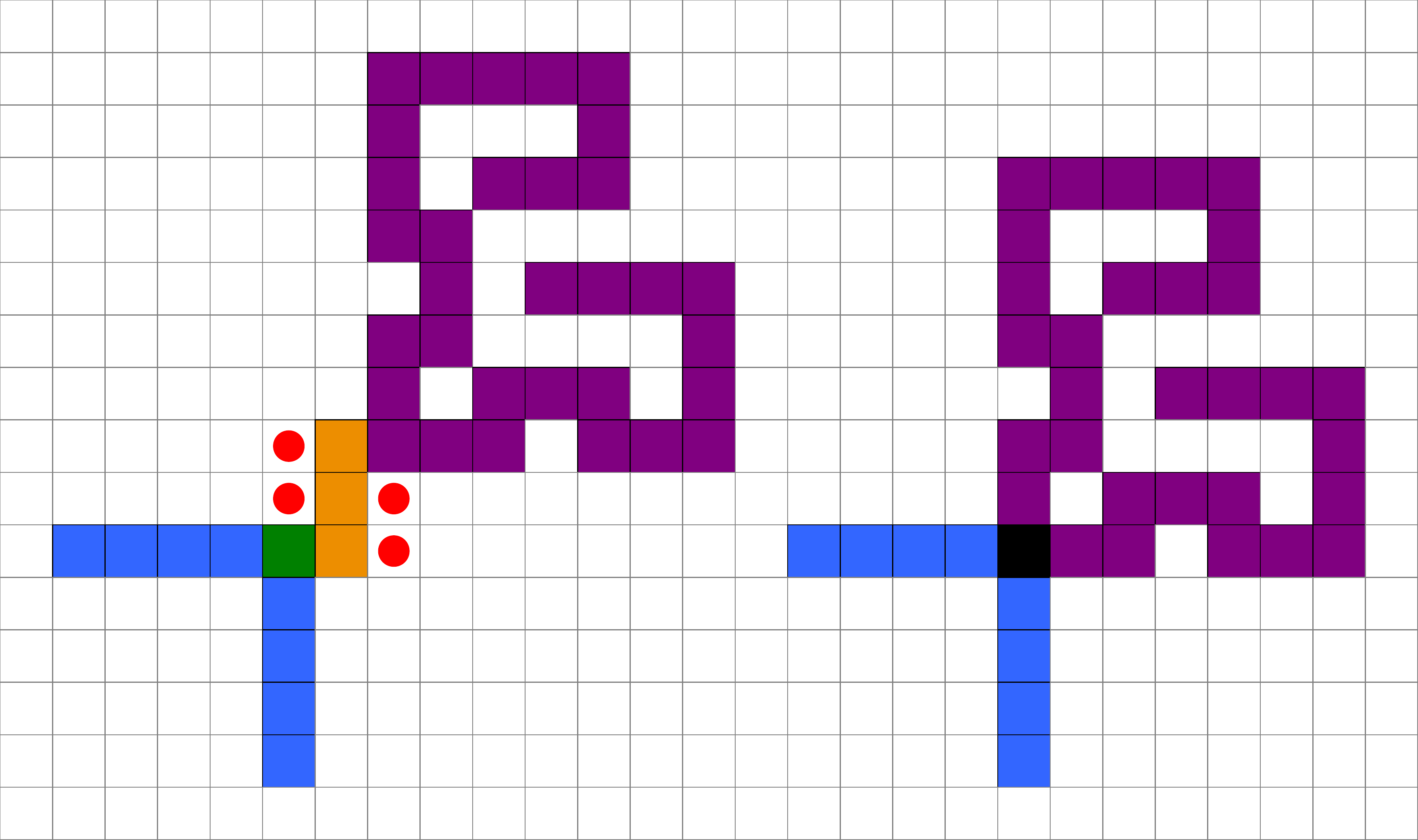}
	\caption{Examples of protected and blocked corners. The positions highlighted with red dots are the places in which an agent can place itself to traverse the gadget.}
	\label{fig:sqcorners}
\end{figure}

\begin{figure}[ht]
	\centering
	\includegraphics[width=0.65\linewidth]{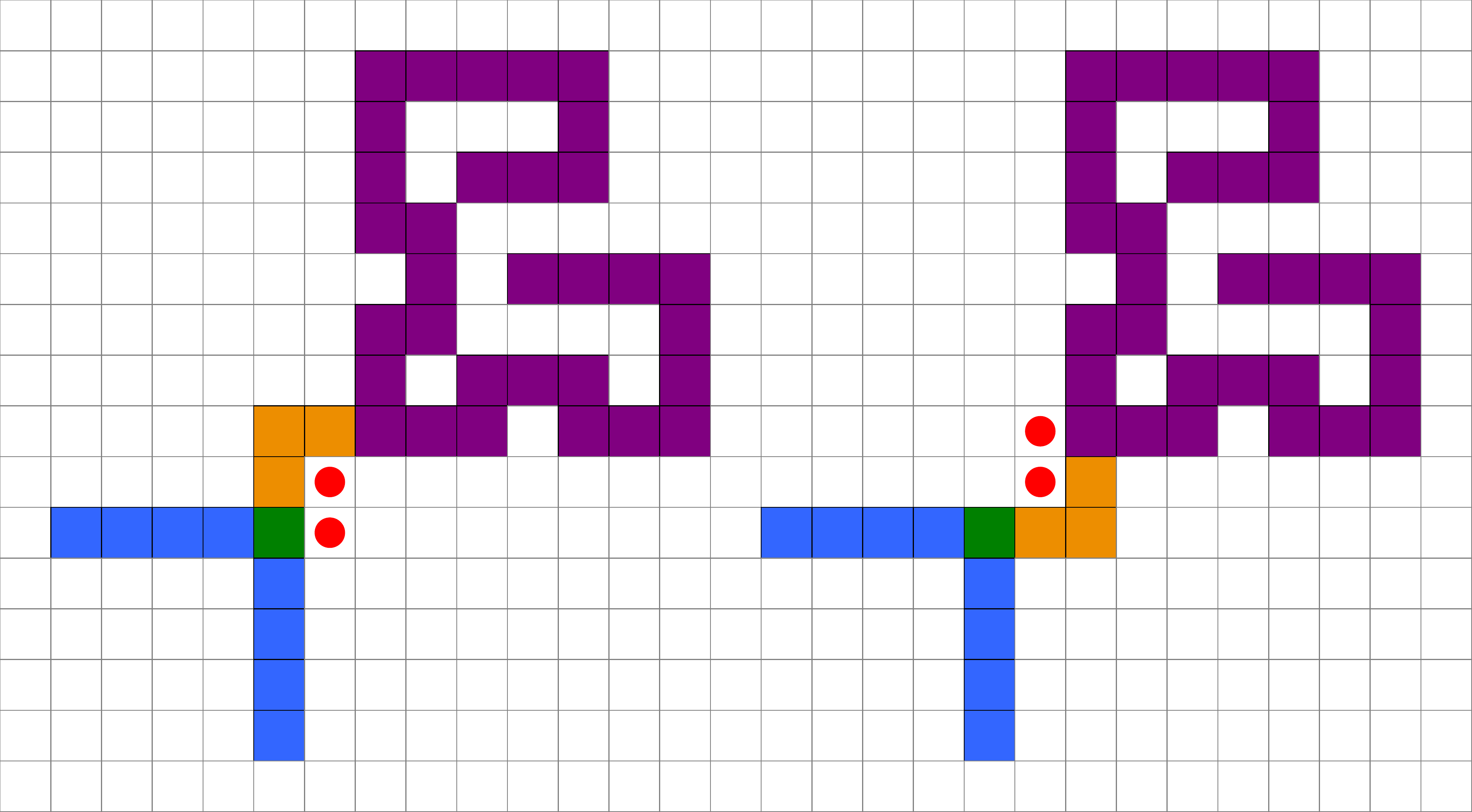}
	\caption{The two other possible states that the protected corner can be in (left and right leaning, respectively).}
	\label{fig:sqcornerstates}
\end{figure}

A possible way to represent wire cuts is shown in Figure~\ref{fig:sqwirecut}. As in the hexagonal case, when the agent and the auxiliary module make bridge, they also form a cycle. The cycle is local (if the modules form a bridge on the right side) or global (as done in our example). A global cycle can cause an auxiliary module at another location to move (a solo agent). 

\begin{figure}[ht]
	\centering
	\includegraphics[width=0.7 \linewidth]{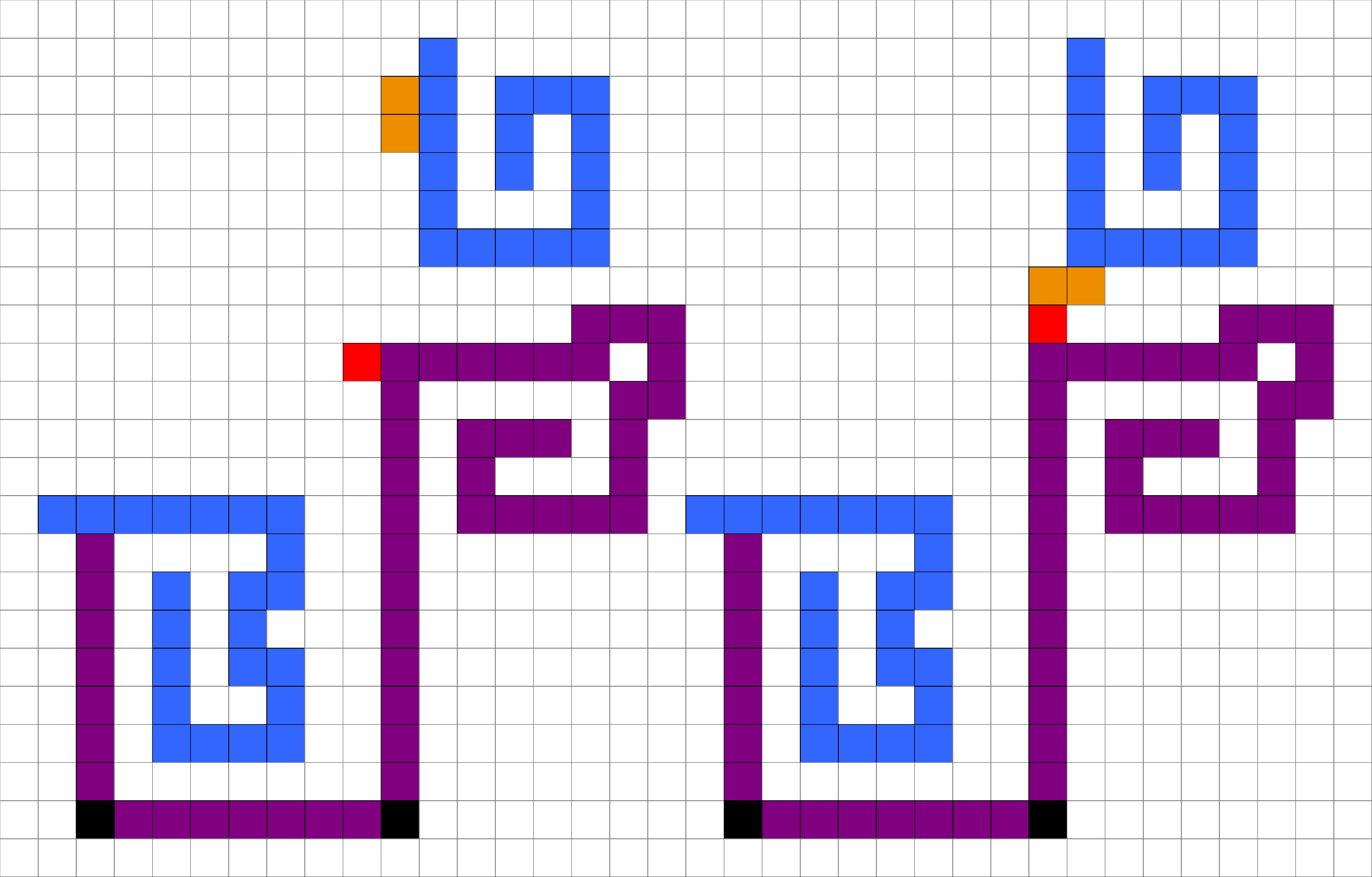}
	\caption{(left) Wire cut gadget with the auxiliary module highlighted in red and the agent coming from the right side. (right) the two agent modules and the auxiliary module can form a bridge so that the three modules get onto the wire cut gadget.}
	\label{fig:sqwirecut}
\end{figure}

As in the hexagonal case, we need a way to prevent this solo agent to move between gadgets. In the hexagonal model we used a $1$-gap, but there is no such equivalent when considering square modules. Instead, we need to provide a slightly larger and more complicated gadget similar to the wire cut gadget (shown in Figure~\ref{fig:sq1gap}). 

\begin{figure}[ht]
	\centering
	\includegraphics[width=0.45\linewidth]{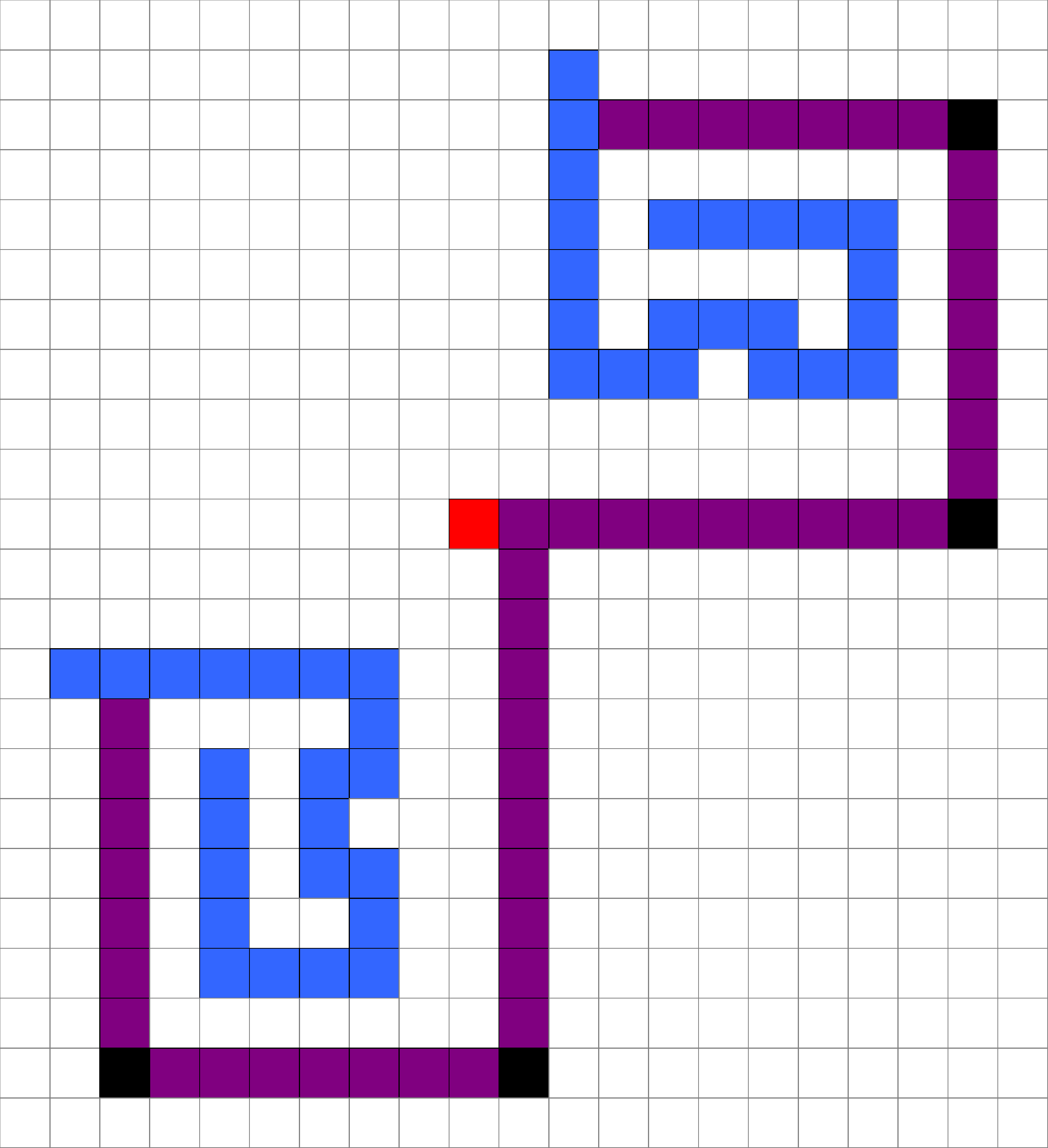}
	\caption{An example of a construct in the square model that is equivalent to a $1$-gap. A single module cannot cross it, but a $2$-module agent can. Notice how the gadget contains an auxiliary red module. One module must remain within the gadget at all times.}
    \label{fig:sq1gap}
\end{figure}

\begin{lemma}
The gadget shown in Figure~\ref{fig:sq1gap} acts as the equivalent of $1$-gap in the hexagonal model. That is, a solo agent cannot traverse it, but two agent modules will be able to. After the agent has traversed the gadget, a single auxiliary module must remain within.
\end{lemma}
\begin{proof}
Recall that this gadget is added to each segment. Thus, in addition to what is shown in the figure, we must add three bends (with protected corners) to make sure that the wire gadget is a straight segment.

The behavior is very similar to the wire cut gadget. Indeed, the two gadgets are almost identical. The only difference is that this gadget is connected whereas the wire gadget is not. This construction contains an auxiliary module that can move (shown in red). Because of the model limitations, this module can be diagonally adjacent to the blue modules (that are part of the wire gadget), but cannot attach to them. A solo agent arriving from either of the blue wires is stopped from reaching the corners because of the indentations. Thus, no significant change can happen.

The situation changes if two modules (rather than one) reach the gadget either of the wires. Akin to the process shown in the wire gadget (see Figure~\ref{fig:sqwirecut}), the auxiliary module together with the agent can form a bridges and leave the wire gadget. By doing the same operations in reverse, they can move onto the other half of the wire gadget and traverse through this gadget. Note that when the three modules form a bridge, a cycle is created. Unlike in the wire cut gadget, both cycles are local and no module other than the bridge can move.
\end{proof}

The wire switch gadget is shown in Figure~\ref{fig:sqwireside} and is very similar to the hexagonal case. The agent can enter from either side, attach itself to the highlighted red dots and allow the partner modules (shown in orange with a red dot) on the other side of the gadget to move. In the process we change the state of protected corners, thus enforcing that the gadget acts as a 1-toggle.

\begin{figure}[ht]
    \begin{minipage}[t]{0.5\textwidth}
        \centering
        \includegraphics[width=0.8\linewidth]{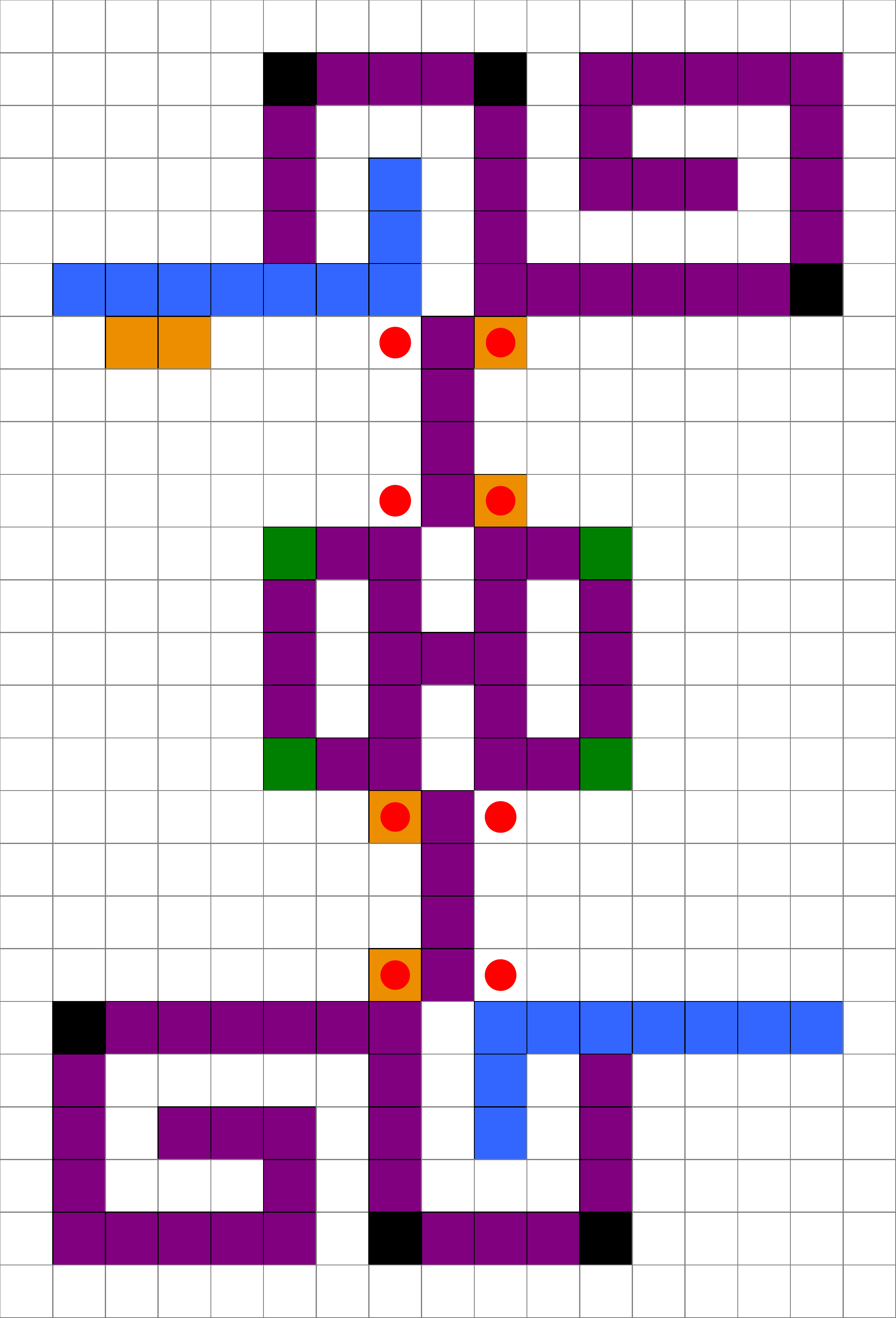}
    \end{minipage}
    \begin{minipage}[t]{0.5\textwidth}
        \centering
    	\includegraphics[width=0.8\linewidth]{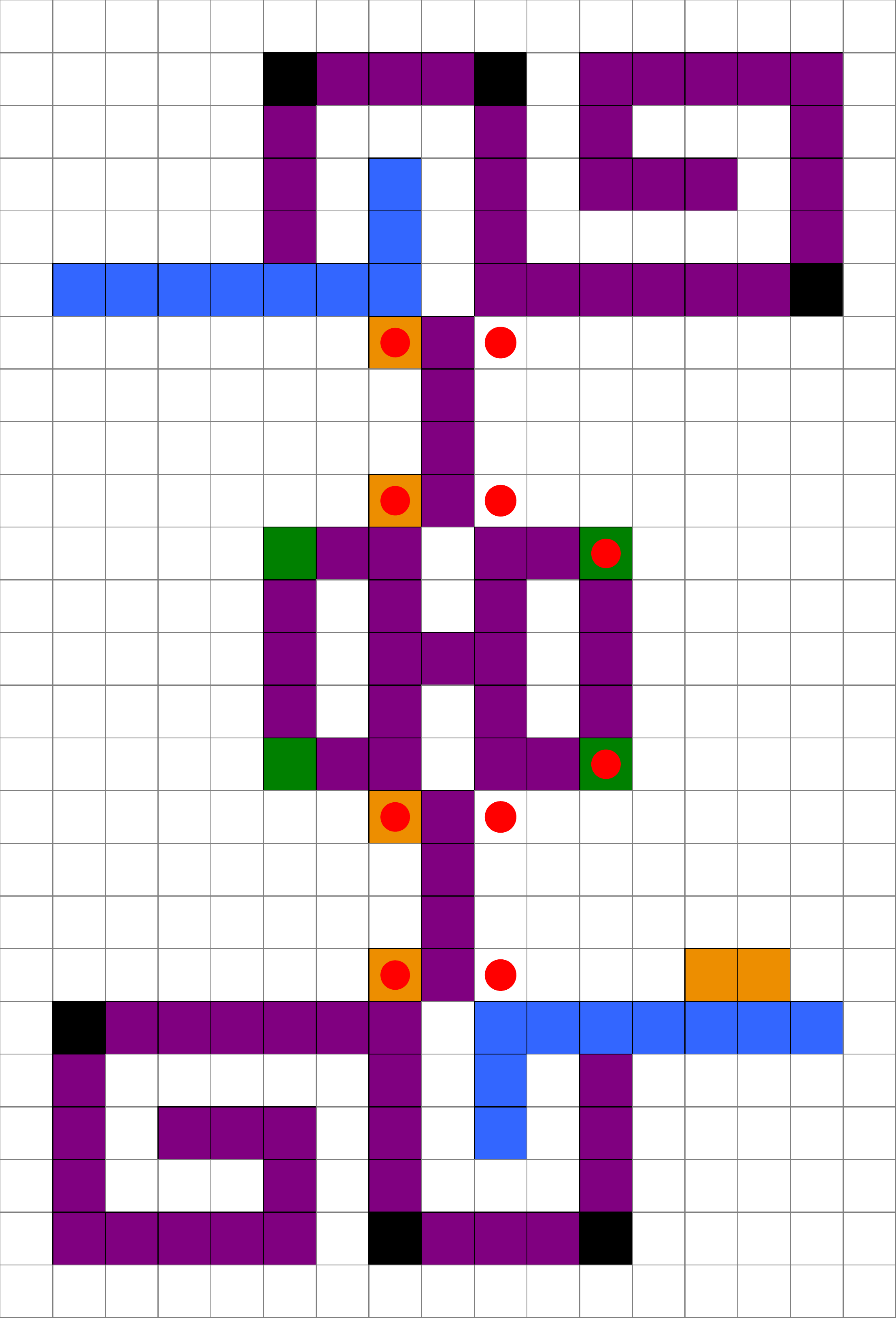}
    \end{minipage}
    \caption{(left) Wire side switch gadget. If the agent (shown in orange) enters the wire below the upper half of the wire, then it can move to the nearest highlighted red dots. This allows the orange modules (highlighted with a dot) on the other side of the gadget to move.
    (right) these two modules can move down and exit the gadget on the other side of the wire. As in the hexagonal case, we see it globally as the gadget changing state and the agent changing sides.}
\label{fig:sqwireside}
\end{figure}
}

\later{\subsubsection{Branching hallway and 2-toggle gadgets}}

\both{ The branching hallway gadget is show in Figure \ref{fig:sqbranch}} and works like the hexagonal version.

\later{
It works in a similar way as the wire side switch gadget and the hex branching hallway. There are four pairs of critical positions, and two must always remain occupied for connectivity. When the agent enters we have the option of swapping the positions that are occupied, provided that it can reach an empty position. As in the hexagonal case, this simulates a 1-toggle in both of the left wires. We can add a 1-toggle in the left wire by adding a protected spiral (not shown in the figure).}

\begin{figure}[hbt]
	\centering
	\includegraphics[width=0.3\linewidth]{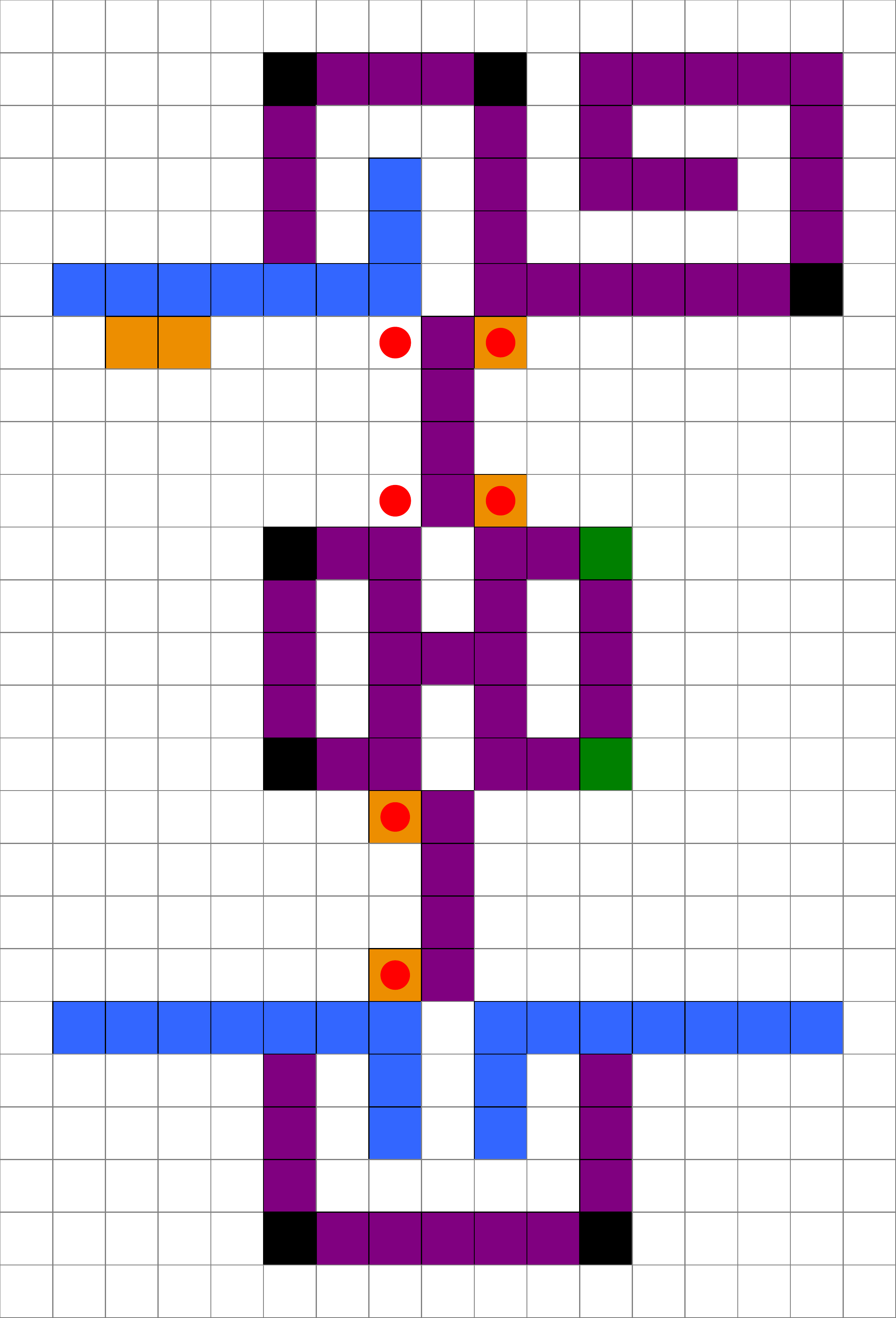}
	\caption{The branching hallway gadget for squares under the restricted model.}
	\label{fig:sqbranch}
\end{figure}

\both{The L2T gadget in the open state can be see in Figure \ref{fig:sql2t}.} Again, this gadget has the same exact functionality as its hexagonal counterpart.
The reduction works the same and the proof for Theorem~\ref{theo_squarerestrictedhard} follows a similar format as Theorem~\ref{theo_hexrestrictedhard}. The gadgets presented here together with the details in the Appendix complete the proof of Theorem~\ref{theo_squarerestrictedhard}.

\later{
Conceptually, it is identical to its hexagonal counterpart: 8 modules within the gadget are split into 4 groups and can conceptually move (but not all at a time). When the agent enters the gadget from one of the top edges, it can use the auxiliary red module on the same side to bridge over to the gadget, then bridge over to the almost disconnected component on the side. This allows the 3 modules on the bottom left (the two orange and the red) to move and then bridge over to the bottom left wire. Two of the three modules will proceed downwards.
}

\begin{figure}[ht]
    \centering
    \includegraphics[width=\linewidth]{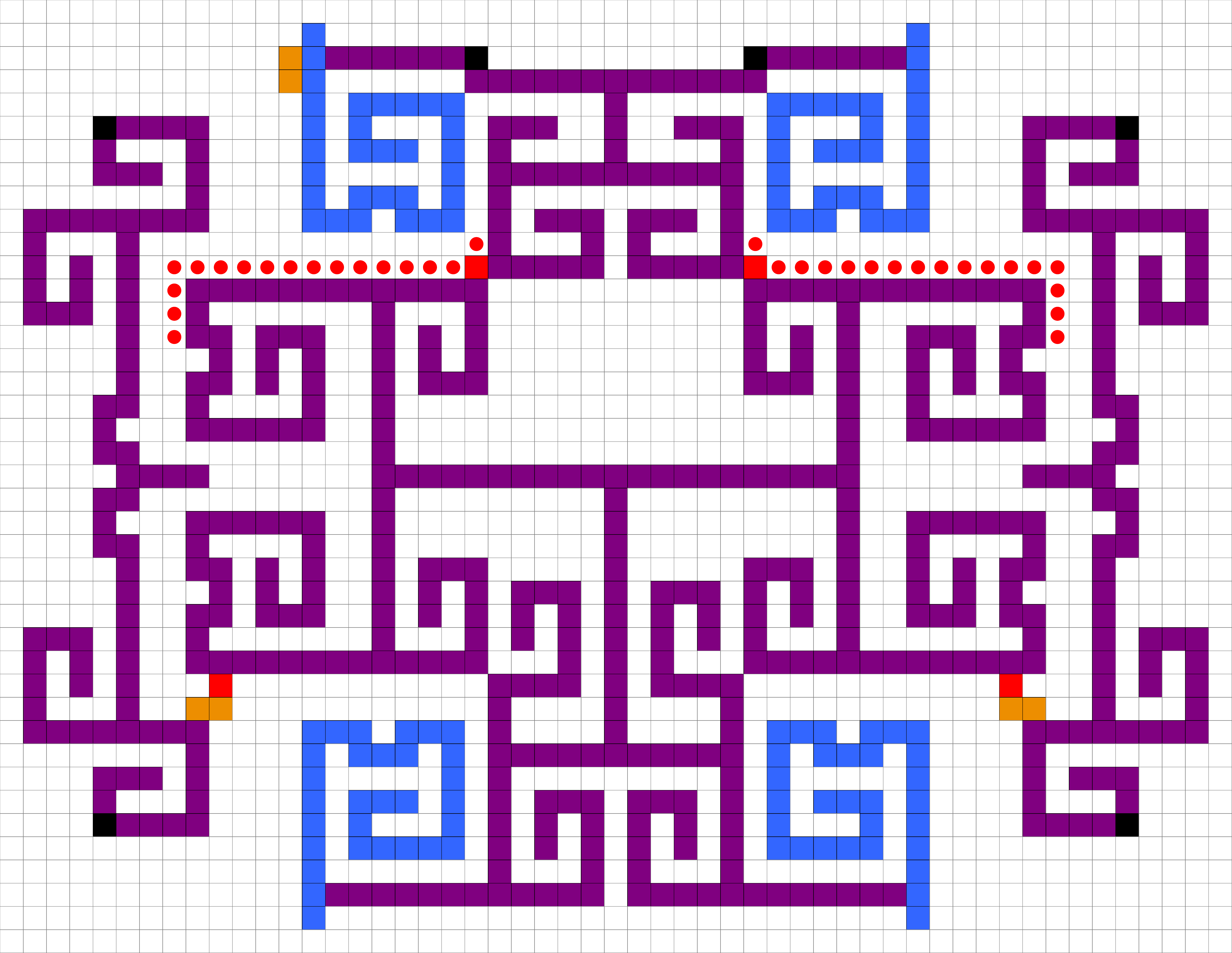}
    \caption{L2T in the open state. At this point the gadget can come from either of the top wires.}
    \label{fig:sql2t}
\end{figure}

\later{
After this traversal, the L2T will be in the state seen in \ref{fig:sql2t2}. Due to the same reasons as in the hexagonal case, this state of the L2T can only be traversed from the bottom left by undoing all of our previous moves. It cannot be crossed from any other wire.

\begin{figure}[ht]
    \centering
	\includegraphics[width=\linewidth]{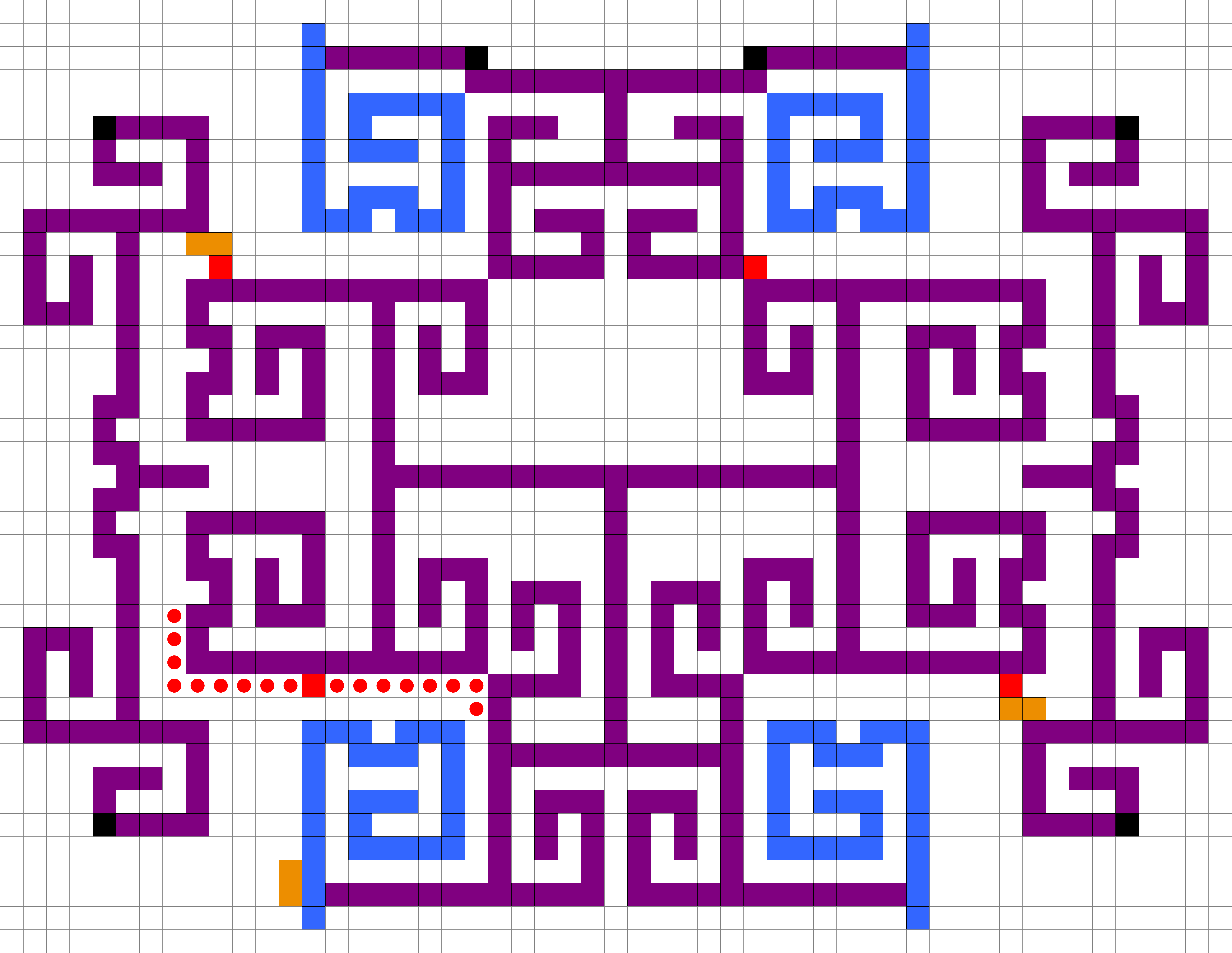}
	\caption{L2T in the closed state, traversable only from the bottom left upwards.}
    \label{fig:sql2t2}
\end{figure}

\subsubsection{Win gadget}

Finally, the win gadget is shown in Figure~\ref{fig:sqwin}. Just like in the hexagonal case, it is a slightly modified side switch gadget. The gadget is attached to a single wire and can only change its state when an the agent enters it. The traversal is done by the agent placing itself in the two red dotted positions at the top left of gadget, freeing the two orange red dotted modules at the top right. These two modules can now move down, cross the two protected corners and reach the target yellow modules. After that, these place themselves in the two red dotted positions at the bottom right, freeing the orange red dotted modules at the bottom left. These modules can now cross the other two protected corners and leave along the same wire the agent entered. The end state of the win gadget is shown in Figure \ref{fig:sqwin}, right.

\begin{figure}
    \centering
    \begin{minipage}[t]{0.4\textwidth}
        \centering
        \includegraphics[width=\linewidth]{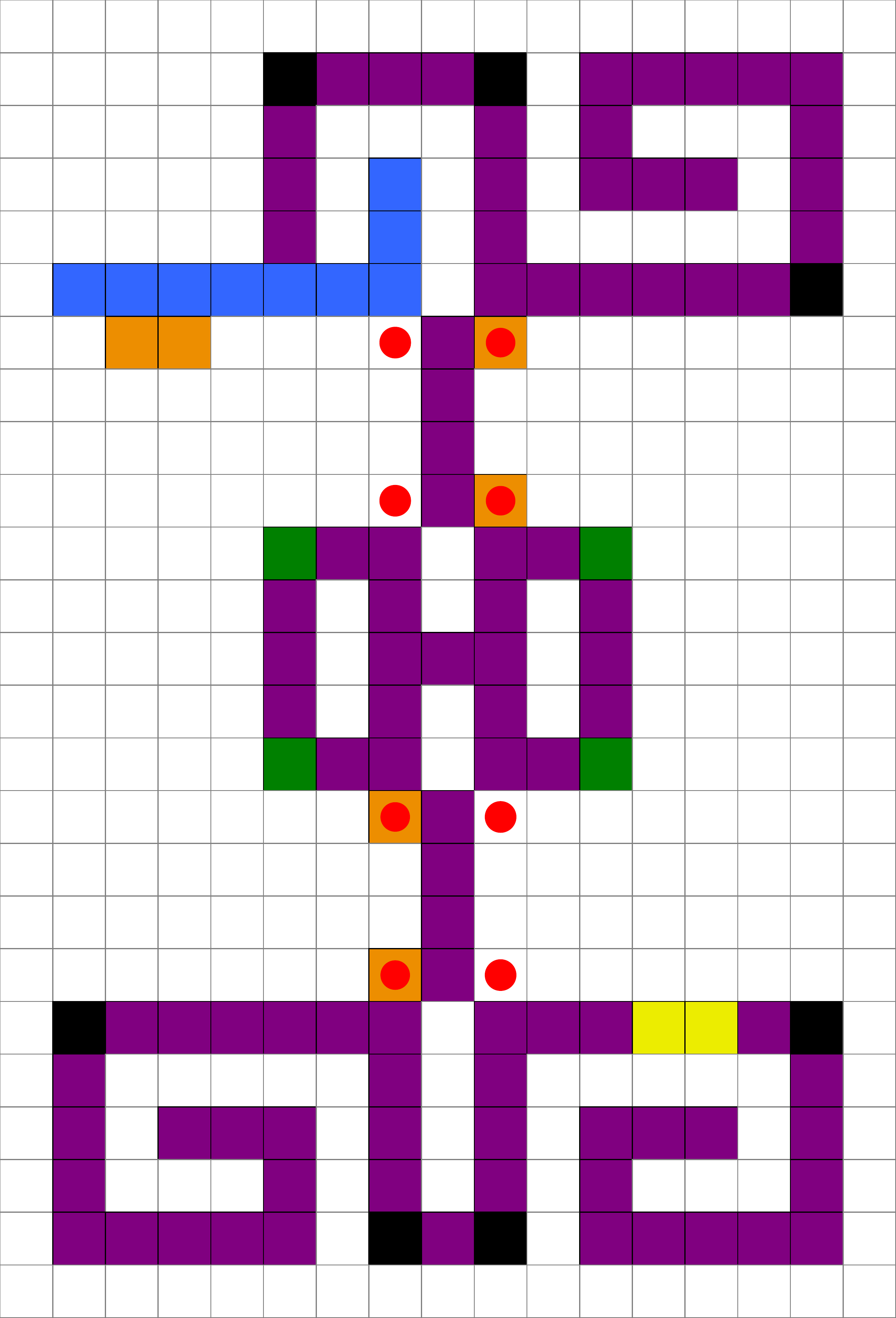}
    \end{minipage}
    \qquad
    \begin{minipage}[t]{0.4\textwidth}
        \centering
    	\includegraphics[width=\linewidth]{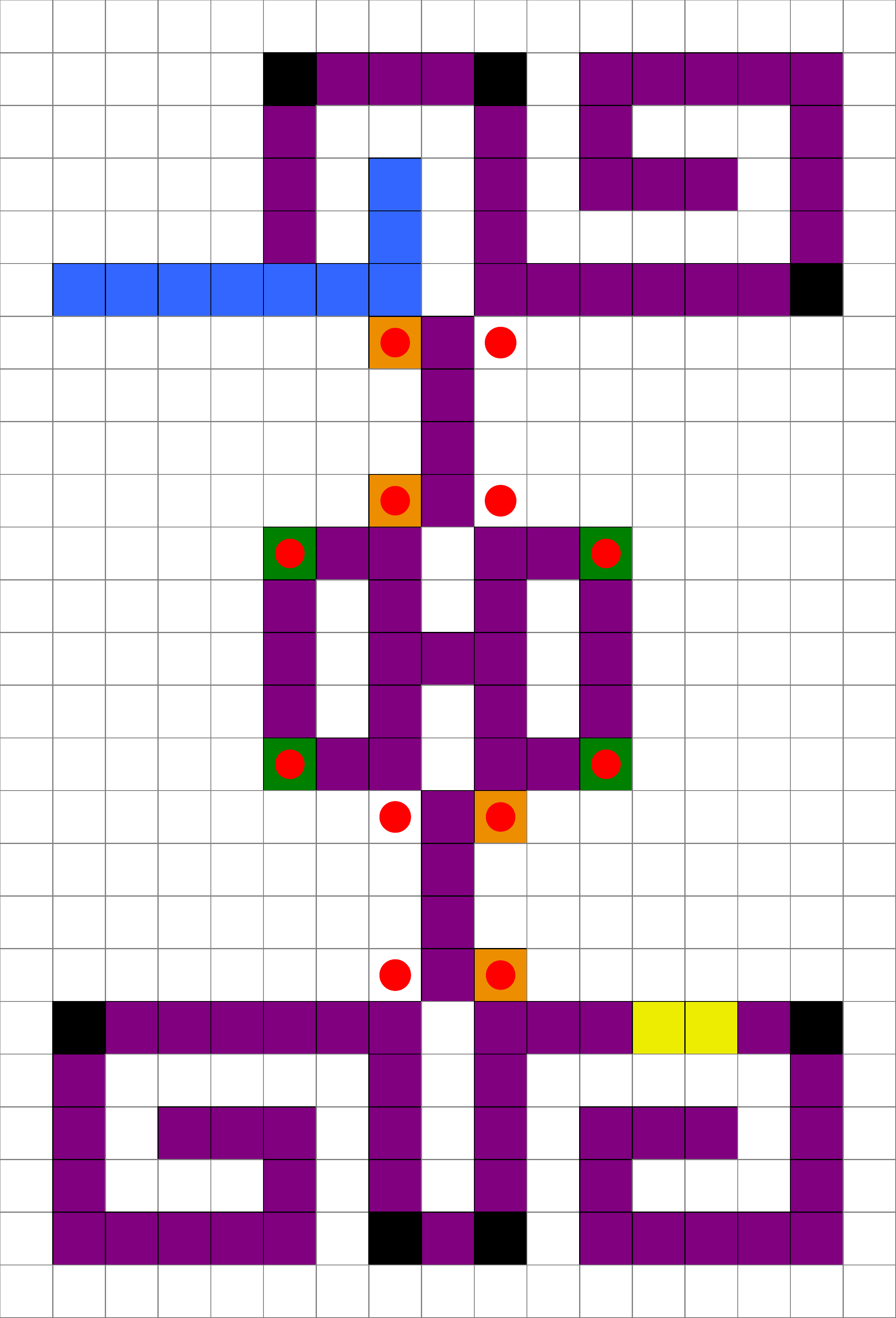}
    \end{minipage}
    \caption{(left) Win state gadget in the inactivated state. The yellow modules denote the goal position that must be reached by the agent. As in the hexagonal case, rather than tracking whether or not the agent has reached this position we check the position of nearby modules.
    (right) Once the agent has placed itself in the highlighted positions, we can move around the modules marked with a red dot. The overall result is that the location of 4 modules has changed, and four spiral gadgets have changed the state. We call this situation the finished state.}
    \label{fig:sqwin}
\end{figure}
}

\subsection{Hardness for the square model for monkey and leapfrog models}
\label{sec:PSPCAE-sq-lf-monkey}

Our final reduction applies to both remaining models for square modules.

\begin{theorem}
\label{thm:square-monkeyleap}
Given two configurations of $n$ square modules, it is PSPACE-hard to determine if we can reconfigure from one to the other. This reduction holds both under the the monkey and the leapfrog model.
\end{theorem}

The reduction is also from 1-toggle-protected motion planning with the locking 2-toggle, but simpler. The main differences are as follows:

\begin{itemize}
    \item A leapfrog move can pass trough obstacles or bends without creating global cycles. All the cycles created by the agent module are local, with size at most $8$, which allow us to have purely local arguments. 
    \item Because of this change, we can now represent the agent with a single module. This eliminates the need to prove that multiple modules have to work together (and all other intricacies related to the case of a $2$-module agent). 
    \item Another interesting advantage is that we can represent a wire with two parallel sequences of modules (5 units apart). The agent will move between the two lines, which reduces the need of worrying about which side the agent is on.
    \item Finally, the reduction works for the leapfrog model, but even if we allow monkey moves the result holds. Thus, a single reduction will work for both models.
\end{itemize}

\begin{figure}[htb]
	\centering
	\includegraphics[width=0.8\linewidth]{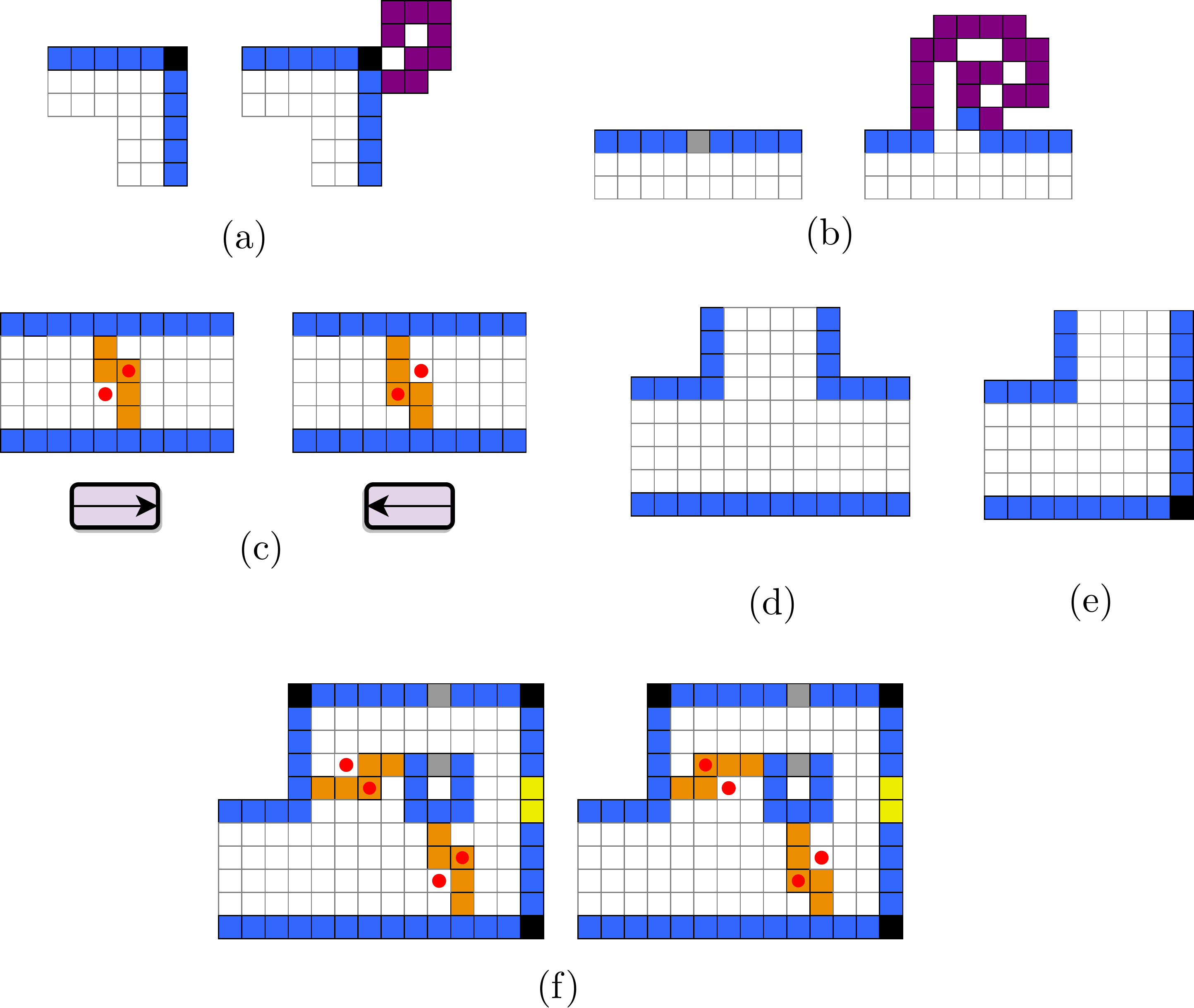}
	\caption{Gadgets used in PSPACE reduction (for leapfrog and monkey models).}
	\label{fig:monkey-gadgets-square}
\end{figure}

\later{As illustrated in Figure~\ref{fig:square_moves}~(c), in order for a module to be movable under the leapfrog/monkey model it must have one of its vertices incident to 3 empty positions. We also use the fact that modules that are cut vertices cannot move or they would break the connectivity of the configuration. Those two properties will serve to show that most modules cannot move. Whenever this happens we say that the module is \emph{blocked}.

We now describe our gadgets. We start with the wire to give some intuition about how the agent will move in our construction.
The wire gadget is shown in Figure~\ref{fig:monkey-gadgets-square}~(c) and consists of two parallel paths of modules at distance $5$ apart. We call those paths \emph{walls}.
The agent will move between the walls of the wire. We distinguish empty positions in our gadgets between \emph{interior} and \emph{exterior}. In our figures, the interior positions are shown as white squares (exterior positions are not shown).

We will also use {\em blocked corners} and {\em wire cuts} (shown in Figure~\ref{fig:monkey-gadgets-square}~(a) and Figure~\ref{fig:monkey-gadgets-square}~(b), respectively). These constructions work like in previous reductions to either prevent a corner from moving or to prevent cycles, respectively. Note that both gadgets insert modules in the exterior that cannot move. Moreover, if the agent were to pass through the interior part of either construction it does not create any global or local cycle. In the remainder of the construction we depict blocked corners with black and wire cuts in gray.

In each wire gadget, we will place a \emph{1-toggle} gadget, shown in Figure~\ref{fig:monkey-gadgets-square}~(c).
It consists of a path of $5$ orange modules connecting the walls of a wire. Note that this modules will create a cycle, which we break using wire cuts. Each of the toggles is shown in the figure. 

The \emph{branching hallway} and \emph{turn} gadget are shown in Figure~\ref{fig:monkey-gadgets-square}~(d) and (e), respectively. The branching hallway allows the agent to continue along either of the two other directions whereas the turn forces the agent to change its direction.

\begin{lemma}
The wire together with the branching hallway and turn gadgets properly simulate the movement options of an agent in 1-toggle-protected motion planning with the locking 2-toggle for both the monkey and leapfrog models.
\end{lemma}
\begin{proof}
For each gadget, it is easy to observe the following properties: (i) an agent module in an interior position can only move to another interior position; (ii) without the agent module, the construction has no cycles using only modules of the gadget, and (iii) every degree-1 module is blocked.

Consider the case when the agent needs to traverse through a 1-toggle. There are 3 positions in which the agent coming from the left could go to create local cycles. Two of them place the agent adjacent to an orange and a blue module. In these positions the agent creates a local cycle of length 4, in which the only movable module is the agent. 

The third option is the one marked with a red dot. In it the cycle includes $3$ orange modules and the orange module with a red dot becomes movable and become the new agent. As in the other constructions, moving it breaks the cycle and renders the previous agent module immovable. Note how now the agent can be on either boundary of the wire. If the agent were to enter the gadget from the right, only the two first options are available and the agent can't pass through the gadget (as desired, since the state of the 1-toggle only allows traversals from left to right). 

Similarly, we can show the same properties for the branching hallway and the turn gadget. the only aspect worth mentioning is that if the agent comes from the bottom side of the wire and wants to go upwards in a branching path this action is not directly possible. In this case, we recall that all of the 3 adjacent wire gadgets are equipped with a 1-toggle gadget. The agent can return to that toggle and use it to cross to the other side.
\end{proof}



The \emph{win} gadget is shown in  Figure~\ref{fig:monkey-gadgets-square}~(f) (left side shows inactivated state and right side shows the activated state). Note that the gadget is not shown to scale because the details of the wire cut gadgets that are omitted. 

\begin{lemma}
The win gadget can change its state if and only if it is reached by the agent.
\end{lemma}
\begin{proof}
As in previous cases, it is easy to verify properties (i)--(iii). In particular, no module can move until the agent reaches the win gadget. In order to transition from the left to the right state of the gadget, the agent module must arrive from the left, place itself in the vertical 1-toggle gadget, move around the center part of the gadget an place itself on the top of the vertical 1-toggle gadget.
\end{proof}
}

\begin{figure}[htb]
	\centering
	\includegraphics[width=0.9\linewidth]{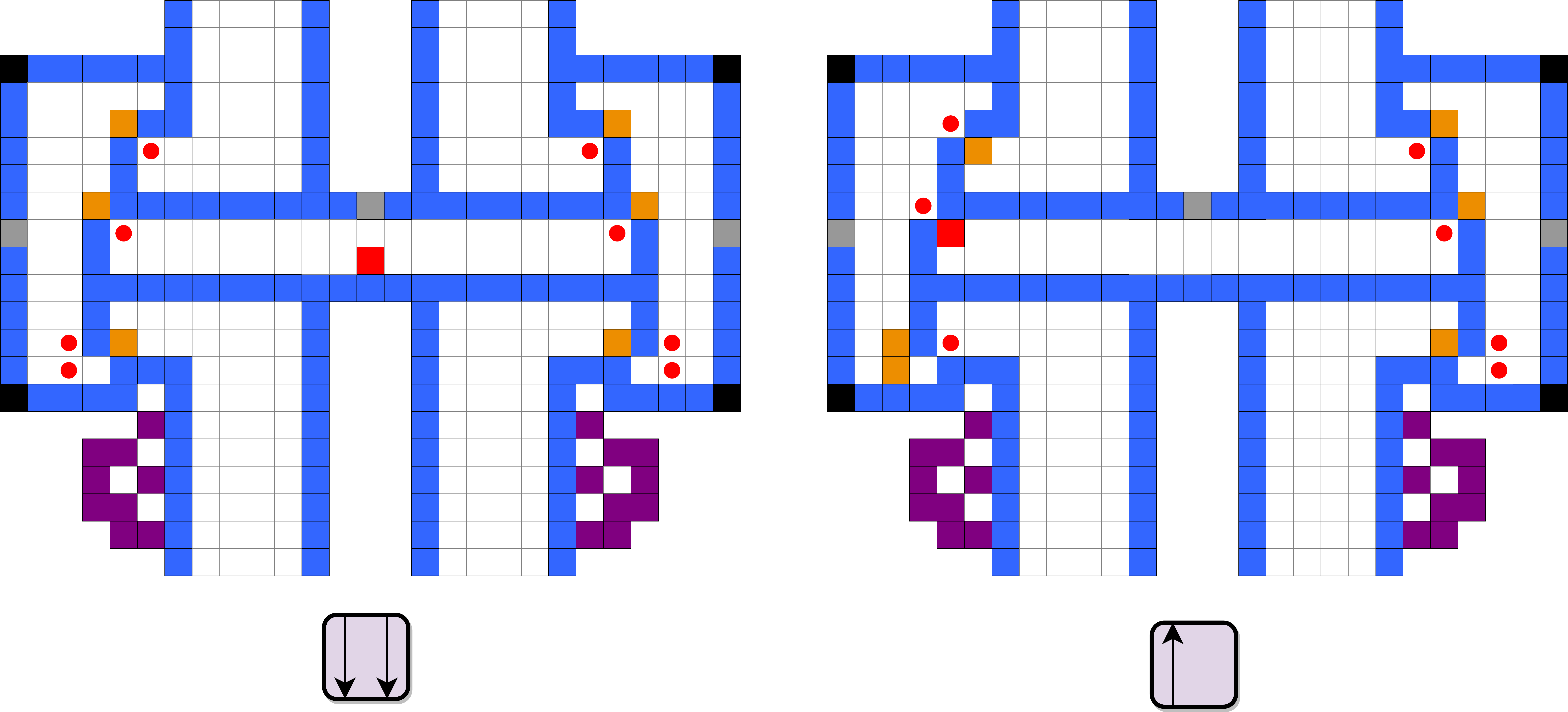}
	\caption{L2T gadget with square modules for  the Monkey model. In the figure two of the three possible states are shown (third one is symmetric).}
	\label{fig:l2t-monkey-square}
\end{figure}

\later{

Finally we present the \emph{L2T} gadget (Figure~\ref{fig:l2t-monkey-square}). Notice how all modules in this position are locked except a single auxiliary module (shown as red in the figure). We split the interior of the gadget into $7$ parts.
Four of them are connected to the four incident wires (where the agent can come from). Then, there is the central pocket containing the auxiliary module, and two side pockets containing two orange modules each.

\begin{lemma}
The construction presented in Figure~\ref{fig:l2t-monkey-square} properly simulates a L2T for square modules under both the leapfrog and monkey models.
\end{lemma}
\begin{proof}
First, let's analyze the configuration before the agent reaches it. Say we start in the configuration of Figure~\ref{fig:l2t-monkey-square}, left. The auxiliary module can only move within the central pocket, and among those positions only 4 create a local cycle. Two of them are useless since it is the only movable module in such position, the other two are marked with a red dot.

If the auxiliary module places itself at either of the two positions, then the diagonally adjacent module (shown as orange can move). Note that if orange moves, then the auxiliary module must remain in place. This single orange module has a similar situation, it can move around but it will only create meaningless cycles (where it is the single module of the cycle that can move).

The situation changes when the agent comes from either of the upper wires. Regardless of which of the upper wires it comes from it has three positions that can create cycles, but only one is meaningful (creates a cycle that allows the orange module to move). 

This orange module has the same limitations as before unless both the agent and the red auxiliary modules place themselves in the corners of the same side. In this case we have two orange modules that can move. In this case, the two together can form a small cycle towards the bottom of that pocket. This cycle allows the orange module in that in the bottom wire to move and proceed as the agent. Note that if the agent were to initially come from one of the bottom wires it would only produce meaningless cycles.

In short, if the agent were to come from either of the lower wires it would not change anything. If it instead comes from one of the top wires and the auxiliary module is free to move, together then can free two orange modules that allow a third module to continue in the matching wire. that third module becomes the new agent and we have successfully traversed the L2T and changed its state. 

Say that we changed the state of the gadget to the situation shown in Figure~\ref{fig:l2t-monkey-square} right. If the agent comes from either of the two upper wires it can only create a meaningful change in the upper right position (marked with a red dot). That will free one orange modules, but that single module cannot do a meaningful change (as before, both orange modules need to be around to make a cycle). However, if the agent comes from the lower left wire, it can undo all operations and reverse the configuration to the initial step.

Note how all of these arguments hold even if we allow the monkey move operation. Thus, the lemma holds for both models.
\end{proof}

With the L2T we have all of the gadgets needed for the reduction.
Combining all the gadgets, we obtain a configuration that might have global cycles. We can always add wire break gadgets to wires in order to break such cycles such that the module configuration forms a tree.

As in the other reductions, the only difference between the initial and target configurations is the state of the win gadget. If the agent can traverse through the instance and toggle the state of the win gadget, then it can reverse all moves and return all other gadgets to its initial state. This completes the proof of Theorem~\ref{thm:square-monkeyleap}.
}

\section{Conclusions}\label{sec:conclusion}
Although this paper settles the question of whether universal reconfiguration is possible for all models, it also spans several interesting problems. First, for hexagonal modules under the monkey model (where universal reconfiguration is possible), there is a gap between the  upper bound of our algorithm (Theorem~\ref{thm:alg}) and the naive  $\Omega(n^2)$ lower bound (number of moves needed to transform a horizontal strip into a compact hexagon). Even if the gap is closed (possibly with a completely different algorithm), then the interest would be to design a distributed algorithm and/or to consider a strategy that does many moves in parallel.

For the other models universal reconfiguration is not possible, but it would be nice to find a local property that would allow reconfiguration between many configurations. For example, if we allow monkey moves with square modules, reconfiguration is possible as long as both configurations have 5 modules on the outer shell that can move~\cite{musketeers}. We wonder if the concept of musketeers or crew can be extended to other models. 
We remark that for the hexagonal monkey model this technique cannot be directly applied. 
Indeed, a key step of that approach was that whenever there is no module that can freely pivot on the external boundary, 
we bridge such that 
a well-defined module $m$ (extremal for a potential function defined on the coordinates of the modules) 
can pivot along the external boundary without disconnecting the configuration. 
Moreover, 
as many modules as helpers used for bridging can be liberated locally around $m$. 
However, this is not always possible for hexagonal modules, as shown in Figure~\ref{fig:no_musket}: Three helper modules are be necessary for bridging, but only two can be liberated locally around $m$. 

\begin{figure}[ht]
	\centering
	\includegraphics[page=19,scale=0.85]{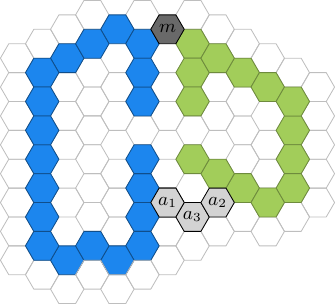}
	\caption{The strategy in~\cite{musketeers} cannot be directly translated to pivoting hexagons in the monkey model.}
	\label{fig:no_musket}
\end{figure}

\iffull
\paragraph{Acknowledgements.}
\ackn
\fi

\clearpage
\bibliographystyle{plainurl}
\bibliography{Pivoting-hex}

\begin{thebibliography}{10}

\bibitem{musketeers}
Hugo~A. Akitaya, Esther~M. Arkin, Mirela Damian, Erik~D. Demaine, Vida
  Dujmovi\'cc, Robin Flatland, Matias Korman, Bel\'en Palop, Irene Parada,
  Andr{\'e} van Renssen, and Vera Sacrist{\'a}n.
\newblock Universal reconfiguration of facet-connected modular robots by
  pivots: The {$O(1)$} musketeers.
\newblock In {\em Proc. 27th Annual European Symposium on Algorithms (ESA)},
  volume 144, pages 3:1--3:14, 2019.
\newblock \href {https://doi.org/10.4230/LIPIcs.ESA.2019.3}
  {\path{doi:10.4230/LIPIcs.ESA.2019.3}}.

\bibitem{metamodule1}
Greg Aloupis, Nadia Benbernou, Mirela Damian, Erik~D. Demaine, Robin Flatland,
  John Iacono, and Stefanie Wuhrer.
\newblock Efficient reconfiguration of lattice-based modular robots.
\newblock {\em Computational Geometry: Theory and Applications},
  46(8):917--928, 2013.
\newblock \href {https://doi.org/10.1016/j.comgeo.2013.03.004}
  {\path{doi:10.1016/j.comgeo.2013.03.004}}.

\bibitem{squeezing11}
Greg Aloupis, Sébastien Collette, Mirela Damian, Erik~D. Demaine, Robin
  Flatland, Stefan Langerman, Joseph O'Rourke, Val Pinciu, Suneeta Ramaswami,
  Vera Sacristán, and Stefanie Wuhrer.
\newblock Efficient constant-velocity reconfiguration of crystalline robots.
\newblock {\em Robotica}, 29(1):59--71, 2011.
\newblock \href {https://doi.org/10.1017/S026357471000072X}
  {\path{doi:10.1017/S026357471000072X}}.

\bibitem{EMCube}
Byoung~Kwon An.
\newblock {EM-Cube}: Cube-shaped, self-reconfigurable robots sliding on
  structure surfaces.
\newblock In {\em Proc. 2008 IEEE International Conference on Robotics and
  Automation (ICRA)}, pages 3149--3155, 2008.
\newblock \href {https://doi.org/10.1109/ROBOT.2008.4543690}
  {\path{doi:10.1109/ROBOT.2008.4543690}}.

\bibitem{iogadgets}
Joshua Ani, Erik~D. Demaine, Dylan~H. Hendrickson, and Jayson Lynch.
\newblock Trains, games, and complexity: 0/1/2-player motion planning through
  input/output gadgets.
\newblock {\em CoRR}, abs/2005.03192, 2020.
\newblock \href {http://arxiv.org/abs/2005.03192} {\path{arXiv:2005.03192}}.

\bibitem{heuristics-square}
Nora Ayanian, Paul~J. White, {\'A}d{\'a}m H{\'a}l{\'a}sz, Mark Yim, and Vijay
  Kumar.
\newblock Stochastic control for self-assembly of {XBots}.
\newblock In {\em Proc. ASME International Design Engineering Technical
  Conferences and Computers and Information in Engineering Conference
  (IDETC-CIE)}, 2008.
\newblock \href {https://doi.org/10.1115/DETC2008-49535}
  {\path{doi:10.1115/DETC2008-49535}}.

\bibitem{balanza2019full}
Jose Balanza-Martinez, Austin Luchsinger, David Caballero, Rene Reyes, Angel~A
  Cantu, Robert Schweller, Luis~Angel Garcia, and Tim Wylie.
\newblock Full tilt: Universal constructors for general shapes with uniform
  external forces.
\newblock In {\em Proc. 30th Annual ACM-SIAM Symposium on Discrete Algorithms
  (SODA)}, pages 2689--2708, 2019.

\bibitem{nadia}
Nadia~M. Benbernou.
\newblock {\em Geometric Algorithms for Reconfigurable Structures}.
\newblock PhD thesis, Massachusetts Institute of Technology, 2011.

\bibitem{caballero2020relocating}
David Caballero, Angel~A. Cantu, Timothy Gomez, Austin Luchsinger, Robert
  Schweller, and Tim Wylie.
\newblock Relocating units in robot swarms with uniform control signals is
  {PSPACE}-complete.
\newblock In {\em Proc. 32th Canadian Conference on Computational Geometry},
  2020.

\bibitem{Chiang01}
Chih-Jung Chiang and Gregory~S. Chirikjian.
\newblock Modular robot motion planning using similarity metrics.
\newblock {\em Autonomous Robots}, 10:91--106, 2001.
\newblock \href {https://doi.org/10.1023/A:1026552720914}
  {\path{doi:10.1023/A:1026552720914}}.

\bibitem{Toggles_FUN2018}
Erik~D. Demaine, Isaac Grosof, Jayson Lynch, and Mikhail Rudoy.
\newblock Computational complexity of motion planning of a robot through simple
  gadgets.
\newblock In {\em Proc. 9th International Conference on Fun with Algorithms
  (FUN)}, volume 100, pages 18:1--18:21, 2018.
\newblock \href {https://doi.org/10.4230/LIPIcs.FUN.2018.18}
  {\path{doi:10.4230/LIPIcs.FUN.2018.18}}.

\bibitem{motionplanning2}
Erik~D. Demaine, Dylan~H. Hendrickson, and Jayson Lynch.
\newblock Toward a general complexity theory of motion planning: Characterizing
  which gadgets make games hard.
\newblock In {\em Proc. 11th Innovations in Theoretical Computer Science
  Conference (ITCS)}, volume 151, pages 62:1--62:42, 2020.
\newblock \href {https://doi.org/10.4230/LIPIcs.ITCS.2020.62}
  {\path{doi:10.4230/LIPIcs.ITCS.2020.62}}.

\bibitem{pushing-squares}
Adrian Dumitrescu and J\'anos Pach.
\newblock Pushing squares around.
\newblock {\em Graphs and Combinatorics}, 22(1):37--50, 2006.
\newblock \href {https://doi.org/10.1007/s00373-005-0640-1}
  {\path{doi:10.1007/s00373-005-0640-1}}.

\bibitem{Dumitrescu-Suzuki-Yamashita-2004}
Adrian Dumitrescu, Ichiro Suzuki, and Masafumi Yamashita.
\newblock Motion planning for metamorphic systems: feasibility, decidability,
  and distributed reconfiguration.
\newblock {\em {IEEE} Transactions on Robotics}, 20(3):409--418, 2004.

\bibitem{MeltGrow}
Robert Fitch, Zack Butler, and Daniela Rus.
\newblock Reconfiguration planning for heterogeneous self-reconfiguring robots.
\newblock In {\em Proc. 2003 IEEE/RSJ International Conference on Intelligent
  Robots and Systems (IROS)}, volume~3, pages 2460--2467, 2003.
\newblock \href {https://doi.org/10.1109/IROS.2003.1249239}
  {\path{doi:10.1109/IROS.2003.1249239}}.

\bibitem{MLattice_planning}
Enguang Guan, Zhuang Fu, Weixin Yan, Dongsheng Jiang, and Yanzheng Zhao.
\newblock Self-reconfiguration path planning design for {M-Lattice} robot based
  on genetic algorithm.
\newblock In {\em Proc. 2011 International Conference on Intelligent Robotics
  and Applications (ICIRA)}, volume 7102, pages 505--514, 2011.
\newblock \href {https://doi.org/10.1007/978-3-642-25489-5_49}
  {\path{doi:10.1007/978-3-642-25489-5_49}}.

\bibitem{GPCBook09}
Robert~A. Hearn and Erik~D. Demaine.
\newblock {\em Games, {Puzzles}, and {Computation}}.
\newblock A. K. Peters/CRC Press, 2009.

\bibitem{Vertical98}
Kazuo Hosokawa, Takehito Tsujimori, Teruo Fujii, Hayato Kaetsu, Hajime Asama,
  Yoji Kuroda, and Isao Endo.
\newblock Self-organizing collective robots with morphogenesis in a vertical
  plane.
\newblock In {\em Proc. 1998 IEEE International Conference on Robotics and
  Automation (ICRA)}, volume~4, pages 2858--2863, 1998.
\newblock \href {https://doi.org/10.1109/ROBOT.1998.680616}
  {\path{doi:10.1109/ROBOT.1998.680616}}.

\bibitem{HurtadoMRA15}
Ferran Hurtado, Enrique Molina, Suneeta Ramaswami, and Vera~Sacrist{\'{a}}n
  Adinolfi.
\newblock Distributed reconfiguration of {2D} lattice-based modular robotic
  systems.
\newblock {\em Autonomous Robots}, 38(4):383--413, 2015.
\newblock \href {https://doi.org/10.1007/s10514-015-9421-8}
  {\path{doi:10.1007/s10514-015-9421-8}}.

\bibitem{bkirby-iros07}
Brian Kirby, Burak Aksak, Seth~Copen Goldstein, James~F. Hoburg, Todd~C. Mowry,
  and Padmanabhan Pillai.
\newblock A modular robotic system using magnetic force effectors.
\newblock In {\em Proc. 2007 IEEE International Conference on Intelligent
  Robots and Systems (IROS)}, volume~3, pages 2787--2793, 2007.
\newblock \href {https://doi.org/10.1109/IROS.2007.4399444}
  {\path{doi:10.1109/IROS.2007.4399444}}.

\bibitem{compacting_squares}
Irina Kostitsyna, Irene Parada, Willem Sonke, Bettina Speckmann, and Jules
  Wulms.
\newblock Compacting squares.
\newblock Manuscript, 2020.

\bibitem{heuristics}
Tom Larkworthy and Subramanian Ramamoorthy.
\newblock A characterization of the reconfiguration space of self-reconfiguring
  robotic systems.
\newblock {\em Robotica}, 29(1):73--85, 2011.
\newblock \href {https://doi.org/10.1017/S0263574710000718}
  {\path{doi:10.1017/S0263574710000718}}.

\bibitem{Fractum}
Satoshi Murata, Haruhisa Kurokawa, and Shigeru Kokaji.
\newblock Self-assembling machine.
\newblock In {\em Proc. 1994 IEEE International Conference on Robotics and
  Automation (ICRA)}, volume~1, pages 441--448, 1994.
\newblock \href {https://doi.org/10.1109/ROBOT.1994.351257}
  {\path{doi:10.1109/ROBOT.1994.351257}}.

\bibitem{M-tran}
Satoshi Murata, Eiichi Yoshida, Akiya Kamimura, Haruhisa Kurokawa, Kohji
  Tomita, and Shigeru Kokaji.
\newblock {M-TRAN}: self-reconfigurable modular robotic system.
\newblock {\em IEEE/ASME Transactions on Mechatronics}, 7(4):431--441, 2002.
\newblock \href {https://doi.org/10.1109/TMECH.2002.806220}
  {\path{doi:10.1109/TMECH.2002.806220}}.

\bibitem{density}
An~Nguyen, Leonidas~J. Guibas, and Mark Yim.
\newblock Controlled module density helps reconfiguration planning.
\newblock In {\em Algorithmic and Computational Robotics: New Dimensions (2000
  WAFR)}, pages 23--36. 2001.

\bibitem{metamorphic96}
Amit Pamecha, Chih-Jung Chiang, David Stein, and Gregory Chirikjian.
\newblock Design and implementation of metamorphic robots.
\newblock In {\em Proc. 1996 ASME Design Engineering Technical Conferences and
  Computers in Engineering Conference}, 1996.

\bibitem{metamodule2}
Irene Parada, Vera Sacrist{\'a}n, and {Rodrigo I.} Silveira.
\newblock A new meta-module for efficient reconfiguration of hinged-units
  modular robots.
\newblock In {\em Proc. 2016 IEEE International Conference on Robotics and
  Automation (ICRA)}, pages 5197--5202, 2016.
\newblock \href {https://doi.org/10.1109/ICRA.2016.7487726}
  {\path{doi:10.1109/ICRA.2016.7487726}}.

\bibitem{HexBot09}
Hossein Sadjadi, Omid Mohareri, Mohammad~Amin Al-Jarrah, and Khaled Assaleh.
\newblock Design and implementation of {HexBot}: A modular self-reconfigurable
  robotic system.
\newblock In {\em Proc. 6th International Symposium on Mechatronics and its
  Applications (ISMA)}, pages 1--6, 2009.
\newblock \href {https://doi.org/10.1109/ISMA.2009.5164784}
  {\path{doi:10.1109/ISMA.2009.5164784}}.

\bibitem{SuperBot}
Behnam Salemi, Mark Moll, and Wei-Min Shen.
\newblock {SUPERBOT}: A deployable, multi-functional, and modular
  self-reconfigurable robotic system.
\newblock In {\em Proc. 2006 IEEE/RSJ International Conference on Intelligent
  Robots and Systems (IROS)}, pages 3636--3641, 2006.
\newblock \href {https://doi.org/10.1109/IROS.2006.281719}
  {\path{doi:10.1109/IROS.2006.281719}}.

\bibitem{M-blocks}
Cynthia Sung, James Bern, John Romanishin, and Daniela Rus.
\newblock Reconfiguration planning for pivoting cube modular robots.
\newblock In {\em Proc. 2015 IEEE International Conference on Robotics and
  Automation (ICRA)}, pages 1933--1940, 2015.
\newblock \href {https://doi.org/10.1109/ICRA.2015.7139451}
  {\path{doi:10.1109/ICRA.2015.7139451}}.

\bibitem{ICubes}
Cem {\"{U}}nsal, H.~Han~Kili{\c{c}}{\c{c}}{\"{o}}te, and Pradeep~K. Khosla.
\newblock {I(CES)-Cubes}: A modular self-reconfigurable bipartite robotic
  system.
\newblock In {\em Proc. 1999 SPIE Conference on Mobile Robots and Autonomous
  Systems}, volume 3839, pages 258--269, 1999.
\newblock \href {https://doi.org/10.1117/12.360346}
  {\path{doi:10.1117/12.360346}}.

\bibitem{PolyBot}
Mark Yim, David~G. Duff, and Kimon Roufas.
\newblock {PolyBot}: a modular reconfigurable robot.
\newblock In {\em Proc. 2000 IEEE International Conference on Robotics and
  Automation (ICRA)}, pages 514--520, 2000.
\newblock \href {https://doi.org/10.1109/ROBOT.2000.844106}
  {\path{doi:10.1109/ROBOT.2000.844106}}.

\bibitem{micro00}
Eiichi Yoshida, Shigeru Kokaji, Satoshi Murata, Kohji Tomita, and Haruhisa
  Kurokawa.
\newblock Miniaturization of self-reconfigurable robotic system using shape
  memory alloy actuators.
\newblock {\em Journal of Robotics and Mechatronics}, 12(2):96--102, 2000.
\newblock \href {https://doi.org/10.20965/jrm.2000.p0096}
  {\path{doi:10.20965/jrm.2000.p0096}}.

\end{thebibliography}

\clearpage
\appendix
\magicappendix

\end{document}